\documentclass[%
 reprint,
 amsmath,amssymb,
 aps,pra,
 longbibliography
]{revtex4-2}

\usepackage{mathtools}

\DeclarePairedDelimiter\floor{\lfloor}{\rfloor}
\usepackage[utf8]{inputenc}
\usepackage{enumitem}
\usepackage{qcircuit}
\usepackage{amsmath}
\usepackage{amssymb}
\usepackage{amsthm}
\usepackage{physics}
\usepackage{dcolumn}% Align table columns on decimal point
\usepackage{bm}% bold math
\newtheorem{theorem}{Theorem}

\newtheorem{lemma}[theorem]{Lemma}
\usepackage[justification=centerlast]{caption}

\newtheorem{definition}[theorem]{Definition}

\newcommand{\qeds}{\nobreak \ifvmode \relax \else
      \ifdim\lastskip<1.5em \hskip-\lastskip
      \hskip1.5em plus0em minus0.5em \fi \nobreak
      \vrule height0.75em width0.5em depth0.25em\fi}

\usepackage{color, soul}

\usepackage{graphicx}
\usepackage{subcaption}
\usepackage{sidecap}
\usepackage{dsfont}
\usepackage{color}
\usepackage[hidelinks]{hyperref}
\usepackage[normalem]{ulem}
\usepackage{enumitem}
\usepackage[font=footnotesize]{caption}
\usepackage{enumitem}

\definecolor{JM}{rgb}{0.3,0.31,0.7}
\definecolor{AS}{rgb}{0.1,1,0.1}
\definecolor{purple}{RGB}{160,32,240}

\newcommand{\subsubsubsection}[1]{\paragraph{#1}\mbox{}\\}
\begin{document}

\preprint{APS/123-QED}

\title{\texorpdfstring{Logical accreditation: a framework for efficient certification\\ of fault-tolerant computations}{Logical accreditation: a framework for efficient certification of fault-tolerant computations}}

\author{James Mills$^{1}$}
\email{J.Mills-7@sms.ed.ac.uk}
\author{Adithya Sireesh$^{1}$}
\author{Dominik Leichtle$^{1}$}
\author{Joschka Roffe$^{1}$} 
\author{Elham Kashefi$^{1,2}$}

 \affiliation{$^1$School of Informatics, University of Edinburgh, United Kingdom \\$^2$Laboratoire d’Informatique de Paris 6, Sorbonne Université, France}

\begin{abstract}

As fault-tolerant quantum computers scale, certifying the accuracy of computations performed with encoded logical qubits will soon become classically intractable. 
This creates a critical need for scalable, device-independent certification methods. 
In this work, we introduce logical accreditation, a framework for efficiently certifying the correctness of quantum computations performed on logical qubits. 
Our protocol is robust against general noise models, far beyond those typically considered in performance analyses of quantum error-correcting codes. 
Through numerical simulations, we demonstrate that logical accreditation can scalably certify quantum advantage experiments and indicate the crossover point where encoded computations begin to outperform physical computations. 
The framework also enables evaluation of whether logical error rates are sufficiently low that error mitigation can be efficiently performed, extends entropy benchmarking to the regime of fault-tolerant computation, and upper bounds the infidelity of the logical output state of a computation.
Underlying the framework is a novel randomised compilation scheme that converts arbitrary logical circuit noise into stochastic Pauli noise.
This scheme includes a method for twirling non-transversal logical gates beyond the standard $T$ gate, resolving an open problem posed by [Piveteau et al. PRL 127, 200505 (2021)]. 
By bridging fault-tolerant computation and computational certification, logical accreditation offers a scalable, practical means of certifying the accuracy of quantum computations performed using encoded logical qubits.

\end{abstract}

\maketitle

\section{Introduction}

\begin{figure*}

\begin{subfigure}[t]{0.25\textwidth}
\centering
\includegraphics[width=0.99\textwidth]
{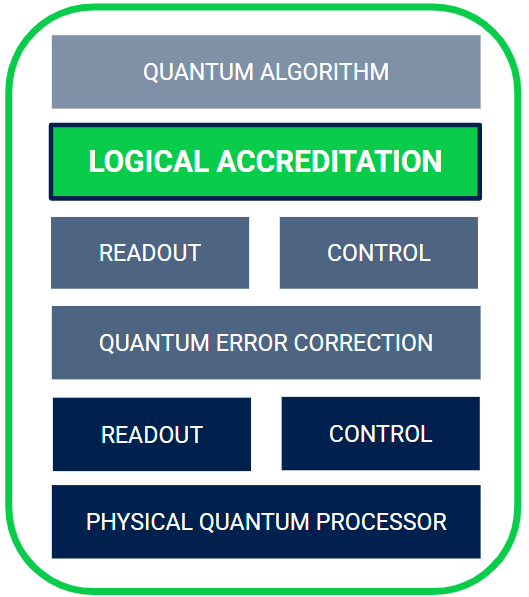}
\caption{}
\end{subfigure}
\begin{subfigure}[t]{0.73\textwidth}
\centering
\hspace{3.2em}\includegraphics[width=0.89\textwidth]
{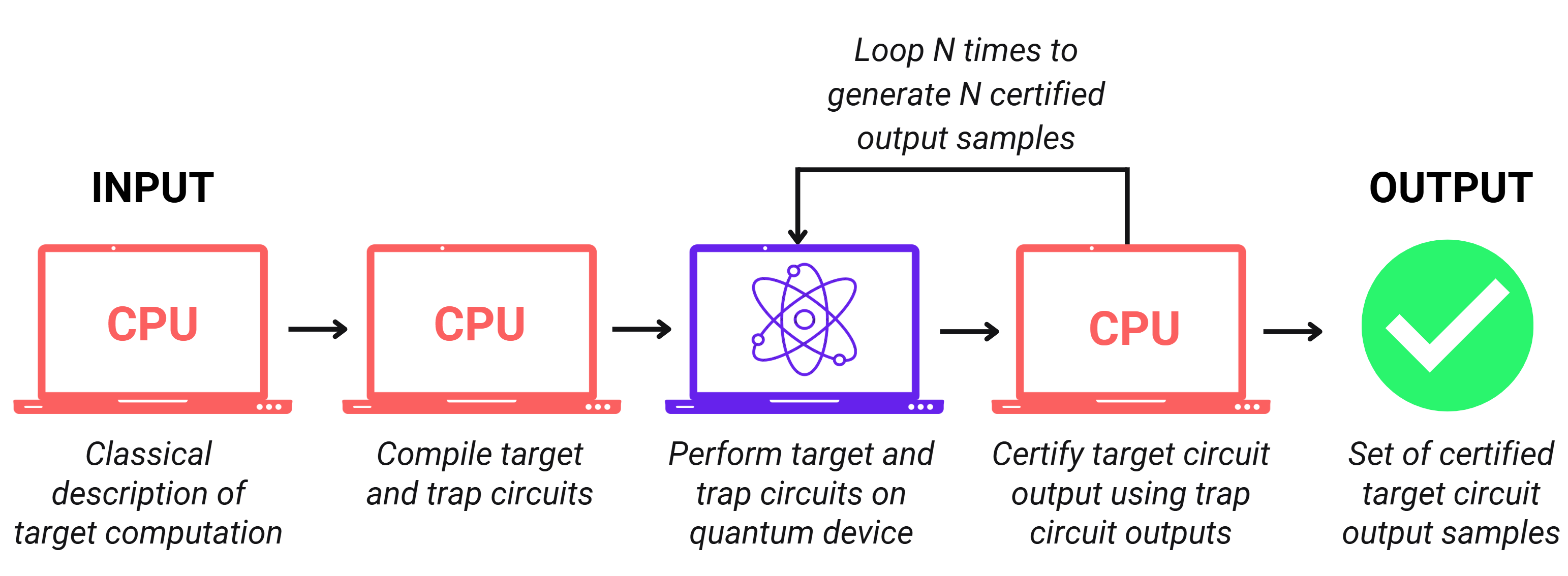}
\caption{}
\end{subfigure}
\caption{
(a) The fault-tolerant quantum computing stack with the logical accreditation layer positioned between the logical control and the quantum algorithm layers.
(b) Logical accreditation is performed with a quantum device and a classical processing unit (CPU). 
The input is a classical description of the target computation. 
This is used to compile both the target and trap circuits according to the required circuit structure and gate set. 
These circuits are run on the quantum device in a randomised order, and the measured bit strings are recorded on the CPU. 
The trap circuit outputs are then post-processed on the CPU to certify the target circuit output. 
To obtain $N$ certified samples, the sampling procedure is repeated $N$ times. 
The output is a certified set of $N$ bit string samples from the target circuit.
}
\label{fig:simple_schematic_diagram}
\end{figure*}

Experimental realisations of quantum error correction have reached a turning point \cite{acharya_quantum_2025, bluvstein_logical_2024, paetznick_demonstration_2024}. 
With the ongoing improvement of physical qubit quality and system scale, we are beginning to enter the regime where quantum computations can be carried out using encoded logical qubits. 
However, we are still far from realising the large-scale, fully fault-tolerant architectures required for universal computation using high-distance codes and complete fault-tolerant gate sets. For the foreseeable future, we will be operating in an intermediate regime of \textit{early fault tolerance}, where resource constraints limit both the error-correcting codes and the fault-tolerant protocols that can be implemented \cite{katabarwa_early_2024, suzuki_quantum_2022}. 
In this regime, it can be unclear whether encoding computations in logical qubits actually improves performance over direct computation with physical qubits.
%, given that magic state purification and logical noise suppression at scale are not yet feasible. 
This uncertainty poses a critical challenge for quantum error correction development: how can we assess the correctness of fault-tolerant computations under realistic conditions?

This challenge is not restricted to the near term. 
Even once it is possible to build large-scale fault-tolerant quantum devices, many assumptions underpinning quantum error correction, e.g., local, stochastic, or Markovian noise models, will remain problematic.
How can we be confident that these assumptions hold in practice, or that logical-level computations remain correct under deviations from particular noise models?
For example, correlated errors were observed in recent surface code experiments in superconducting qubit hardware, revealing a logical error floor that prevents further error suppression \cite{acharya_quantum_2025}.
Such factors limit the ability of classical simulations and analytical models to certify computational correctness at scale, motivating the need for alternative, device-independent methods of certification. 
A new layer in the fault-tolerant quantum computing stack is needed: a tool for certifying that encoded quantum computations behave as intended, even for general, possibly highly correlated, noise.

This paper introduces such a tool by bridging the fields of quantum error correction and computational certification. 
Specifically, we introduce \textit{logical accreditation} as a scalable framework for efficient certification of computations performed with logical qubits encoded using stabiliser codes.
Our protocol brings techniques from the field of certification (also known as verification or accreditation) \cite{gheorghiu_verification_2019}, and extends them to the fault-tolerant regime by applying them to logical circuits implemented using stabiliser codes. 
The sampling overhead of logical accreditation is independent of both the number of logical qubits and the circuit depth.
This inherent scalability means that the framework can be applied to certify computational correctness in regimes that exceed the capabilities of the most advanced classical simulators.

Fig. \ref{fig:simple_schematic_diagram} (a) shows the positioning of the certification framework within the fault-tolerant quantum computing stack, while Fig. \ref{fig:simple_schematic_diagram} (b) provides a schematic overview of the framework.
In logical accreditation, the target computation is embedded within an ensemble of \textit{trap} computations. 
The trap computations are designed to be structurally identical to the target computation but are configured to have deterministic outputs in the absence of noise. 
This deterministic trap circuit behaviour can be leveraged to provide a rigorous upper bound on the total variational distance (TVD) between the noisy and ideal output distributions of the target computation.

A key advantage of logical accreditation is its ability to certify the accuracy of a \textit{specific} computation performed on quantum hardware. 
This distinguishes it from methods that measure average performance, like the gate fidelity information provided by randomised benchmarking \cite{knill_randomized_2008, magesan_scalable_2011, helsen_general_2022} or the holistic device performance measured by quantum volume \cite{cross_validating_2019, jurcevic_demonstration_2021, pelofske_quantum_2022}. 
Crucially, logical accreditation provides these guarantees under far more general noise assumptions than are typically considered in the analysis of quantum error correction protocols.

One technical contribution of this work is a compilation scheme, inspired by randomised compiling \cite{wallman_noise_2016, hashim_randomized_2021}, that transforms general logical errors into stochastic Pauli noise. 
Within this scheme, we introduce a novel method that
%for twirling arbitrarily-rotated magic states, which 
resolves the open question of how to twirl non-transversal logical gates beyond the $T$ gate \cite{piveteau_error_2021}. 

This work builds on ideas from the family of protocols known as accreditation, originally
developed to certify circuits executed on noisy physical qubits in both the digital and analog settings \cite{ferracin_accrediting_2019,ferracin_experimental_2021, jackson_accreditation_2024, jackson_improved_2025}, as well as from cryptographic verification protocols \cite{fitzsimons_unconditionally_2017, gheorghiu_verification_2019, leichtle_verifying_2021}.
Another methodology closely related to accreditation is mirror circuit fidelity estimation, which estimates circuit fidelity using randomised mirror circuits \cite{proctor_establishing_2022}. 
This provides a scalable approach to estimating average circuit fidelity over many runs and has primarily been developed for unencoded physical quantum circuits, although it may be extendable to logical circuits. 
In contrast, logical accreditation is designed for encoded logical circuits and enables a probabilistic guarantee on the correctness of a specific run of the target circuit, where the target circuit execution position is uniformly randomised as part of the protocol. 
The two approaches therefore address complementary aspects of establishing trust in noisy quantum computations.

We demonstrate the utility of the logical accreditation framework through numerical simulation of instantaneous quantum polynomial-time (IQP) and Trotterised circuits, also using noisy intermediate-scale quantum (NISQ) accreditation to certify unencoded physical computations \cite{ferracin_experimental_2021}. 
This enables a direct comparison that indicates the crossover point where fault-tolerance delivers a practical advantage over unencoded quantum computation. 
The framework can also be used to extend entropy density benchmarking \cite{stilck_franca_limitations_2021,demarty_entropy_2024} to the fault-tolerant regime, assess whether quantum error mitigation techniques can be efficiently applied to specific logical circuits, and provide an upper bound on the infidelity of the output state of a logical computation.

\vspace{-1.3em}
\section{Preliminaries}

\subsection{Computation with logical qubits encoded using stabiliser codes}

An $[[n,k, d]]$ stabiliser code encodes $k$ logical qubits using $n$ physical qubits, with a code distance $d$, where the code distance is the minimum weight with which any non-trivial logical operator acts. 
A stabiliser group $\mathcal{S}$ is a commutative subgroup of the Pauli group not containing the operator $-I^{\otimes n}$, and its elements form an Abelian subgroup of the Pauli group.
The stabiliser group has a minimal representation in terms of a set of independent generator elements of size $n-k$ i.e. $\mathcal{G}=\{g_1, \ldots,g_{n-k}\}$.
The simultaneous $+1$ eigenspace of a given stabiliser group defines the code-space of a QEC code.
The code-space is then a $2^k$-dimensional subspace of the larger $2^n$-dimensional Hilbert space, and its orthogonal complement is the error-space.
A logical basis of the form $\{\ket{0}_L,\ket{1}_L\}^{\otimes k}$ can be defined in the code-space, along with logical Pauli operators that act on the encoded qubits.

If a state, $\ket{\psi}_L$, is within the code-space of the code, then $P_i\ket{\psi}_L = \ket{\psi}_L $ if $P_i \in \mathcal{S}$.
During cycles of error correction, $n-k$ stabiliser measurements are made, using the elements of the generator group $\mathcal{G}$.
Depending on whether errors have occurred, each stabiliser measurement projects the logical state onto either the $+1$ or $-1$ eigenspace of the measurement operator, where the measurement projector is $P_{g_i, s} = \frac{1}{2}(I + (-1)^s g_i)$ for $g_i \in \mathcal{G}$ and syndrome measurement output $s \in \{0,1\}$.
If the syndrome measurement outputs of all the stabiliser generators are $0$, this indicates the state is within the code-space.
If any syndrome measurements output $1$ this indicates the state is in the error-space.
Repeated rounds of syndrome measurements, with corrections either applied or tracked as errors occur, can be used to protect quantum information stored in the logical code-space for use as a quantum memory.

\begin{figure*}[!]
    \centering
    \begin{subfigure}[t]{0.47\textwidth}
        \centering
        \includegraphics[width=\textwidth]{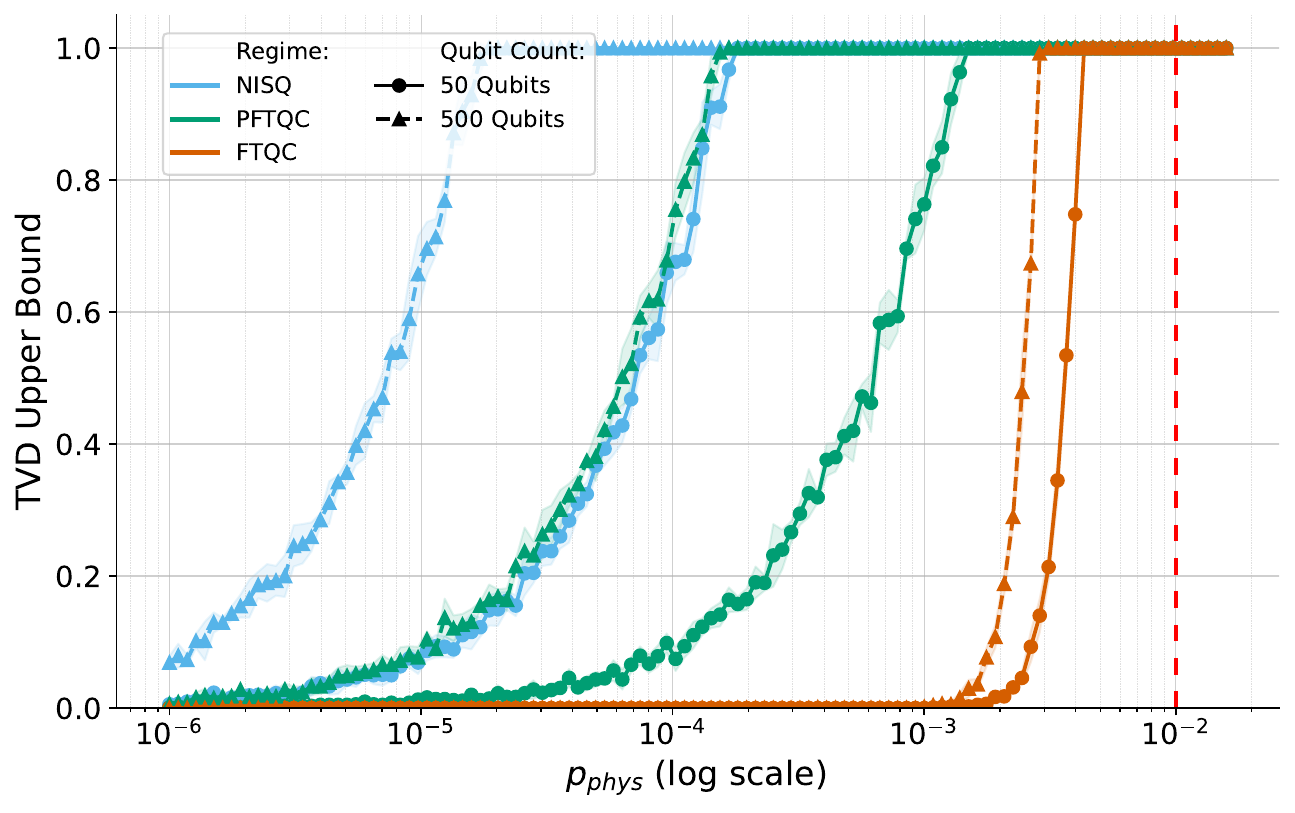}
        \caption{IQP circuits}
        \label{fig:iqp_fixed_layers}
    \end{subfigure}
    \begin{subfigure}[t]{0.47\textwidth}
        \centering
        \includegraphics[width=\textwidth]{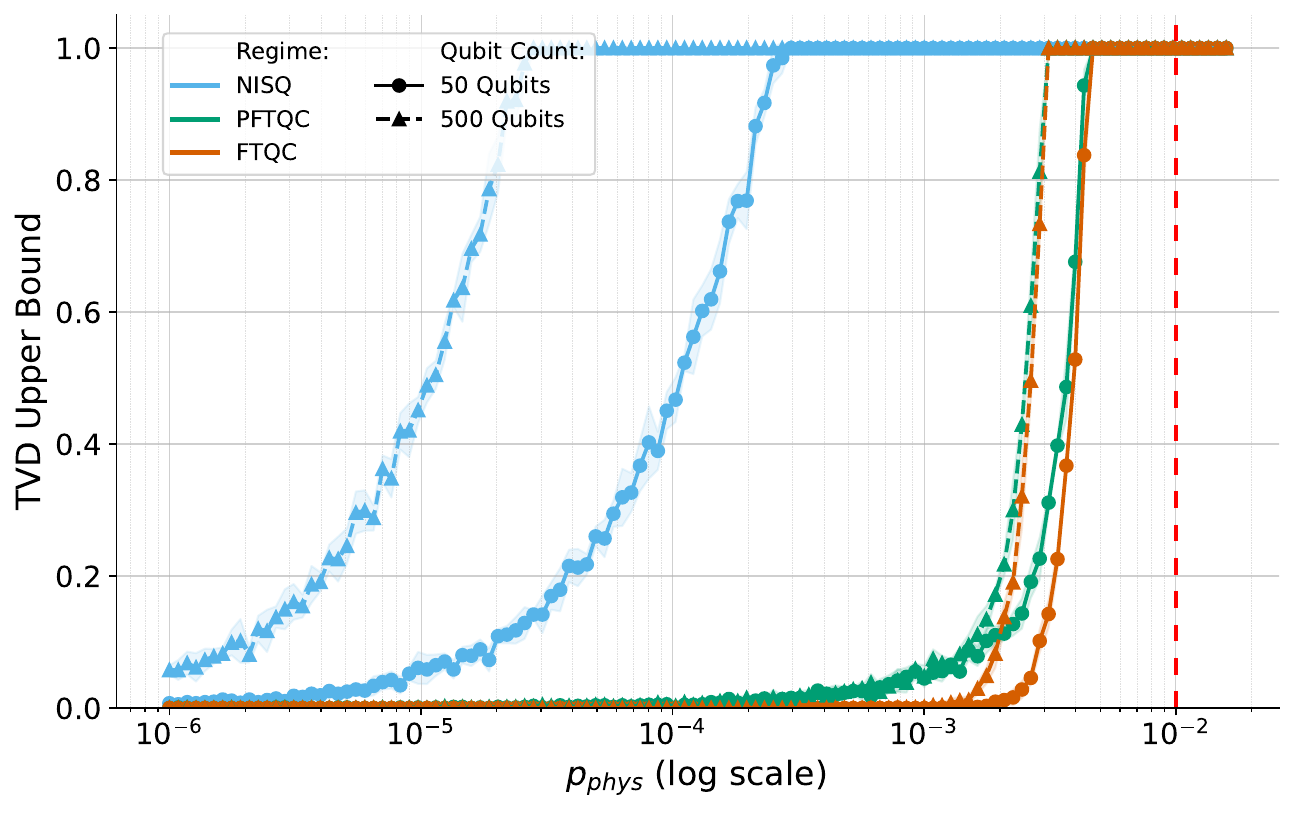}
        \caption{Trotterised circuits}
        \label{fig:chem_fixed_layers}
    \end{subfigure}   
    \begin{subfigure}[t]{0.47\textwidth}
        \centering
        \includegraphics[width=\textwidth]{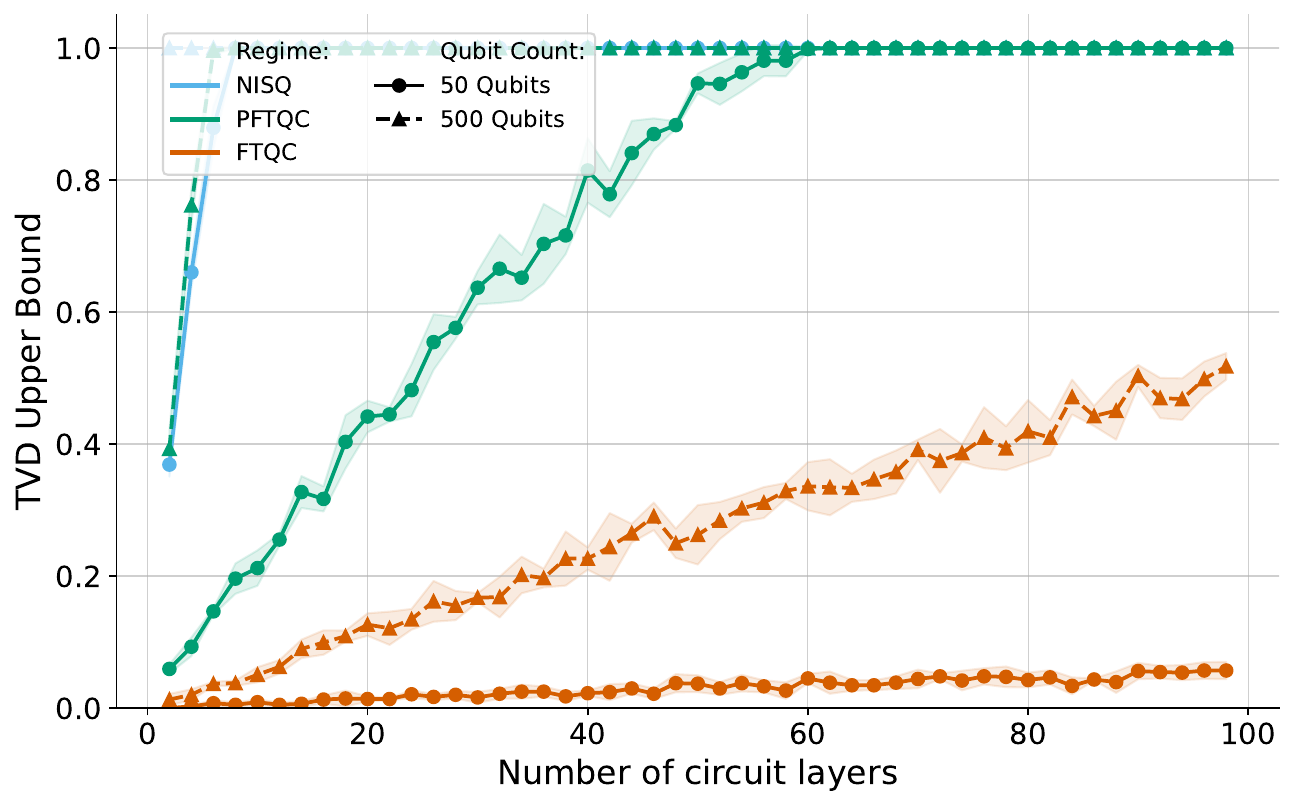}
        \caption{IQP circuits}
        \label{fig:iqp_varying_layers}
    \end{subfigure}
    \begin{subfigure}[t]{0.47\textwidth}
        \centering
        \includegraphics[width=\textwidth]{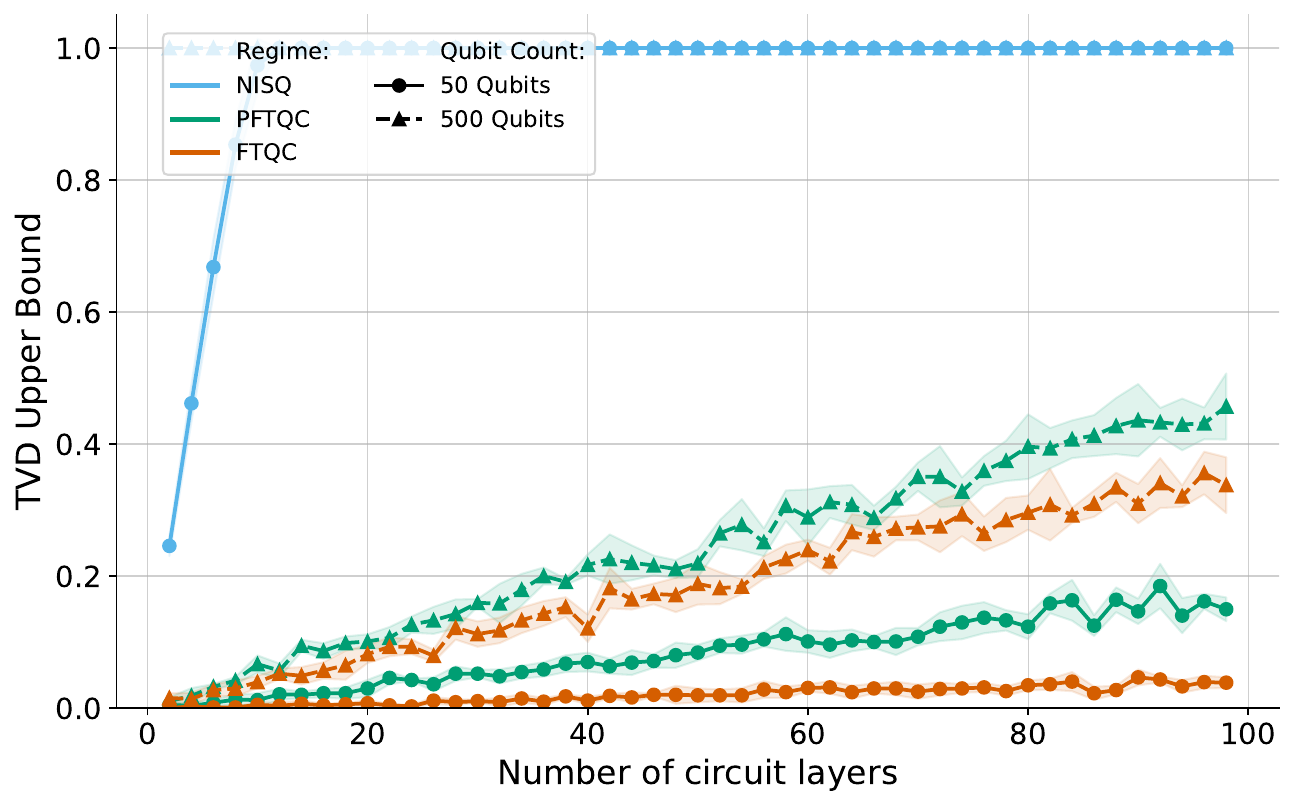}
        \caption{Trotterised circuits}
        \label{fig:chem_varying_layers}
\end{subfigure}   
    \caption{\emph{Numerical simulation results of logical accreditation applied to IQP and Trotterised circuits.} 
    Circuits were run in three regimes: (i) noisy physical qubits (labelled NISQ), (ii) logical qubits without magic state purification (labelled PFTQC), and (iii) logical qubits with magic state purification (labelled FTQC).
    Logical error rates were modelled using surface code parameters with code distance $d=11$. 
    Each protocol run used $500$ trap circuits.
    The method from \cite{ferracin_experimental_2021} was applied to obtain the TVD bounds for NISQ circuits, while for PFTQC and FTQC circuits logical accreditation was instead used. 
    In (a) and (b), TVD upper bounds are plotted against physical error rate for IQP and Trotterised circuits respectively; the number of circuit gate layers was fixed at 40.
    In (c) and (d), TVD upper bounds are plotted against the number of circuit gate layers again for each circuit type; the physical error rate was fixed at $p_{\text{phys}} = 10^{-3}$.}
    \label{fig:iqp_chem_combined}
\end{figure*}

To perform computation with logical states, it is necessary not only to reliably store quantum information in the code-space but also to apply logical gate operations. 
For a given quantum error-correcting code, a subset of logical gates known as transversal gates can be implemented fault-tolerantly without incurring additional qubit overhead. 
The Eastin-Knill theorem states that no quantum error-correcting code has a set of transversal gates with which universal computation can be performed \cite{eastin_restrictions_2009}.

In the logical accreditation framework, we assume that qubits are encoded with a stabiliser code where Clifford operations may be performed transversally, or using alternative fault-tolerant methods like lattice surgery \cite{horsman_surface_2012}.
Non-Clifford gates, required for universal computation, are performed through the consumption of magic states using gate teleportation or projective measurement.
Magic states are quantum states prepared using ancillary logical qubits, and in this work, are assumed to be of the form
\begin{equation}
\begin{split}
\ket{\theta} &:= \frac{1}{\sqrt{2}}\big(\ket{0} + e^{i \theta}\ket{1}\big).\\
\end{split}
\end{equation}
If a magic state is phase rotated by angle $\theta = m \cdot \pi/4$ for $m \in \mathbb{Z}$, then a single round of gate teleportation may be used to apply the desired gate.
If $\theta \neq m \cdot \pi/4$, a repeat-until-success (RUS) approach may instead be used \cite{akahoshi_partially_2024,akahoshi_compilation_2025}.
In RUS gates, the correct rotation by angle $\theta$ is applied with probability 1/2.
If the incorrect rotation is applied, a new magic state rotated by $2\theta$ is prepared and the gate is retried; this process is repeated until the intended rotation is achieved.
The average number of repetitions for RUS gate success is
\begin{equation} \label{RUS_gate_sum}
  1\times \frac{1}{2}+2 \times \frac{1}{4} + \ldots=  \sum_{i=1}^{\infty} \frac{i}{2^i} = 2.
\end{equation} 
A projective Pauli measurement approach can be applied to use magic states to perform high-weight logical Pauli rotations. 
Projective measurements may be performed fault-tolerantly using lattice surgery, and the number of high-weight logical rotation operations can be optimised during circuit compilation \cite{litinski_game_2019}.

\subsection{Fault-tolerant quantum computation} \label{model_logical_noise_section}

A logical qubit encoded using an $[[n, k, d]]$ stabiliser code can tolerate up to $\floor{(d-1)/2}$ physical errors. 
If this bound is exceeded, the decoder may apply an incorrect correction, resulting in a logical error. 
Logical errors are unwanted operations that alter the encoded state without taking it out of the code-space. 
Since they do not produce a detectable syndrome, they are not corrected during rounds of error correction.

The space-time constraints of early fault-tolerant quantum devices will necessitate trade-offs that balance error suppression with computation depth.
These constraints will mean that logical error rates cannot be arbitrarily reduced as the scale of computations increases.
For example, decreasing the code distance and degree of magic state purification will allow larger computations to be performed, but at the cost of increasing logical error rates.
Lower code distances decrease the number of qubits that need to be affected by errors before a logical error can occur, while reducing the degree of magic state purification increases the noise introduced into the logical computation each time a non-transversal non-Clifford gate is performed.
Computational frameworks suitable for early fault-tolerant devices have been proposed in which magic states are not purified, but are prepared and immediately used to perform noisy logical non-Clifford operations \cite{piveteau_error_2021, akahoshi_partially_2024}.

We will distinguish between two regimes of logical computation: \textit{partial fault-tolerance} and \textit{full fault-tolerance}.
Following the convention used in \cite{toshio_practical_2025} and \cite{ akahoshi_partially_2024}, we will use the term partial fault-tolerance to mean logical computation performed without purification of magic states, 
while by full fault-tolerance, we refer to logical computation where magic state purification is implemented so that logical non-Clifford gates are performed with approximately the same error rates as logical Clifford gates.
In Appendix \ref{frameworks_logical_computation}, we describe the Clifford plus $T$ framework for fault-tolerant computation and the Clifford plus noisy $T$ framework for partially fault-tolerant computation. 
We also describe an alternative method for partially fault-tolerant computation that uses a gate set composed of Clifford gates and noisy analog rotation gates.

\vspace{-0.7em}
\section{Logical Accreditation Framework}

\subsection{Overview}

Logical accreditation is a framework for certifying fault-tolerant computations, where logical qubits are encoded using quantum stabiliser codes.
The input to the logical accreditation framework is a classical description of the computation to be run on the quantum device.
In the first step of the protocol, the input computation is compiled into a logical circuit, termed the \textit{target computation}.
The target circuit is compiled using a repeating structure of gate layers containing single- or multi-qubit non-Clifford gates and C$Z$ gates, sandwiched between gate layers containing single-qubit Clifford gates.
Then $M$ distinct logical \textit{trap circuits} are compiled. 
These circuits are designed to mirror the structure of the target circuit but produce a known deterministic bit string output in the absence of noise.
The logical target and trap circuit structures are shown in Fig. \ref{fig:target_and_trap_circuit_structure} (a) and (b), respectively.

Within the framework, a circuit compilation scheme termed logical randomised compiling is introduced to twirl all logical circuit noise into logical stochastic Pauli noise.
If one of the trap circuits is affected by noise, its construction guarantees that, conditional on an error, the ideal output is not measured with a probability lower bounded by a constant, as described in Appendices \ref{trap_detect_errors} and \ref{bounds_on_prob_of_error_canc}.

During logical accreditation, the target and trap circuits are run in a random order on the quantum device using logical qubits encoded using a stabiliser code with the chosen logical gate set.
The measurement outcomes of the trap circuits are then used to certify the accuracy of the target circuit output.
Specifically, the accuracy of the target circuit output is quantified in terms of an upper bound $\gamma$ on the TVD between the experimental target circuit output distribution, $\mathcal{D}_{exp}$, and its ideal output distribution, $\mathcal{D}_{ideal}$.

\vspace{-1em}
\subsection{Logical circuit noise} \label{log_circ_noise_sect}
\vspace{-0.6em}

We will use the term \textit{logical noise} to refer to the effective noise acting on logical qubits after a full error correction cycle (including syndrome extraction and decoding), arising from physical errors that are not successfully corrected.
The logical accreditation protocol makes the following assumptions on the structure of the noise affecting the quantum computation:

\begin{enumerate}[label=(A\arabic*)]
    \item Physical single-qubit gate noise within a single circuit run is gate-independent, but can depend on gate positioning within the circuit.
    \item Logical noise affecting single-qubit logical Clifford gates
    is gate-independent, but can depend on the gate positioning within the circuit.
    \item For each logical circuit run, the joint evolution of the
logical qubits and any environment is described by a completely
positive trace-preserving (CPTP) map, and the noise affecting
distinct logical circuit runs is independent.
\end{enumerate}
In the fully fault-tolerant setting, where magic states are prepared fault-tolerantly, an additional assumption is required for the trap circuit construction:
\begin{enumerate}[label=(A\arabic*), resume]
    \item The overall operation used to combine two target magic states into a single trap magic state is non-noise-decreasing.
    That is, the error rate of the output trap magic state is no smaller than the combined error rates of the two input target magic states.
\end{enumerate}
The above assumptions are mild in the sense that they allow for very general noise to affect the logical circuits during the protocol.
This can be noise that is non-local and highly correlated, including the most common types of quantum device noise such as cross-talk effects, qubit decoherence and amplitude damping, and gate calibration errors.
The robustness of the method to violation of these assumptions is discussed in Section \ref{robustness_sect}.

The differences between the target and trap circuits are limited to modifications of physical and logical single-qubit gates. 
Provided these assumptions are valid, the target and trap circuits are vulnerable to the same types of noise, despite the differences in gate types within the circuits.
Assumption A1 may be physically justified, as single-qubit gates generally have the lowest error rates compared to other quantum operations, making variation between their noise profiles marginal.
If single-qubit logical Clifford gates are performed transversally using only physical single-qubit gates, then Assumption A2 follows as a consequence of Assumption A1.
It is true, however, that for most quantum codes arbitrary single-qubit Clifford gates cannot be implemented transversally, and for these Assumption A2 should be interpreted as an approximate modelling assumption, analogous to the gate-independent noise assumptions commonly used in randomised benchmarking \cite{knill_randomized_2008}. 
This is most physically justifiable when logical Clifford gates are implemented using similar fault-tolerant primitives, such that the dominant noise mechanisms are largely shared across gates. 
The protocol is robust to weak violations of this assumption, as shown in Theorem \ref{robustness_theorem}, so moderate gate-dependent variations in logical noise do not significantly affect the resulting accreditation bound. 
In practice, the validity of this approximation could be assessed experimentally before implementing logical accreditation on a quantum device by characterising and comparing logical error rates across different gates. 
Recent demonstrations of encoded operations suggest that gates implemented using similar fault-tolerant constructions can exhibit comparable logical error rates \cite{erhard_entangling_2021, zhang_demonstrating_2025}.
These works do not establish the gate-independence of logical Clifford noise for a given code, but they do indicate that logical operations implemented using similar fault-tolerant constructions can exhibit logical error rates of comparable magnitude.
It is likely that, as the quality of quantum hardware improves, the logical error rates of these operations will decrease, and any differences between distinct logical operations will become smaller.

The independence assumption of A3 is used to guarantee the convergence of the error bound estimator.
This might be relaxed if some other type of convergence guarantee is provided.
Whether the noise is assumed to be Markovian or non-Markovian determines which soundness bound is used for the protocol; this is discussed in Section \ref{sec:comp_and_sound}.
Assumption A3 does not require the noise to be stationary across circuit runs, but rather that the outputs of distinct runs can be modelled as independent random variables. 
This allows for time-dependent noise, such as drift, provided it does not introduce correlations between circuit runs. 
The randomisation of target and trap circuit ordering ensures that the target and trap circuit performance reflects average device behaviour.
Provided that the circuit outputs can be modelled as independent (though not necessarily identically distributed) random variables, the estimator produced by the protocol remains well-defined and corresponds to an average over the associated time-dependent noise instances.
Under these conditions, standard concentration bounds such as Hoeffding's inequality apply. 
Temporal correlations between runs would violate this assumption and affect the convergence guarantees of the protocol.

Approximations made during circuit compilation can be another source of computational error.
Approximation errors can occur if, for example, the Solovay-Kitaev algorithm is used to approximate arbitrary unitary gates to within an operator norm distance of $\epsilon>0$ from the ideal unitary, with a gate sequence of length $O(\log^c(1/\epsilon))$, for a constant $c>1,$ using only a discrete gate set \cite{kitaev_quantum_1997, kitaev_classical_2002}.
Compilation errors will also occur when approximation methods such as Trotterisation are used to simulate the evolution of a Hamiltonian with time \cite{suzuki_general_1991, berry_efficient_2007, poulin_trotter_2015}.
Although the logical accreditation framework is compatible with such methods, it cannot detect errors caused by compilation approximations.
However, bounds on approximation errors of this kind may often be efficiently computed using classical methods, allowing them to be straightforwardly accounted for.

\vspace{-0.7em}
\subsection{Certification framework}

We now describe the main result of this work.
In the logical accreditation framework, a target circuit and $M$ trap circuits are performed with a random ordering on a quantum device using encoded logical qubits.
During the protocol, logical noise twirling is applied to the target and trap circuits using a random compilation scheme, as detailed in Section \ref{Log_rand_com_section}.
This compilation is designed to convert general logical circuit noise, such as that arising from incorrect decoding or from imperfect magic state preparation, into logical stochastic Pauli noise.
For each possible circuit position $i$, the output state that
would be generated if the target circuit were executed in that
position may be decomposed as the convex combination
\begin{equation}
    \label{eqaccreditation}
    \rho_{out}^{(i)}=(1-p^{(i)}_{err})\rho_{out,id}^{(i)} + p_{err}^{(i)} \rho_{noisy}^{(i)},
    \end{equation}
where the index $i\in \{1,...,M+1\}$ indicates the circuit position, $p_{err}^{(i)}$ is the logical circuit error rate, $\rho_{out,id}^{(i)}$ is the ideal output state in the absence of noise, and $\rho_{noisy}^{(i)}$ is the output state encompassing the effects of logical circuit noise.

The trap circuits are designed to deterministically output a known bit string in the absence of noise.
If a trap circuit is run on the quantum device and the measured output bit string differs from the ideal output, this is recorded as a trap circuit failure.
The logical noise information derived from the trap circuit failures is used to upper bound the TVD of the measured output of the target circuit from the ideal output.

The framework requires a classical description of the target computation as input.
In the first step of logical accreditation, the target computation is compiled into a logical circuit using a gate set that can include the following logical gate operations: single-qubit Clifford gates, C$Z$ gates, and Pauli operator rotation gates of arbitrary weight and rotation angle.
Non-Clifford gates are performed by consuming magic states, using either gate teleportation or a projective measurement scheme.
A repeating circuit structure is used for the circuit compilation, consisting of layers of single-qubit Clifford gates and layers of multi-qubit Clifford gates and non-Clifford single- or multi-qubit Pauli rotation gates.
Each logical qubit within one logical gate layer is acted on by at most one logical gate.
Next, $M$ randomly constructed trap circuits are generated. 
These circuits replicate the structure of the target circuit, but with all non-Clifford logical circuit components replaced by Clifford operations. 
Each state injection routine, such as those used to implement $T$ gates, is modified to inject a logical $S$ state. 
These substitutions are designed to preserve the logical noise profile of the target circuit while ensuring that each trap circuit yields a deterministic output. 
The target and trap circuit structures are shown in Fig. \ref{fig:target_and_trap_circuit_structure} (a) and (b), respectively, and the trap circuit construction is discussed in Section \ref{sec:construction-trap-computations}.

The target and trap circuits are run on the quantum device with a random ordering, and their measured bit string outputs are recorded. 
The measured output bit strings of the trap circuits are used to compute an upper bound on the TVD between the experimental and ideal target circuit output distributions.

The output of logical accreditation is summarised by the following statement:
\vspace{0.1em}
\vspace{3em}\begin{theorem} \label{main_theorem}
The logical accreditation protocol provides an upper bound on the total variational distance (TVD) between the experimental target circuit output distribution and the ideal output distribution of the form
\begin{equation}\frac{1}{2}\sum_{s \in \{0,1\}^n} |p_{exp}(s) - p_{ideal}(s)| \leq \gamma, % 2 p_{inc},
\end{equation}
where $\gamma$ is experimentally estimated from the measured outputs of the logical trap computations.
The upper bound holds with $(1-\alpha)$-confidence when the number of trap circuits used for the protocol satisfies the inequality $M\geq\frac{1}{2\epsilon^2}\log(\frac{2}{\alpha})$.
The TVD bound holds under the previously stated assumptions \textnormal{A1-A3}.
\end{theorem}

The derivation of this bound can be found in Appendix \ref{computing_bound_app}. 
Logical accreditation can be applied to compute both distributions and expectation values with certified error bounds.
The protocol can be repeated an arbitrary number of times to generate a certified set of output bit strings for the target computation, as shown in Fig. \ref{fig:simple_schematic_diagram} (b).
Alternatively, the framework could be modified to certify multiple target circuit samples at each execution of logical accreditation, rather than just one.
Logical accreditation also provides a practical means to certify
observable estimation. For any Hermitian observable $O$, the
expectation-value error satisfies
\begin{equation}
\left|
    \langle O\rangle_{\rho_{\mathrm{exp}}}
    -
    \langle O\rangle_{\rho_{\mathrm{id}}}
\right|
\leq
2\gamma\|O\|_{\infty}
\end{equation}
with confidence at least $1-\alpha$, where
$\rho_{\mathrm{exp}}$ is the experimental target output state
averaged over the uniformly random target execution position,
$\rho_{\mathrm{id}}$ is the ideal target output state, and
$\|O\|_{\infty}$ denotes the operator norm, also called the
spectral norm.
Furthermore, it can be used to upper-bound the infidelity of the logical output state of the target circuit, as:
\begin{equation} \label{main_text_infidelity_bound}
1 - F(\rho_{out,id}, \rho_{out}) \leq \gamma,
\end{equation}
where $F(\cdot,\cdot)$ denotes the fidelity between two quantum states, and $\rho_{out,id}$ and $\rho_{out}$ are the ideal and experimentally generated target circuit output states; this result is derived in Appendix \ref{sec:upper-bound-infidelity}, along with numerical analysis plotted in Fig. \ref{fig:tvd-tcr-infidelity} comparing the eqn. \ref{main_text_infidelity_bound} bound with the infidelity of states generated using IQP circuits.

\vspace{-0.7em}
\subsection{Robustness} \label{robustness_sect}
\vspace{-0.5em}

If the previously stated noise assumptions are violated, the bound provided by logical accreditation is then of the form: $\gamma' = \gamma + \epsilon_d$, where $\gamma$ is the upper bound provided by the protocol if the assumptions hold,  $\gamma'$ is the upper bound if the assumptions are violated, and $\epsilon_d$ is the difference between these two induced by violation of the assumptions.

The following result is derived in Appendix \ref{robustness_derivation}:
\vspace{2em} \begin{theorem} \label{robustness_theorem}
    If Assumptions A1 and A2 are violated, the absolute difference in the computed upper bound, $\gamma'$, from the upper bound computed where the assumptions hold, $\gamma$, is bounded
     \begin{equation} \label{robustness_bound}
 | \gamma - \gamma'| \leq [M(1-\beta)]^{-1}\sum_k\sum_j || \mathcal{E}_{j}^{(k)} - \mathcal{E}_{j}^{(k)}{'}||_{\diamond} ,
 \end{equation}
 where $\{\mathcal{E}_{j}^{(k)}\}_j$ and $\{\mathcal{E}_{j}^{(k)}{'}\}_j$ denote the noise channels affecting the $k$-th trap circuit where the protocol assumptions hold, and where the assumptions are violated, respectively.
\end{theorem}
This implies that any deviation in the bound due to violation of the assumptions depends only linearly on the diamond norm distance of each of the corresponding noise channels affecting each circuit.
If the assumptions are only weakly violated, i.e. each term $|| \mathcal{E}_{j}^{(k)} - \mathcal{E}_{j}^{(k)}{'}||_{\diamond} $ is small for all $j$ and $k$, the outputs of the trap circuits remain close to those obtained when the assumptions hold.
Therefore, the protocol is robust to any small violations of the assumptions, since these will have a small effect on the computed upper bound.

\begin{figure}
    %\centering
    \begin{subfigure}[t]{0.47\textwidth}
     %   \centering
        \includegraphics[width=\textwidth]{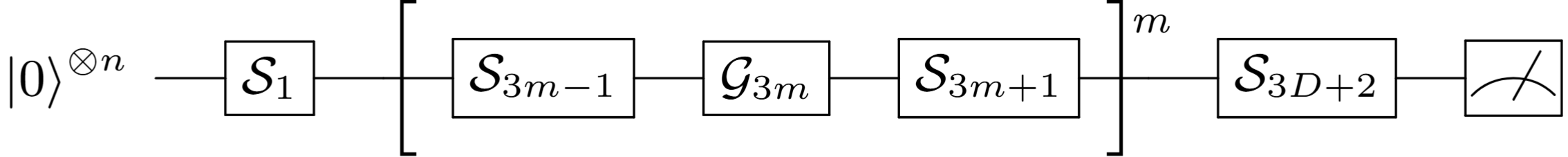}
       \caption{}
    \end{subfigure}
    \hspace{4em}
    \begin{subfigure}[t]{0.47\textwidth}   
            \includegraphics[width=\textwidth]{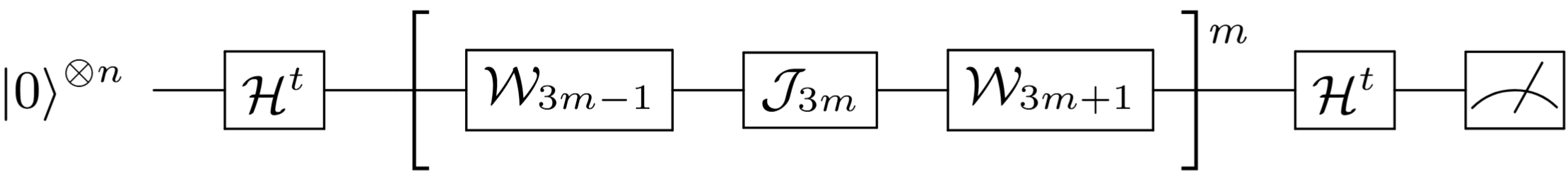}
       \caption{}
    \end{subfigure}
\caption{\textit{The logical circuit structure used for compilation of target and trap circuits, including state preparation and measurement.} 
The logical gate layers are numbered by subscript with their ordering in time.
The bracketed parts of the circuits indicate a repeated structure, with $m \in \{1,\ldots,D\}$, where the gates within each of the layers can change between repetitions.
The circuit diagram in (a) shows the structure used to compile the target circuit, where the notation $\mathcal{S}$ denotes a layer of single-qubit Clifford gates, and $\mathcal{G}$ a gate layer that can contain multi-qubit Clifford gates, non-Clifford gates and identity gates. 
The circuit diagram in (b) shows the trap circuit structure used for compilation, where $\mathcal{H}$ denotes a layer of Hadamard gates, with $t \in \{0,1\}$, $\mathcal{W}$ a single-qubit randomising gate layer containing $S$, $S^{\dagger}$ and $H$ gates, and $\mathcal{J}$ a gate layer that can contain multi-qubit Clifford gates, trap circuit versions of the corresponding non-Clifford gates found in the target circuit, and identity gates. 
The logical randomised compiling gate layers are omitted from these circuit diagrams for the sake of simplicity.
}
\label{fig:target_and_trap_circuit_structure}
\end{figure}

\vspace{-0.8em}
\subsection{Completeness and soundness} \label{sec:comp_and_sound}
\vspace{-0.5em}

If the measured output of a trap circuit deviates from the ideal output, a logical error is guaranteed to have occurred.
Therefore, the false negative rate of the protocol is 0, and it has perfect completeness.
There are two factors to be considered for the protocol soundness:
(1) The probability that if an error affects a trap circuit that it will be detected.
(2) The probability that if multiple errors affect a trap circuit, they cancel and so are not detected.
In Appendix \ref{trap_detect_errors}, we show that any single logical Pauli error of arbitrary weight occurring at any single time-step during a trap circuit is detected with probability $p_{det} \geq 1/2$.
We derive upper bounds on the probability of error cancellation for three different scenarios.
Non-Markovian noise can result in classically correlated errors that may cancel in the trap circuits.
The three scenarios reflect different assumptions that can be made about the Markovianity of the noise.
For each scenario, an upper bound on the probability of error cancellation is provided.
First, we show that if logical error rates are sufficiently low and noise is Markovian then error cancellation can be neglected. 
Second, if the dominant sources of error are multi-qubit Clifford gates and gates requiring magic states, then the probability of error cancellation is bounded by $p_{canc} \leq 1/2$.
This bound is suitable under the assumption that the dominant gate noise is non-Markovian, i.e. the noise from multi-qubit Clifford gates and gates requiring magic states, while single-qubit Clifford gate noise is Markovian or sufficiently weak that cancellation effects can be neglected.
And third, we show that with a slight modification to the trap circuit construction, the probability of any error cancellation is bounded by $p_{canc} \leq 7/8$.
For this result, no assumptions are made about the Markovianity of the noise.
These results are derived in Appendix \ref{bounds_on_prob_of_error_canc}.

Combining the bound for error detection probability with each of the three bounds for error cancellation probability yields three different leading-order soundness parameter values: $\delta_1 = 1/2$, $\delta_2 = 3/4$, and $\delta_3 = 15/16$.
The approximation errors associated with $\delta_1$ and $\delta_2$ are bounded in Lemmas \ref{neglect_canc} and \ref{only_J_canc_bound}, respectively.
These soundness parameters upper-bound the false positive rate of the trap circuits, and may be used to derive different TVD bounds. 
This is described in Appendix \ref{computing_bound_app}.
Fig. \ref{fig:heatmap-fp-pftqc} shows a heatmap of false positive rates for different numbers of qubits and physical error rates when performing IQP circuit sampling in the regime of full fault-tolerance. 
The simulations were run for $40$-layer circuits, varying the number of logical qubits from 5 to 50 and the physical error rate from $10^{-3}$ to $5 \times 10^{-2}$. 
In all cases, the experimentally observed false positive rates are well below the theoretical soundness bounds.

\begin{figure}%[!h]
    \centering
    \includegraphics[width=1\linewidth]{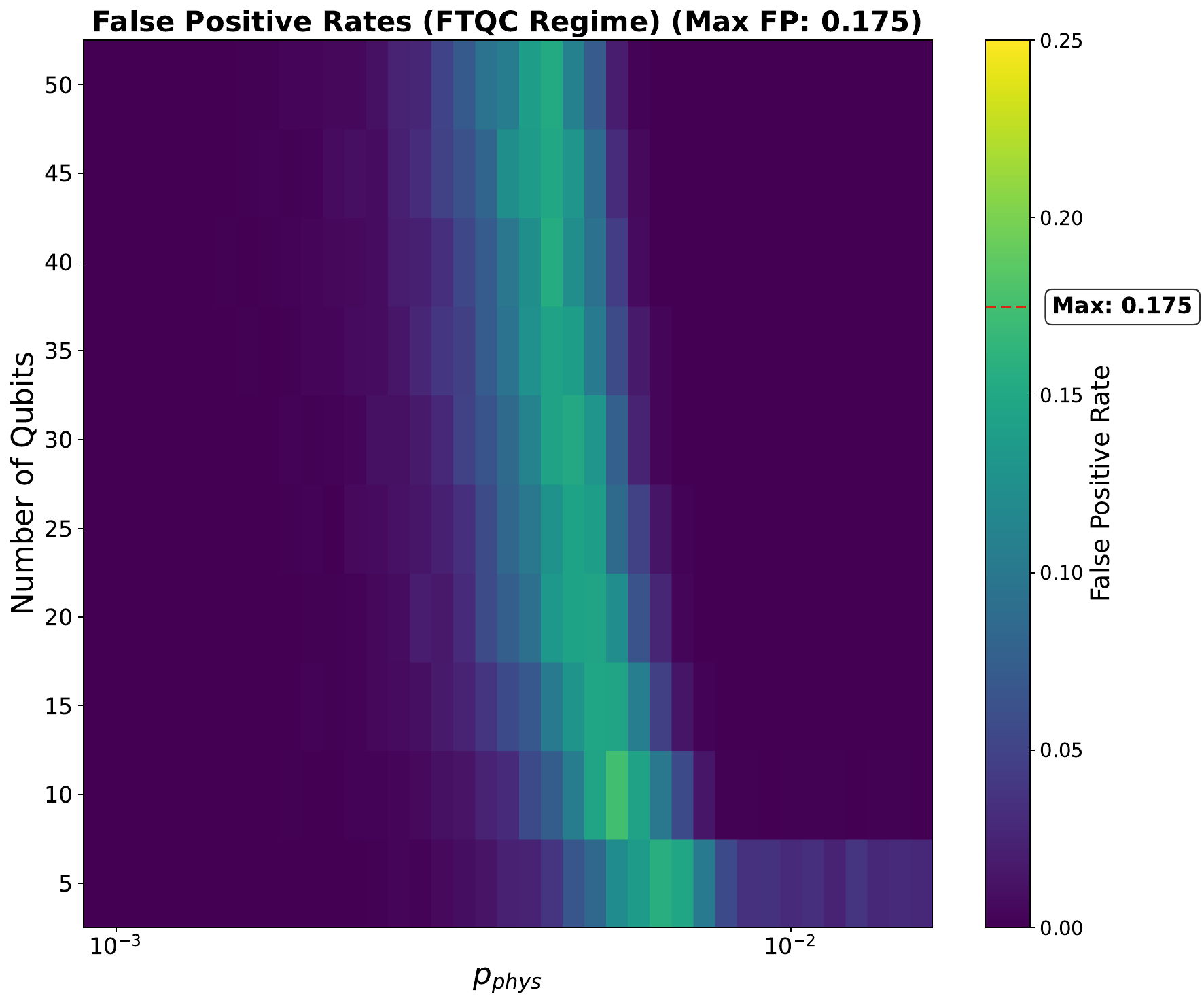}
    \caption{\emph{Soundness heatmap for IQP circuit sampling.} 
    This heatmap shows the trap circuit false positive rates for different numbers of logical qubits and physical error rates for IQP circuit sampling in the regime of full fault-tolerance (FTQC). 
    A logical depolarising noise model was used for the numerical simulations, and the false positive events resulted from error stabilisation and cancellation effects. 
    The false positive rates observed are well below the analytical soundness bounds provided in Section \ref{sec:comp_and_sound}. }
    \label{fig:heatmap-fp-pftqc}
\end{figure}

\vspace{-1.4em}

\subsection{Trap circuit construction}
\label{sec:construction-trap-computations}

\vspace{-0.6em}
The randomly generated trap circuits have the same structure as the target circuit, and stabilise the initialised logical state in the absence of logical errors.
Each of the logical gate operations performed during a trap circuit either compiles to an identity operation, or to an operation that stabilises the input logical state. 
This means the trap circuits provide a deterministic bit string output in the absence of noise.

In the trap circuits, the positioning and type of the C$Z$ gates are identical to those in the target circuit, but are sandwiched uniformly at random by $S$ and $H$ logical gates.
The randomly applied $S$ and $H$ operations, combined with the C$Z$ gates, are logically equivalent to randomly oriented CNOT gates.
Layers of logical Hadamard gates are included at the beginning and end of the trap circuits with probability $1/2$, otherwise layers of identity operations are applied. 
The randomly chosen $S$ and $H$ gates prevent the same Pauli noise from cancelling in different trap circuits, meaning the trap circuits detect logical Pauli errors affecting the trap circuits with probability lower-bounded by a constant.

Those gates in the trap circuits corresponding to logical single- and multi-qubit Pauli rotation gates in the target circuit are modified to stabilise the initialised logical state.
This is primarily achieved through the application of different magic state preparation procedures.
We will refer to magic states used in the target circuit as target magic states, and magic states used in the trap circuits as trap magic states.
If no magic state purification is used for the target and trap computations, the protocol overhead consists solely of an additional sampling cost, which scales independently of system size.
Two versions of each trap circuit are run on the quantum device. 
The first involves only preparation of magic states of the form $\ket{\pi}$, and the second only of the form $\ket{\pi/2}$.
In the first case, a fault-tolerant $S^\dagger$ gate is then applied immediately before injection, converting each state into a $\ket{\pi/2}$ state.
If either version of the trap circuit records an error when run on the quantum device, this is recorded as failure for that trap circuit.
Together, the $\ket{\pi}$ and $\ket{\pi/2}$ magic states are vulnerable to all possible logical Pauli errors that might occur during imperfect magic state preparation.

If methods for the fault-tolerant preparation of magic states are applied, including transversal, cultivation-based, distillation-based, and code-switching methods, a different procedure is used to generate the trap magic states.
If there are $K$ non-Clifford gates in the target circuit, each requiring a magic state, then in the trap circuits gate teleportation is applied to pairs of fault-tolerantly prepared $\ket{\pi/4}$ magic states to produce $K/2$ $\ket{\pi/2}$ states. 
These are assigned uniformly at random to the corresponding gates in the trap circuit.
The remaining $K/2$ gates requiring magic states in the trap circuits are performed using $\ket{\pi/2}$ states prepared directly in a fault-tolerant manner, i.e. without the need for purification, cultivation, or code switching methods.
This trap circuit construction has no additional magic state or sampling overhead. 
However, it does require an additional assumption regarding the noise behaviour of the magic state combination step, namely that the operation used to combine two  $\ket{\pi/4}$ states into a single trap magic state is non-noise-decreasing. 
And so the effective error rate of each resulting trap magic state is no smaller than the combined error rates of its two input target magic states. 
This ensures that the trap circuits may be used to provide a conservative bound on the target circuit logical error rate.
If it is assumed that fault-tolerantly prepared $\ket{\pi/4}$ and $\ket{\pi/2}$ states are of the same quality, then no additional overhead or assumption is required.

Under the previously stated noise assumptions, the trap circuit construction ensures that both the target and trap circuits are subject to the same distribution of possible logical noise channels.
More details are provided in Appendix \ref{trap_circuits_section}.

\subsection{Logical randomised compiling} \label{Log_rand_com_section}

The logical accreditation protocol employs \textit{logical randomised compiling}, a compilation scheme we introduce for Pauli twirling logical noise channels.
Inspired by the NISQ circuit compilation method of randomised compiling \cite{wallman_noise_2016}, this scheme applies to encoded logical qubits and gates performed by the consumption of magic states.
Logical Pauli twirling operations are directly integrated into the target and trap circuits without changing their computational logic. 
The twirling induces effective noise channels, transforming general logical noise, including, for example, coherent noise that can adversely affect logical qubits \cite{bravyi_correcting_2018}, into stochastic logical Pauli noise. 
The logical circuits then produce measurement outcomes consistent with their being subject to the effective stochastic Pauli channels induced by the twirling.

We also introduce a novel method for twirling arbitrary logical non-Clifford gates, addressing an open problem originally outlined in Piveteau et al.~\cite{piveteau_error_2021}. 
This involves independently twirling the magic state preparation noise and the noise from the injection gadget.
The method described in \cite{piveteau_error_2021} for twirling $\ket{\pi/4}$ magic states using the gate set $\{I, e^{-i\pi/4}SX\}$ relied on the $T$ gate being in the third level of the Clifford hierarchy, while our method avoids this requirement.
The only non-Clifford component of magic state injection is magic state preparation, with all other operations being Clifford operations.
Each component of magic state injection is Pauli twirled at the logical level.
For non-fault-tolerantly prepared arbitrarily phase-rotated magic states, the rotation angle is updated depending on each twirling instance.
We require that the phase rotation of $\ket{\theta}$ magic states depends on physical single-qubit rotation gates, rather than physical multi-qubit rotation gates.
Using Assumption (A1), the phase-rotation of the non-fault-tolerantly prepared magic states can then be updated without changing the state-preparation noise.
As the fault-tolerantly prepared magic states do not require any updates to the state preparation,  
the logical randomised compiling method is compatible with standard fault-tolerant magic state preparation methods, including transversal, distillation-based, and cultivation protocols. 
All implementation details of the logical randomised compiling scheme are provided in Appendix~\ref{log_rand_comp}.

There is evidence that for high code distances, logical noise affecting qubits encoded using the surface code converges toward logical stochastic Pauli noise \cite{bravyi_correcting_2018}.
However, if the number of available qubits is limited, this will restrict the achievable code distance.
This necessitates the use of logical randomised compiling. 
Recent work showed that it is possible to twirl over a subset of the Clifford group that commutes with a non-Clifford logical gate operation \cite{tsubouchi_symmetric_2025}.
We note that the approach that we propose could be adapted to perform logical non-Clifford gate twirling over the full Clifford group.

\subsection{Overhead} \label{overhead_section}

We now provide details regarding the resource overhead required to implement logical accreditation in terms of additional sampling cost and circuit depth.

The sampling overhead of logical accreditation comes from estimating the trap circuit failure rate to a desired statistical precision. 
Increasing the number of trap circuit runs reduces the statistical uncertainty in the estimated $\gamma$ bound. Estimating the mean trap failure probability to within additive error $\epsilon$ with confidence at least $1 - \alpha$ requires $M = O\big( \epsilon^{-2} \log(1/\alpha) \big)$ random trap circuit runs. 
Importantly, this sample complexity is independent of system size and depends only on the desired statistical precision.

Recompiling a given target circuit into the logical circuit structure shown in Fig.~\ref{fig:target_and_trap_circuit_structure} incurs at most a constant-factor increase in logical depth. For example, in a Pauli-based computation, all operations contributing to the logical unitary are contained within the $\mathcal{G}$-type layers. The inclusion of additional $\mathcal{S}$-type layers therefore increases the depth by at most a constant factor (at most three in the construction considered here).
In practice, for standard compilation schemes (e.g. Clifford+$T$), the increase can be smaller as such circuits already involve alternating layers of Clifford and non-Clifford operations.

In partially fault-tolerant regimes, where magic states are prepared directly, the protocol incurs a constant factor of two increase in the number of samples. 
This is because two versions of the trap magic states are required, increasing the number of trap circuit runs by a factor of two.
In fully fault-tolerant settings, including transversal, cultivation-based, distillation-based, and code-switching approaches to magic state preparation, the protocol can be implemented with no additional magic state or sampling overhead. 
This is achieved by preparing the same number of magic states as required for the target circuit and combining pairs of these states to form trap magic states, which are then assigned uniformly at random to the relevant trap circuit gates.

Overall, logical accreditation introduces at most constant overheads in all relevant resources.

\vspace{-0.7em}

\section{Numerical simulations}\label{sec:numerics}

The ability to quickly assess the accuracy of applications run on fault-tolerant quantum devices is essential to their practical utility.
Instantaneous Quantum Polynomial (IQP) circuit sampling and Hamiltonian simulation are two tasks for which certification of computational accuracy is crucial when it comes to assessing whether a quantum advantage has been achieved \cite{shepherd_temporally_2009,bremner_achieving_2017,hangleiter_fault-tolerant_2025}.
We apply the logical accreditation framework in numerical simulations of logical qubits encoded using the surface code to certify the accuracy of IQP circuit sampling and Hamiltonian simulation.
The results of these numerical experiments are presented in the following two subsections.

\vspace{-0.7em}

\subsection{Simulation details}

We numerically simulate partially and fully fault-tolerant circuits to test the logical accreditation framework. 
For comparison, we also use an existing NISQ accreditation protocol to certify computation on physical qubits \cite{ferracin_experimental_2021}.
For the simulations of partially and fully fault-tolerant circuits, we assume our computation is protected using surface code quantum error correction.
Following \cite{acharya_quantum_2025}, the logical error rate $p_L$ for the surface code is modelled using the expression
\begin{equation} \label{p_l_error_rate_eqn}
p_L \propto  \left(\frac{p_{phys}}{p_{th}}\right)^{\frac{d+1}{2}}
\end{equation}
where the physical error rate is denoted by $p_{phys}$, the code threshold error rate is denoted by $p_{th}$, and the code distance is denoted by $d$. 
In all simulations, we assume the surface code threshold error rate is $p_{th} = 0.01$, and use a constant prefactor of $c=0.03$ in the above expression \cite{beverland_assessing_2022}. 

Circuit noise is modelled using depolarising channels, applied either at the physical or logical level depending on the circuit type. 
In NISQ simulations, all gate operations are affected by depolarising noise with noise parameter $p_{phys}$. 
In partially fault-tolerant simulations, Clifford gates are assumed to be encoded and are subject to logical noise at rate $p_L$, while non-Clifford gates are assumed to rely on unpurified magic states and so are subject to physical noise at rate $p_{phys}$. 
In fully fault-tolerant simulations, all gates are subject to logical noise at rate $p_L$, regardless of type.

Each target computation is certified under all three frameworks:
\begin{enumerate}
\item NISQ accreditation using $n$ physical qubit circuits,
\item Logical accreditation with $n$ partially fault-tolerant logical qubit circuits,
\item Logical accreditation with $n$ fully fault-tolerant logical qubit circuits.
\end{enumerate}
This yields three TVD bounds for direct comparison between computation frameworks.
By examining these bounds, we can identify the operating regimes in which NISQ, partially fault-tolerant, or fully fault-tolerant computation is most effective.
Unless otherwise stated, we use $M=500$ trap circuits to compute the TVD bound. Each data point in our plots represents the mean and standard deviation of this bound, calculated over five independent runs, where $500$ trap circuits are independently sampled for each run.

\subsection{Certifying IQP circuit sampling}\label{sec:certifying_iqp_sampling}

IQP circuit sampling is a restricted model of quantum computation believed to be hard to simulate classically \cite{shepherd_temporally_2009}. 
An $n$-qubit IQP circuit starts with all qubits initialised in the $\ket{0}^{\otimes n}$ state, followed by a layer of Hadamard gates, a unitary made from randomly selected diagonal gates, another layer of Hadamard gates, and measurement in the computational basis. 
Approximating the output distribution of IQP circuits, even in the presence of noise, is $\#$P-hard \cite{bremner_average-case_2016, bremner_achieving_2017}.
This makes IQP circuits strong candidates for demonstrating quantum advantage on early fault-tolerant devices \cite{hangleiter_fault-tolerant_2025}.

Logical accreditation can be used to certify that logical error rates are below the threshold required for quantum advantage. 
According to \cite{bremner_average-case_2016}, if an IQP output distribution can be classically sampled to within 1-norm error of $1/192$, then the Polynomial Hierarchy would collapse to the third level. 
Certifying a TVD that satisfies this threshold would therefore provide strong evidence of quantum advantage.
Fig. \ref{fig:IQP} shows how this error threshold constrains the number of possible $T$ gates that can be used in a circuit such that the threshold is satisfied, for a range of $T$ gate noise rates, assuming perfect Clifford gates. 
The highlighted region indicates where error rates are sufficiently low and qubit counts are sufficiently high that a quantum advantage may be achievable.

Results from our simulations are presented in Fig. \ref{fig:iqp_chem_combined}.
In (a), TVD bounds from logical accreditation are plotted for circuits with 40 gate layers and $50$ or $500$ qubits. 
Each circuit is implemented in three versions: NISQ, partially fault-tolerant, and fully fault-tolerant. For the fault-tolerant circuits, we use distance-11 surface codes to model logical noise. 
This distance is chosen to explore a realistic noisy regime where the benefits of fault tolerance become clear.
We find that as physical qubit error rates improve, the TVD bound decreases. 
Once the error rate drops below the code threshold, the fault-tolerant circuits quickly start to outperform the unencoded circuits. 
Fully fault-tolerant circuits perform best overall in terms of the TVD bound over the range of parameters, followed by partially fault-tolerant circuits, and then NISQ circuits.

In the experiments plotted in (c), we fix the physical error rate at $p_{phys} = 10^{-3}$ and vary the number of gate layers, keeping qubit and trap circuit counts fixed as before. 
As circuit depth increases, the TVD bounds increase. 
NISQ circuits quickly approach TVD bound $\approx 1$, while partially and fully fault-tolerant circuits degrade more slowly. 
The TVD bounds increase with circuit depth faster for circuits with more qubits.

We also explore how performance changes with code distance and magic state quality. 
Fig. \ref{fig:iqp-vary-distances} shows that increasing code distance (from 3 to 13) reduces logical error rates and tightens the TVD bound, but at the cost of higher resource overhead.
Finally, we test how the fidelity of magic states affects performance. 
In Fig.~\ref{fig:vary-non-clifford} (a), we fix the physical error rate (and therefore the Clifford error rates) and vary only the $T$ gate error rate arising from imperfect magic states.
Our results show that improving magic state fidelity reduces the TVD, though the effect saturates at low code distances where Clifford noise becomes the dominant source of error.

In Appendix~\ref{apx:resource-analysis-for-IQP-advantage}, we estimate the minimum physical resources needed to achieve a certified TVD below the quantum advantage threshold for IQP sampling on superconducting devices using surface codes.

\begin{figure}%[hbt!]
\centering
\includegraphics[width=0.48\textwidth]{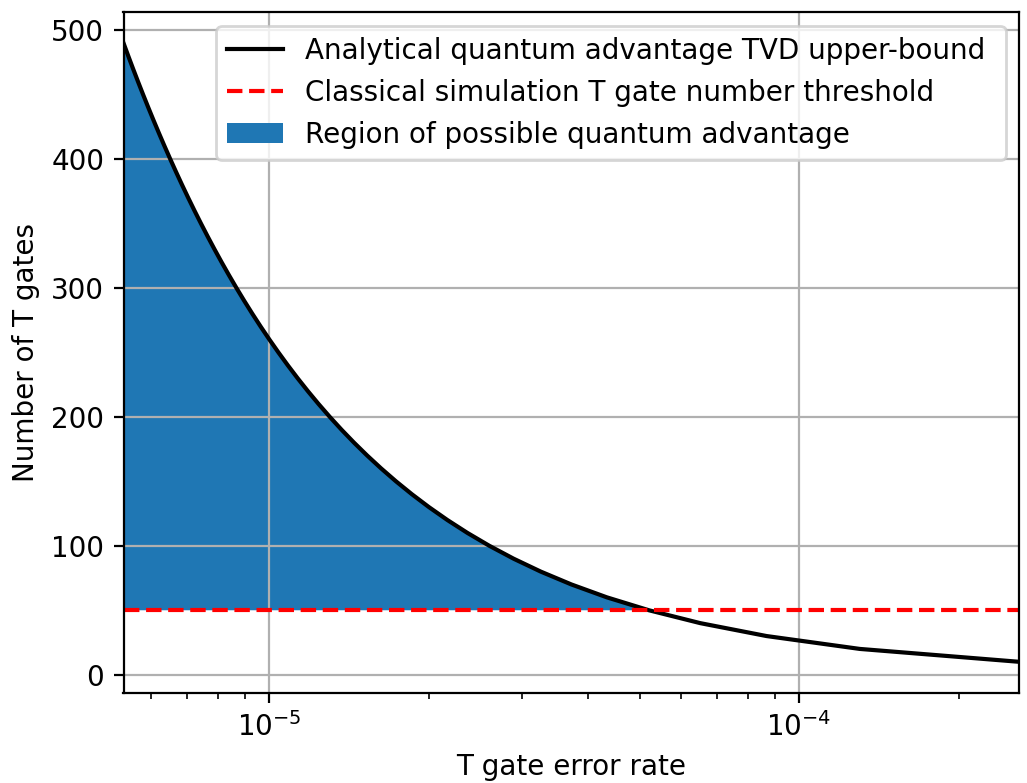}
\caption{\textit{Region of possible quantum advantage for IQP circuit sampling depending on $T$ gate noise and number of $T$ gates.}
The solid line on the plot is the maximum number of noisy $T$ gates such that the quantum advantage threshold derived in \cite{bremner_average-case_2016} is satisfied, plotted as a function of the $T$ gate error rate. 
It is assumed that Clifford gates are error-free, and circuit noise originates solely from $T$ gates. 
The best-known classical algorithms can simulate quantum circuits with up to roughly 50 $T$ gates \cite{bravyi_improved_2016, pashayan_fast_2022}. 
This threshold number of $T$ gates is included as the dashed horizontal line on the plot. 
The shaded region between the dashed and solid lines indicates the parameter regime where a quantum advantage is possible.}
\label{fig:IQP}
\end{figure}

\vspace{-0.8em}
\subsection{Certifying Trotterised Hamiltonian simulation}

We now consider the problem of certifying quantum Hamiltonian simulation with Trotterised circuits.
Hamiltonian simulation is generally considered a promising candidate application for fault-tolerant quantum computation.
For example, simulating the one-dimensional nearest-neighbour Heisenberg model with a random $Z$-field is of broad interest, particularly in the context of self-thermalisation and many-body localisation \cite{nandkishore_many-body_2015}.
Due to its interest in the condensed matter community, it has been used to benchmark various methods of quantum simulation \cite{childs_theory_2021}.

Hamiltonian evolution can be simulated on a quantum computer using Trotterised time evolution \cite{toshio_practical_2025}.
This method has been found to display better relative performance, when applied to achieve a finite target precision, than other asymptotically better scaling techniques \cite{kivlichan_improved_2020}.
This makes it a promising method for simulating quantum evolution on early fault-tolerant devices \cite{toshio_practical_2025, akahoshi_compilation_2025}.

The Trotter-Suzuki decomposition allows the time evolution operator $e^{-iHt}$ to be approximately decomposed into a sequence of $N$ discrete time-steps \cite{kivlichan_improved_2020, toshio_practical_2025}.
The second-order Trotter-Suzuki decomposition approximates these time-step evolution terms as 
\begin{equation}
\Big(e^{-i(\sum_{i=1}^{L} a_i P_i)t/N}\Big)^N \approx \bigg(\prod_{i=1}^{L}e^{-i (\frac{a_it}{2N})P_i} \prod_{i=L}^{1} e^{-i (\frac{a_it}{2N})P_i} \bigg)^N.
\end{equation}
If the Hamiltonian is expressible as a linear combination of Pauli operators, the previous expression may be written as a sequence of multi-qubit Pauli rotations
\begin{equation} \label{trotter_circuit}
\begin{split}
\bigg(\prod_{i=1}^{L} e^{-i (\frac{a_it}{2N})P_i} &\prod_{i=L}^{1} e^{-i (\frac{a_it}{2N})P_i} \bigg)^N \\
&= \bigg(\prod_{i=1}^{L} R_{P_i}\Big( \frac{a_it}{N}\Big) \prod_{i=L}^{1} R_{P_i}\Big( \frac{a_it}{N}\Big) \bigg)^N.
\end{split}
\end{equation}
Consequently, the Trotterised circuit consists of a sequence of Pauli rotations.
To run this circuit using surface code logical qubits, each high-weight logical Pauli rotation gate may be performed via lattice surgery techniques \cite{horsman_surface_2012}, with each non-Clifford Pauli rotation requiring the consumption of a single magic state.
Trotterisation induces an approximation error that may be efficiently bounded, but the computational effects of device noise are more difficult to quantify.

Logical accreditation is compatible with Trotterised target circuits composed of logical multi-qubit Pauli rotation gate operations.
The framework can be used to derive an upper bound on the TVD of the output from the ideal output when measuring an observable of the time-evolved state, as well as an error bound for the corresponding expectation values.
Although most applications involving circuit Trotterisation have some computational steps after the Trotterised evolution, for example, the QCELS algorithm \cite{toshio_practical_2025}, for the sake of generality we only consider bounding the error of the Trotterised circuit.
The approach used in the numerical experiments may straightforwardly be adapted for a particular application of Trotterisation.

\begin{figure*}[t!]
    \centering
    \begin{subfigure}[b]{0.48\textwidth}
        \centering
        \includegraphics[width=\textwidth]{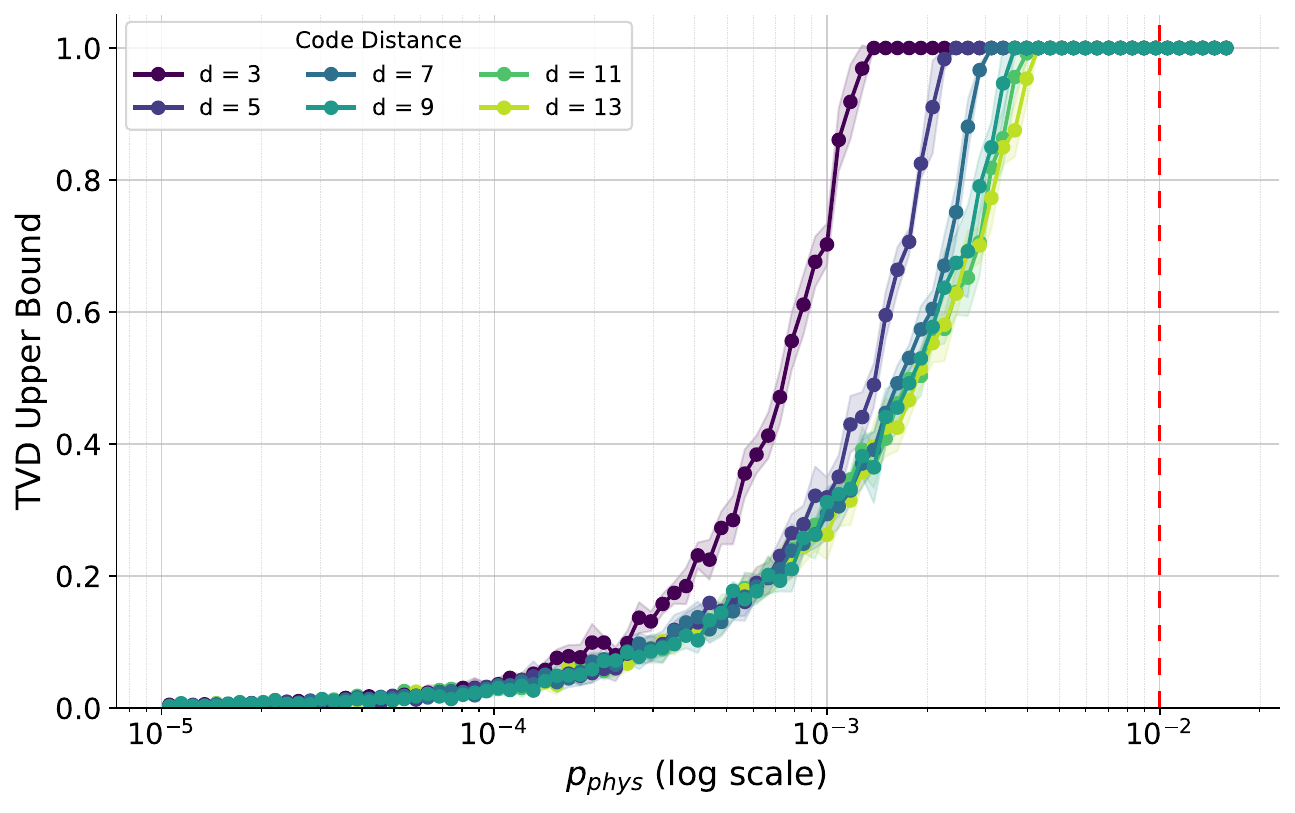}
        \subcaption{Partial fault tolerance}
        \label{fig:iqp-pft}
    \end{subfigure}
    \hfill
    \begin{subfigure}[b]{0.48\textwidth}
        \centering
        \includegraphics[width=\textwidth]{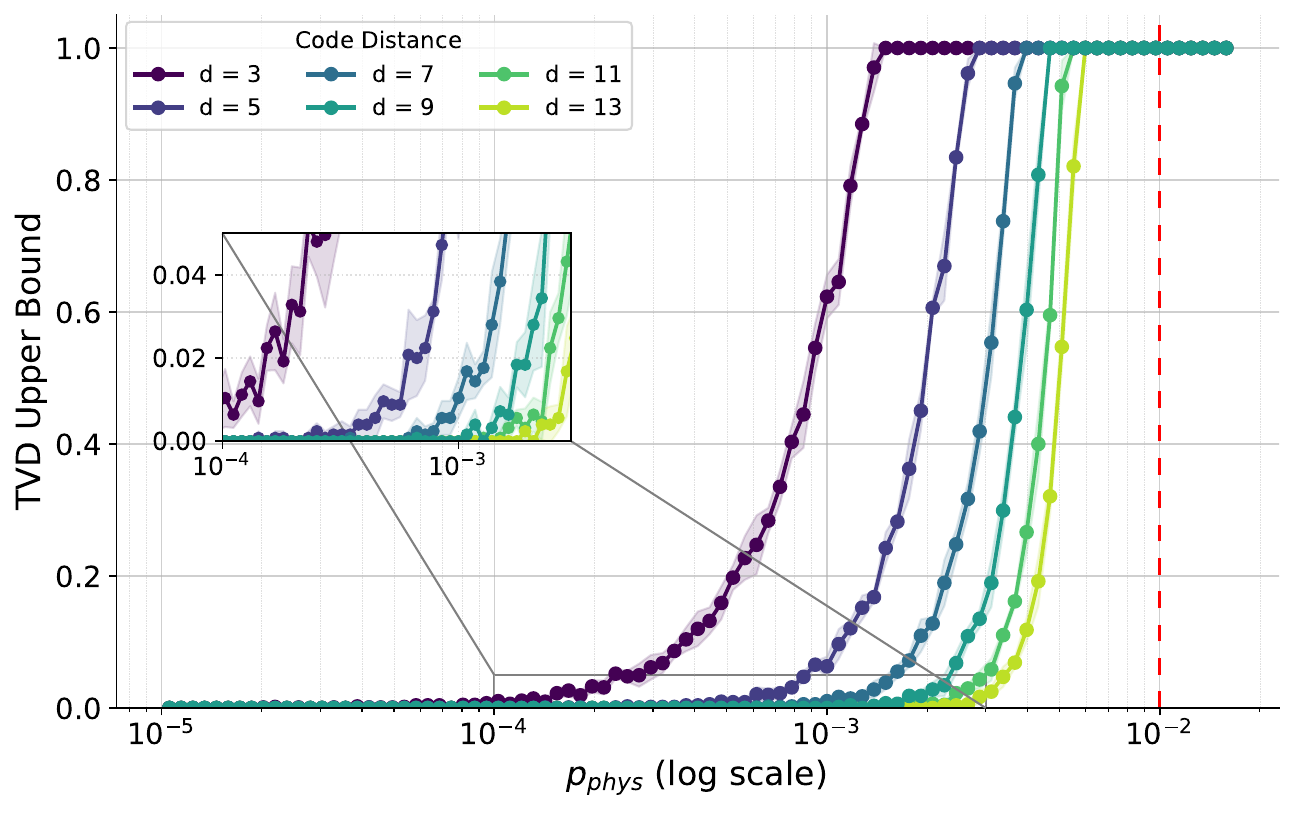}
        \subcaption{Full fault tolerance}
        \label{fig:iqp-full}
    \end{subfigure}
    \caption{\emph{Effect of code distance on the TVD bound for IQP circuits.} 
    Results for numerical simulations of IQP circuits with 15 logical qubits and 40 logical gate layers under (a) partial and (b) full fault tolerance. 
    The TVD upper bound is plotted against physical error rate.
    The fully fault-tolerant setting exhibits stronger dependence on code distance and faster convergence of the TVD bound as the physical error rate decreases. 
    The vertical dashed line marks the surface code threshold. 
    }
    \label{fig:iqp-vary-distances}
\end{figure*}

\begin{figure*}[t!]
    \centering
    \begin{subfigure}[b]{0.48\textwidth}
        \centering
        \includegraphics[width=\textwidth]{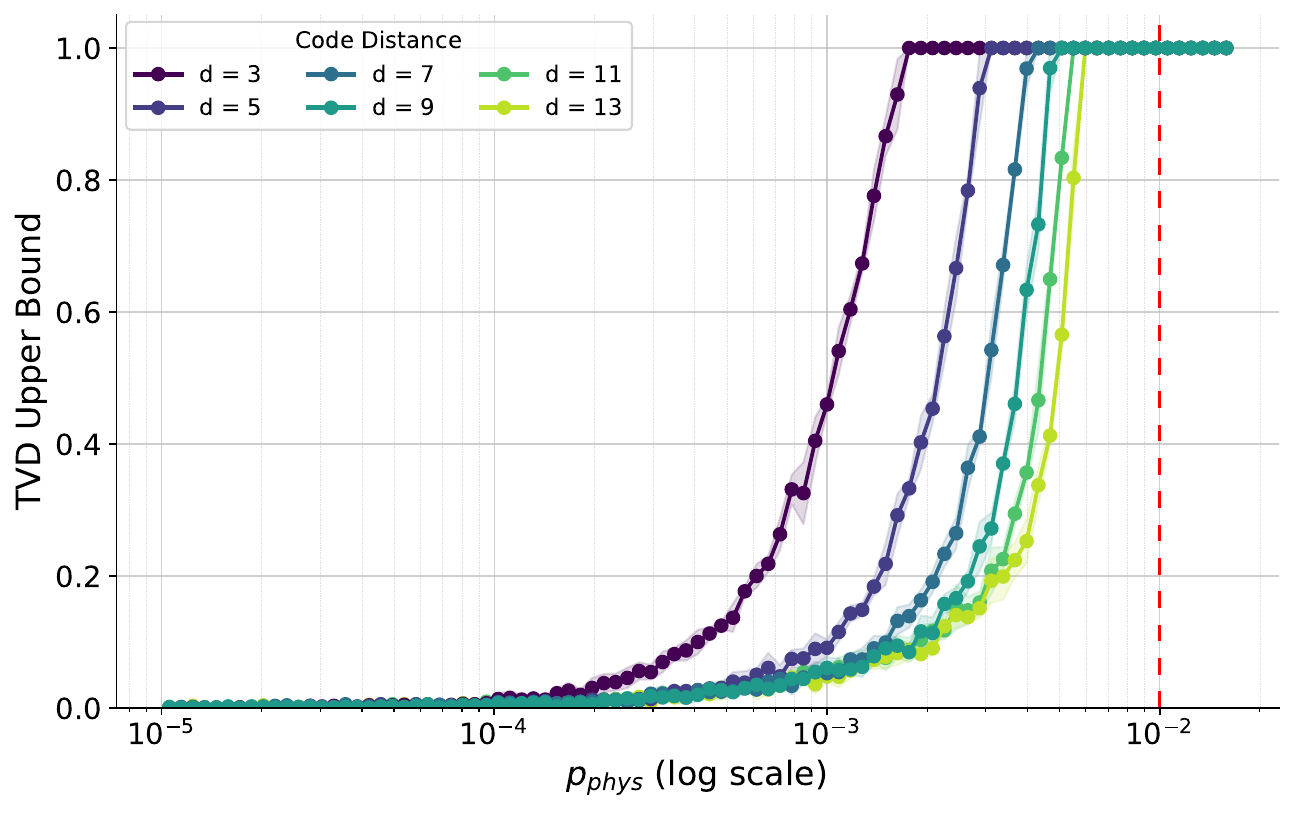}
        \subcaption{Partial fault tolerance}
        \label{fig:chem-pft}
    \end{subfigure}
    \hfill
    \begin{subfigure}[b]{0.48\textwidth}
        \centering
        \includegraphics[width=\textwidth]{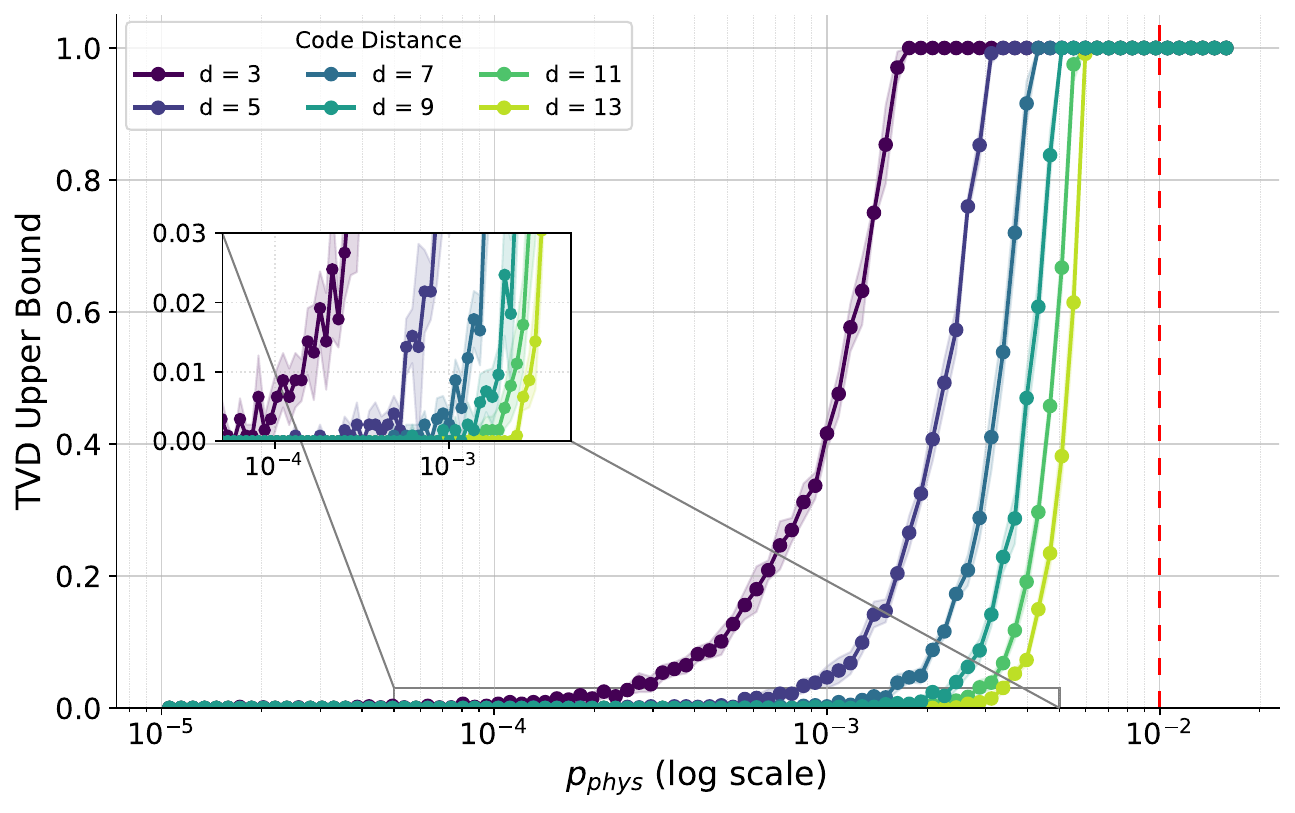}
        \subcaption{Full fault tolerance}
        \label{fig:chem-full}
    \end{subfigure}
    \caption{\emph{Effect of code distance on the TVD bound for Trotterised circuits.} 
    Results for numerical simulations of Trotterised circuits with 15 logical qubits and 40 logical gate layers under (a) partial and (b) full fault tolerance. 
    The TVD upper bound is plotted against physical error rate.
    Compared with partial fault tolerance, the fully fault-tolerant regime exhibits greater sensitivity to code distance and the TVD bound converges faster as the physical error rate decreases. 
    The vertical dashed line marks the surface code threshold.}
    \label{fig:chem-vary-distances}
\end{figure*}

\begin{figure*}[t!]
    \centering
    \begin{subfigure}[b]{0.48\textwidth}
        \centering
        \includegraphics[width=\textwidth]{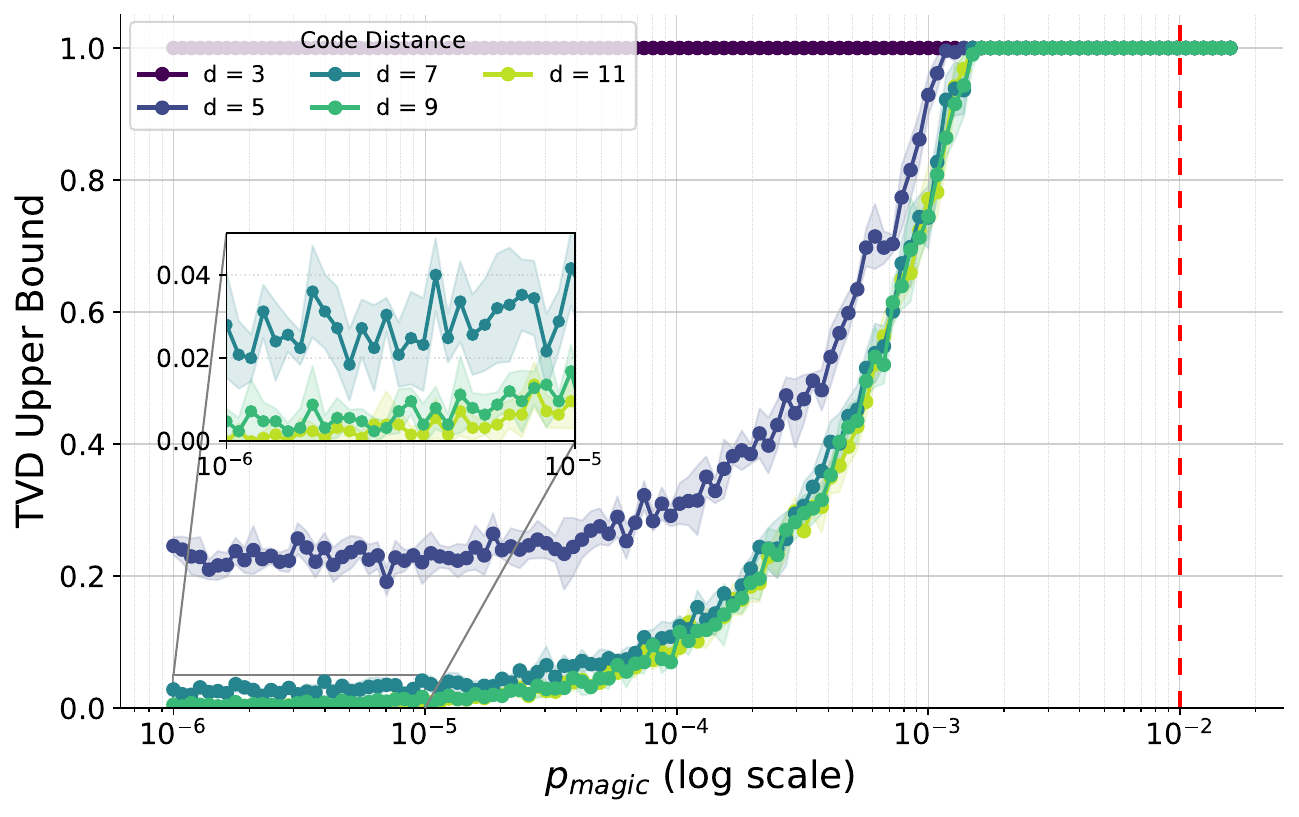}
        \subcaption{IQP circuits}
        \label{fig:iqp-vary-non-clifford}
    \end{subfigure}
    \hfill
    \begin{subfigure}[b]{0.48\textwidth}
        \centering
        \includegraphics[width=\textwidth]{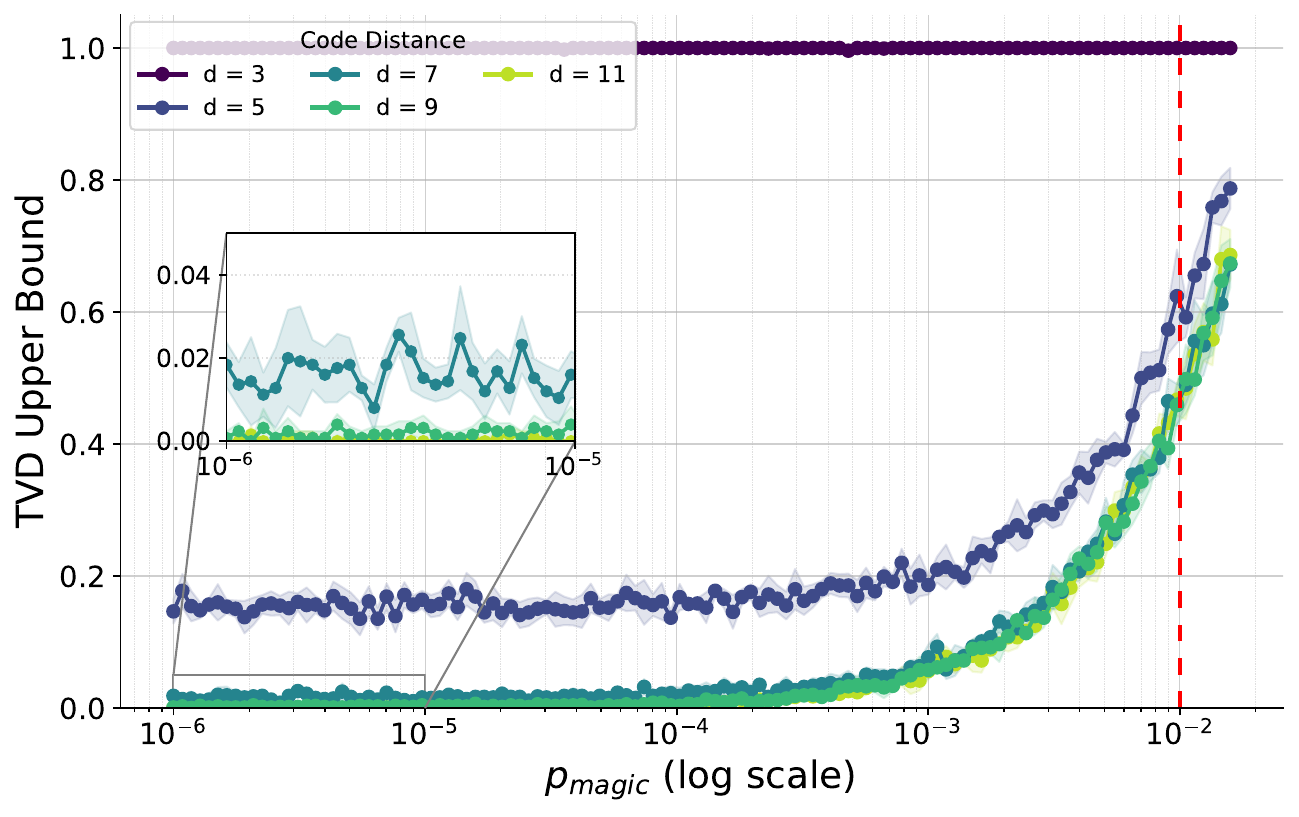}
        \subcaption{Trotterised circuits}
        \label{fig:chem-vary-non-clifford}
    \end{subfigure}
   \caption{
   \emph{Effect of magic state quality on TVD bounds for a range of code distances.} 
   Results from simulations of IQP circuits are shown in (a), and for Trotterised circuits are shown in (b), where the TVD bound from logical accreditation is plotted against magic state error rate.
   All simulations were for 50 logical qubits and 40 logical gate layers, and the physical error rate was fixed at $p_{\text{phys}} = 10^{-3}$. 
   The Clifford gate error rate is determined by the surface code distance, while the magic state error rate is varied. 
   Residual Clifford gate noise imposes a non-zero constant TVD floor, even as magic-state fidelity is improved.
   }
    \label{fig:vary-non-clifford}
\end{figure*}

We applied the logical accreditation framework in numerical simulations of generic Trotterised circuits. 
In the first set of experiments, we vary the physical error rate from $10^{-6}$ to $10^{-2}$ for NISQ, partially fault-tolerant, and fully fault-tolerant circuits. 
%Each circuit is certified using $500$ trap circuits. 
Noise from the high-weight Pauli rotation gates is modelled using global depolarising noise, and each logical gate layer is assumed to involve the same number of error correction cycles with the same associated logical noise.
As shown in Fig.~\ref{fig:iqp_chem_combined} (b), the TVD bound increases most rapidly for NISQ circuits, followed by partially fault-tolerant and then fully fault-tolerant circuits. As physical error rates approach the surface code threshold, the bounds for fault-tolerant circuits saturate near 1.

We then fix the physical error rate at $10^{-3}$ and vary the circuit depth up to 100 rotation layers. 
Fig.~\ref{fig:iqp_chem_combined} (d) shows that the TVD bound again grows fastest for NISQ circuits, with fault-tolerant circuits degrading more slowly.
We also examine how varying the surface code distance from 3 to 13 affects certification (Fig.~\ref{fig:chem-vary-distances}). 
Fully fault-tolerant circuits benefit more from increased code distance, with sharper reductions in TVD at lower physical error rates.
Finally, Fig.~\ref{fig:vary-non-clifford} (b) shows the effect of varying magic state fidelity on circuit performance. 
As in the IQP case, reducing $T$ gate noise tightens the TVD bound until Clifford noise becomes dominant and the improvement plateaus.
We note that although IQP circuits generally possess a simpler circuit structure than Trotterised Hamiltonian circuits, in the implementation used in these experiments they have a higher density of non-Clifford gates. 
This is particularly relevant in the partially fault-tolerant regime, where errors associated with non-Clifford operations dominate. 
Consequently, the IQP circuits exhibit higher logical error rates, which is reflected in the larger values of the corresponding TVD upper bounds in the numerical simulations. 
In contrast, the Trotter circuits exhibit smaller TVD upper bounds, indicating lower effective error rates. This can be seen by comparing results in Fig.~\ref{fig:iqp_chem_combined}(a) and (b), and (c) and (d).

To certify the full Hamiltonian simulation accuracy, the logical accreditation error bound must be combined with bounds on the Trotterisation error. 
For second-order Trotterisation, this error scales as $Wt^3/N^2$, where $W$ is a Hamiltonian-dependent constant \cite{kivlichan_improved_2020, childs_theory_2021, campbell_early_2022}.

\vspace{-0.8em}
\section{Applications}

\vspace{-0.4em}
\subsection{Entropy density benchmarking for fault-tolerant computation}
\vspace{-0.2em}

The logical accreditation framework can provide a means of extending benchmarking methods originally developed for NISQ devices to the fault-tolerant regime. 
To demonstrate this, we now show how it can be used to extend entropy density benchmarking to logical circuits.
The method of entropy density benchmarking was developed to assess the quality of NISQ circuits in terms of entropy accumulation from circuit noise \cite{stilck_franca_limitations_2021}.
Recent work has extended this approach to create heuristic models of entropy accumulation in NISQ circuits run on currently available quantum hardware, using these to compute circuit-size thresholds beyond which quantum advantage is unattainable \cite{demarty_entropy_2024}.
The second-order R\'enyi entropy density of the $n$-logical-qubit
experimental target-circuit output state, $\rho_{\mathrm{out}}$, averaged
over the uniformly random target execution position, is defined as
\begin{equation}
    n^{-1}S^{(2)}\!\left(\rho_{\mathrm{out}}\right)
    :=
    -n^{-1}
    \log_{2}\!\left(
        \operatorname{Tr}\!\left[\rho_{\mathrm{out}}^{2}\right]
    \right),
\end{equation}
where $n$ denotes the number of logical qubits used in the computation
and $S^{(2)}(\cdot)$ denotes the second-order R\'enyi entropy. The value
of $\gamma$ computed by logical accreditation may be used to derive the
following upper bound on the entropy density of the target-circuit
output state:
\begin{equation}
    \label{main_entropy_bound}
    n^{-1}S^{(2)}\!\left(\rho_{\mathrm{out}}\right)
    \leq
    -n^{-1}
    \log_{2}\!\left[
        1-2\gamma+\gamma^{2}\left(1+2^{-n}\right)
    \right].
\end{equation}
This bound, derived in Appendix \ref{entropy_bound_derivation}, holds with the same confidence as the TVD bound provided by the logical accreditation framework.
Logical accreditation can, therefore, be used to efficiently compute an upper bound on the Rényi entropy density of the target circuit, extending the entropy density benchmarking method to circuits performed using encoded logical qubits.
It would be interesting to further develop this connection, perhaps using logical accreditation as a means of creating heuristic models for entropy accumulation in logical circuits, similar to the approach followed in \cite{demarty_entropy_2024} for NISQ circuits.

\vspace{-0.9em}
\subsection{Certifying the practicality of applying quantum error mitigation to logical circuits}
\vspace{-0.4em}

Much recent work has focused on combining quantum error mitigation with quantum error correction for logical computations on early fault-tolerant quantum devices \cite{piveteau_error_2021, lostaglio_error_2021, suzuki_quantum_2022, dutkiewicz_error_2025}.
Error mitigation techniques have been proposed as a means of extending the computational performance of logical computations with unpurified magic states \cite{piveteau_error_2021}, and where code distances are too low to suppress noise sufficiently to achieve desired logical error rates \cite{suzuki_quantum_2022}.
It has been asserted that quantum error mitigation is primarily useful in the regime $\epsilon_GN_G=O(1)$, where $\epsilon_G$ is the independent physical gate error rate and $N_G$ is the number of physical gates in the circuit \cite{endo_hybrid_2021, zimboras_myths_2025}.
In previous work concerning the mitigation of gate errors for noisy physical qubits, the overall amplification in the estimator variance of the mitigated output was computed to be $(1+2\epsilon_G)^{2N_G} \sim e^{4N_G\epsilon_G}$ \cite{endo_practical_2018}.
If the value of $N_G\epsilon_G$ is too large, error mitigation cannot practically be applied due to the exponential scaling of the sampling overhead.
Similar arguments can be made concerning the efficiency of applying error mitigation to computations run on encoded logical qubits, where instead $\epsilon_G$ is the logical gate error rate and $N_G$ is the number of logical gates in the circuit.

Since logical accreditation provides a bound on the total circuit error rate, it can be used to bound the value of $N_G\epsilon_G$.
For example, if the maximum acceptable overhead to apply error mitigation to a logical circuit is a factor of 30 increase in the required number of samples relative to the unmitigated estimator, this condition is satisfied when $N_G\epsilon_G \leq 0.8503$.
The probability of no errors occurring due to gate noise for a circuit with $N_G$ gates is $(1 - \epsilon_G)^{N_G}$, which is approximately $e^{-\epsilon_G N_G}$. 
Assuming errors only occur during the circuit due to independent gate noise, then the total circuit error rate is $p_{err} \approx 1-  e^{-\epsilon_G N_G} $.
Setting $\epsilon_G N_G = 0.8503$ results in an upper bound on the circuit error rate where mitigation can be efficiently applied of $p_{err} \leq 1 - e^{-0.8503}$.
Since the $\gamma$ value provided by logical accreditation is an upper bound for both the TVD and the total circuit error rate, it can be used to check whether the error mitigation efficiency condition is satisfied.
This follows as a corollary if the inequality $ \gamma \leq 1 - e^{-0.8503}$ is experimentally satisfied.
In this manner, logical accreditation can be used to certify the practicality of applying quantum error mitigation to an encoded logical circuit.

% \null
\vspace{-1em}
\section{Conclusion}

% \vspace{-0.6em}

In this work, we introduced logical accreditation as a framework for efficiently certifying fault-tolerant quantum computations. 
The protocol runs trap computations alongside the target computation, and 
analysis of the trap computation measurement outcomes allows the accuracy of the target computation output to be  efficiently certified.
In contrast to traditional performance analyses of quantum error correction, this framework is sensitive to a broad range of noise models.
It can be applied to certify the computational accuracy of any fault-tolerant circuit that is compiled into the requisite form and run on a quantum device.

We demonstrated the protocol through numerical simulations of IQP circuit sampling and Trotterised quantum circuits. 
We also outlined some potential applications, including: (1) extending entropy density benchmarking to the fault-tolerant regime, and (2) assessing whether quantum error mitigation techniques can be efficiently applied to specific logical circuits.
In addition, we introduced a compilation strategy that transforms general decoding and magic-state preparation noise into logical stochastic Pauli noise. 
As part of this, we proposed a method for twirling arbitrarily rotated magic states, showing that non-transversal logical gates, including those beyond the $T$ gate, can also be twirled.

There are many opportunities to extend this work. 
With prior knowledge of noise behaviour, the trap circuit design could be altered to increase the detection probability for the most likely errors, thereby improving soundness.
It may be possible to optimise the protocol for specific quantum error-correcting codes, especially where code structure affects trap circuit design. 
% The protocol may be applied to quantum advantage experiments in the early fault-tolerance regime. 
The protocol could also be used to assess whether logical error rates are low enough that classical simulation algorithms, such as the one presented in \cite{schuster_polynomial-time_2025}, are inefficient.
Recently, verification and certification protocols have been adapted into error mitigation protocols on NISQ devices \cite{mezher_mitigating_2022, harris_error_2025}, it may be possible to adapt logical accreditation in a similar way to mitigate logical circuit noise.
It is important to note that logical accreditation does not measure the impact of gate synthesis errors.
As these errors are likely to play a significant role in early fault-tolerant quantum computation, extending logical accreditation to account for them would be highly valuable.
Finally, the protocol could be used in experiments on near-term devices to test the building blocks of QEC and, as devices improve, to certify the accuracy of large-scale fault-tolerant computations.

\vspace{-1.1em}

\subsection*{Acknowledgments}
\vspace{-0.6em}

We would like to thank Andrew Jackson and Ra\'ul Garc\'ia-Patr\'on for reading earlier versions of this manuscript, for providing very helpful feedback for its improvement, and for interesting and enjoyable discussions.
All authors acknowledge support from the EPSRC Quantum Advantage Pathfinder research program within the UK’s National Quantum Computing Center.
D.L. acknowledges support by the Autonomous Quantum Technologies (AutoQT) project (UKRI project 10004359). 
A.S. was supported by the Engineering and Physical Sciences Research Council (grant number EP/W524384/1).
J. R. is funded by an EPSRC Quantum Career Acceleration Fellowship (grant code: UKRI1224), and further acknowledges support from EPSRC grants EP/T001062/1 and EP/X026167/1.

\bibliography{main}

\appendix
\section*{Appendix}
\addcontentsline{toc}{section}{Appendix} % optional, for ToC
%\onecolumngrid

\vspace{1em}

\section{Trap circuit construction} \label{trap_circuits_section}

The trap circuit construction now described is the one used in the numerics, and in the derivation of the $\beta=0$ and $\beta=1/2$ upper bounds on the trap circuit error cancellation probability.
In Appendix \ref{altered_trap_construction}, a modified trap construction is described that allows for a $\beta=7/8$ error cancellation probability upper bound, this construction takes into account all possible error cancellation events.

The logic of each of the operations performed during the trap circuits either compiles to an identity operation or to an operation that stabilises the input state. 
This ensures that the trap circuits provide a deterministic bit string output in the absence of errors.
The trap circuits are generated using the same circuit structure as the target circuits.
The positioning and type of the multi-qubit Clifford gates are kept the same.
The single-qubit Clifford gate layers are replaced with randomly chosen $S$, $S^{\dagger}$ and $H$ gates.
When randomly applied $S$, $S^{\dagger}$ and $H$ operations sandwich a C$Z$ gate, the resulting combined operation is logically equivalent to a randomly oriented CNOT gate, as shown in Fig. \ref{fig:CZ_rand}.

A layer of logical Hadamard gates is included at the beginning and end of the trap circuits with probability $1/2$.
An important function of the randomly chosen $S$ and $H$ operations is to randomly map logical Pauli error operators to different logical Pauli operators, thereby preventing the same types of error cancellation from happening in different trap circuits.
This property is used to derive the protocol guarantees in Appendix \ref{trap_detect_errors}.

We now describe the treatment of gates performed by the consumption of both unpurified and purified magic states in the framework.

\subsection{Logical gates performed using unpurified magic states}

In the trap circuits, gates corresponding to logical single- and multi-qubit Pauli rotation gates in the target circuit are modified to stabilise the logical state.
The magic state preparation procedure is different for the trap and target circuits.
We will refer to magic states used in the target circuit as \textit{target magic states}, and magic states used in the trap circuits as \textit{trap magic states}.

In the scenario where target and trap magic states are not purified before consumption, two versions of each trap circuit are run on the quantum device. 
If at least one of the two versions of the trap circuit detects an error, this is recorded as a failure for that trap circuit.
In one version, only $\ket{\pi}$ magic states are prepared for consumption, and fault-tolerant $S^\dagger$ gates are applied immediately before injection.
In the other version, only $\ket{\pi/2}$ states are prepared.
Although the state preparation is performed differently, both versions ultimately produce $\ket{\pi/2}$ states.
The reason for this approach is that $\ket{\pi}$ magic states are only vulnerable to $Y$ and $Z$ Pauli errors during state preparation, whereas $\ket{\pi/2}$ magic states are vulnerable to $X$ and $Z$ Pauli errors.
Gate injection with a $\ket{\pi/2}$ state combined with a fault-tolerant Clifford correction gate implements identity or a Pauli gate.
% If a trap magic state is used to perform a Pauli rotation gate, 
If injection results in a Pauli gate then a single fault-tolerant Pauli $\pi$-rotation gate is applied immediately after the gate is applied.
In the absence of errors, the combined logic of these two operations is an identity operation.
In practice, the Pauli correction gates are absorbed into the following logical randomised compiling Pauli gate layer.
When $\ket{\pi/2}$ states are injected, the same adaptive measurement and correction procedure used in the target circuit is also used in the trap circuits. Any error in the adaptive measurement or feed-forward correction induces a non-trivial operation in the trap circuit and is therefore detectable.

The target and trap magic states are prepared using protocols in which the phase rotation of the prepared states is determined by physical single-qubit Pauli rotation gates applied after the logical ancilla $\ket{+}$ state has been encoded.
A consequence of Assumption A1, that single-qubit physical gate noise is gate independent, is that magic state preparation noise is phase rotation-angle independent.
Two examples of logical gate operations for magic state preparation that apply an arbitrary logical phase rotation to a $\ket{+}$ state are shown in Fig. \ref{fig:analog-z-rotation-different-weights} (b) and (c).
These operations are from the partially fault-tolerant STAR framework \cite{akahoshi_partially_2024}, and apply a logical $R_Z(\theta)$ operation using a physical single-qubit Pauli rotation gate, along with C$Z$ and CNOT gates.
And so, under the assumption that physical single-qubit gate noise is gate-independent, the logical state-preparation noise for the target and trap magic states is the same.

Each gate applied using a magic state is randomly sandwiched by $S$, $S^{\dagger}$ and $H$ gates. 
For RUS gates, each gate repetition is sandwiched by these operations. The random sandwiching gate operations are used to randomise error propagation through the trap circuits, ensuring errors are detected with constant probability.

\begin{figure*}[ht]
\centering
\begin{subfigure}[b]{0.3\textwidth}
    \includegraphics[width=\textwidth]{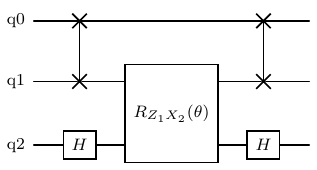}
    \caption{}
    \label{fig:subfig1}
\end{subfigure}
\hfill
\begin{subfigure}[b]{0.3\textwidth}
    \includegraphics[width=\textwidth]{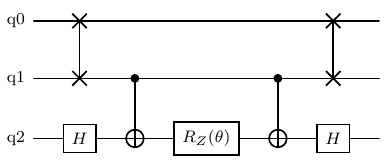}
    \caption{}
    \label{fig:subfig3}
\end{subfigure}
\hfill
\begin{subfigure}[b]{0.3\textwidth}
    \includegraphics[width=\textwidth]{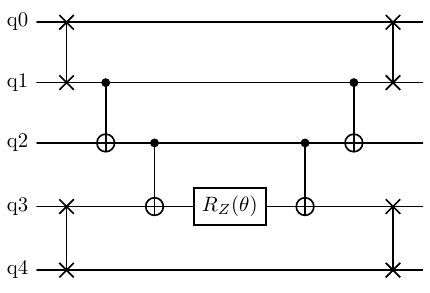}
    \caption{}
    \label{fig:subfig2}
\end{subfigure}
\caption{\textit{Examples of multi-qubit $Z$ rotations used to implement logical $Z$ rotations that are used in magic state preparation in the space-time efficient analog rotation scheme \cite{akahoshi_partially_2024, toshio_practical_2025}.} 
In (a) a weight-2 $ZZ$ rotation is applied to two physical qubits. 
In (b) and (c) a local $Z$ rotation is applied to a single qubit, that is then compiled through the use of additional CNOT operations into a $ZZ$ and a $ZZZ$ rotation, respectively. 
A requirement for implementing the protocol using arbitrarily phase-rotated magic states is that the magic state rotation angles be determined by the single application of a local rotation gate. 
This ensures that, assuming gate-independent noise on single physical-qubit gates, the noise affecting state preparation does not depend on the magic state phase-rotation angle.
The circuit shown in (a) is incompatible with the protocol; the circuits shown in (b) and (c) are compatible with it. }
\label{fig:analog-z-rotation-different-weights}
\end{figure*}

\subsection{Logical gates performed using purified magic states} \label{app_purification_section}

Fully fault-tolerant universal computation is possible if fault-tolerantly prepared magic states are used to perform those logical gates outside the gate set of the QEC code \cite{bravyi_universal_2005,campbell_roads_2017}.
One approach to fault-tolerantly preparing magic states is magic state purification.
The most well-known magic state purification protocol is magic state distillation \cite{bravyi_universal_2005, litinski_magic_2019}.
During distillation, a number of noisy input encoded magic states are purified, or distilled, to generate a smaller number of higher-quality states.
This is usually done by concatenating the code used to encode the logical qubits with another code possessing a transversal non-Clifford gate.
The concatenated code is used for error detection, postselecting on instances where no error is detected in order to generate higher quality magic states.
For the distillation to be successful, the fidelities of the input magic states must be greater than the fidelity threshold of the distillation method used.
If distillation is successfully repeated, at each repetition, magic states of progressively higher quality are generated. 
When a magic state of sufficiently high quality is produced it is then used to perform a non-Clifford gate on the computational logical qubits.
There are also a number of other methods for fault-tolerantly preparing magic states including transversal, cultivation-based, and code-switching approaches. 
Fault-tolerantly prepared trap magic states must have error rates similar to the target magic states for the accreditation bound to be valid.
We now describe how this can be achieved in the context of the Clifford+$T$ computational framework in a way that is compatible with all the methods of fault-tolerant magic state preparation mentioned.

For each trap circuit, the same number of $\ket{\pi/4}$ magic states are fault-tolerantly prepared as are used in the target circuit.
Let us say there are $K$ non-Clifford gates in the target circuit, each requiring a magic state.
Then the same number of the same type of magic states are prepared for each trap circuit.
However, prior to these being used to perform gates in the trap circuits, they are combined in pairs using gate teleportation to produce $K/2$ $\ket{\pi/2}$ states. 
These states are then assigned uniformly at random to the relevant gates performed by the consumption of these states in the trap circuit.
The remaining $K/2$ gates are performed using fault-tolerantly prepared $\ket{\pi/2}$ states.
This construction introduces no additional resource overhead. 
However, it requires an additional assumption regarding the noise behaviour of the magic state combination step, namely that the operation used to combine two $\ket{\pi/4}$ states into a single trap magic state is non-noise-decreasing. 
And so the effective error rate of each resulting trap magic state is no smaller than the combined error rates of its two input target magic states (up to higher-order corrections). This ensures that the trap circuits may be used to provide a conservative bound on the target circuit logical error rate.

In some cases, it may be reasonable to assume that fault-tolerantly prepared $\ket{\pi/2}$ or $\ket{\pi}$ states are of very similar quality to fault-tolerantly prepared $\ket{\pi/4}$ states.
This would depend on the type of magic state preparation method used and the form of the logical noise affecting the circuits.
The reason this assumption is avoided in the previously stated approach is that $\ket{\pi/4}$ states are vulnerable to $X$, $Y$ and $Z$ errors, while $\ket{\pi/2}$ states are stabilised by $Y$ errors and $\ket{\pi}$ states are stabilised by $X$ errors.
And, for example, purifying different magic states for target and trap circuits could then lead to appreciable differences in the fidelities of the purified target and trap magic states, potentially affecting the validity of the certification bound.
However, if this assumption is accurate, 
one could prepare and use different trap and target magic states.

The gates performed by consumption of the fault-tolerantly prepared trap magic states are immediately followed by a cancelling fault-tolerant Pauli $\pi$ rotation gate, resulting in overall identity operation.
This is absorbed into the following Pauli twirling logical gate layer.

\begin{figure*}
\centering
\includegraphics[width=0.65\textwidth]{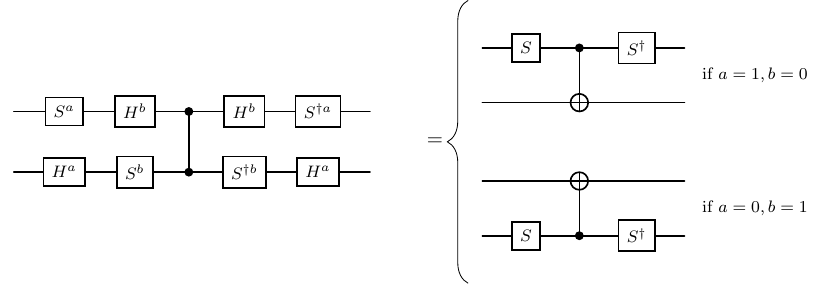}
\caption{In the trap circuits, each C$Z$ gate is randomly sandwiched by $S$ and $H$ gates such that the combined logic of the gates is a randomly oriented CNOT gate. 
This process is shown in the diagrammatic circuit equation, where the bit $b\in\{0,1\}$ is picked uniformly at random for each C$Z$ gate instance on the LHS, and $a=b\oplus1 \mod 2$.}
\label{fig:CZ_rand}
\end{figure*}

\section{Logical randomised compiling} \label{log_rand_comp}

We now describe the method we call \textit{logical randomised compiling}, a compilation technique that effectively transforms general logical circuit noise, from, for example, incorrect decoding or imperfect magic state preparation, into logical stochastic Pauli noise.
This is based on the NISQ circuit compilation technique called randomised compiling, where randomly chosen Pauli gates are compiled together with layers of single-qubit gates in order to Pauli twirl circuit noise into stochastic Pauli noise \cite{wallman_noise_2016}. 
We now describe how to compile the twirling operations into the logical circuits, and how these operations transform noise affecting different components of the logical circuit.
This includes the twirling of magic state preparation noise, both in the scenario where state purification is applied, and where it is not.

It was shown in
Refs.~\cite{winick_concepts_2022,liu_non-markovian_2024}
that randomised compiling may be extended to non-Markovian noise.
The following derivations explicitly treat the case in which the
logical noise at each circuit location is represented by a quantum
channel. For general non-Markovian noise, the complete multi-time
process may instead be represented by its Choi channel. Independently
applying logical Pauli twirls at each circuit location
Pauli-diagonalises this Choi channel, yielding a stochastic Pauli
process whose errors may remain classically correlated in time
\cite{winick_concepts_2022,liu_non-markovian_2024}.

Let
$\boldsymbol{P}=(P_0,\ldots,P_{D+1})$
denote a trajectory of Pauli errors acting at the different circuit
locations. The twirled process may be characterised by a joint
probability distribution
\begin{equation}
    \Pr(\boldsymbol{P})=c_{\boldsymbol{P}},
    \qquad
    c_{\boldsymbol{P}}\geq 0,
    \qquad
    \sum_{\boldsymbol{P}}c_{\boldsymbol{P}}=1.
\end{equation}
In general, this joint distribution need not factorise over the
circuit locations, and the Pauli errors may therefore be classically
correlated in time. Logical randomised compiling removes coherence
between distinct Pauli error trajectories, but does not necessarily
remove these classical temporal correlations. This extension assumes
that the twirling gates are sampled independently and that the noise
is independent of the choice of randomising gates, as specified by
the noise assumptions of the protocol. The implications of
temporally correlated Pauli errors for the soundness of the
certification protocol are discussed in Appendix~\ref{bounds_on_prob_of_error_canc}.

\subsubsection{Compiling together contiguous twirling operations} \label{log_circ_twirling}

Since the logical qubits are encoded using a QEC code, logical single-qubit gates and logical Pauli gates cannot be compiled together, as is done in NISQ randomised compiling \cite{wallman_noise_2016}.
We now describe how the logical randomised compiling operations are integrated into the logical target and trap circuits. The twirling operations are layers of logical single-qubit Pauli operations interleaving all the logical gate layers of the circuit.
Contiguous twirling operations are compiled together, so that a single logical layer of Pauli twirling operations includes both the undoing operation from the previous twirl and the operation for the following twirl.
For example, a single instance of a twirling gate layer $P_i'$ performed at an arbitrary point in the circuit combines both the undoing operation for the $(i-1)$-th twirl, $P_{i-1}^{(2)}$, and the application of the $i$-th twirl, $P_i^{(1)}$, that is
\begin{equation}
P_i' = P_i^{(1)} P_{i-1}^{(2)}.
\end{equation}
Then, rather than having a logical noise channel for both $P_i$ and $P_{i-1}^{\dagger}$, there is a single noise channel associated with the logical Pauli operation $P_i'$.
So the noisy application of $P_i'$ may be written
\begin{equation}
\tilde{P}_i' = \mathcal{E}_{{P}_i'} P_i^{(1)} P_{i-1}^{(2)},
\end{equation}
where the noise channel $\mathcal{E}_{{P}_i'}$ is twirling gate layer independent (see Assumption A2 in Section \ref{log_circ_noise_sect}) and so does not depend on the specific twirling operation being applied.

\begin{figure*}[thb]
\begin{subfigure}[b]{0.25\textwidth}
    \includegraphics[width=\textwidth]{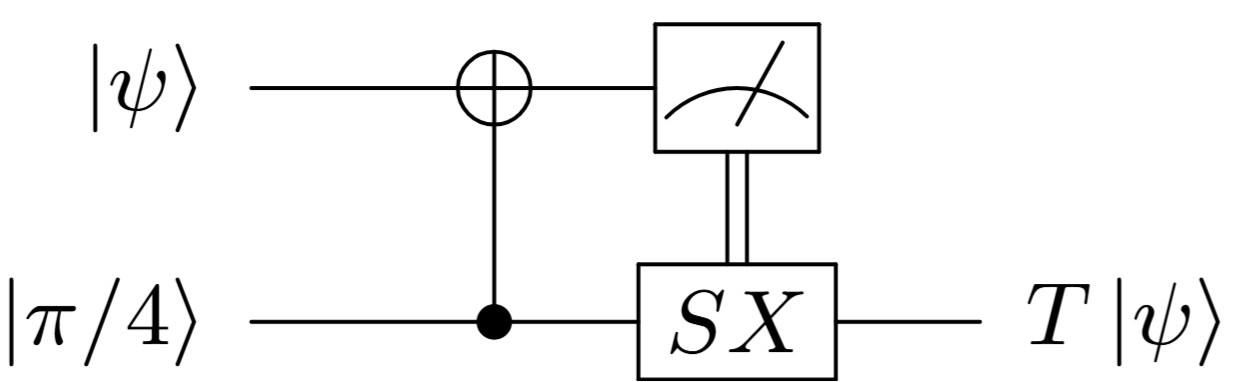}
    \caption{}
\end{subfigure}
\hspace{3.5cm}
% \hfill
\begin{subfigure}[b]{0.27\textwidth}
    \includegraphics[width=\textwidth]{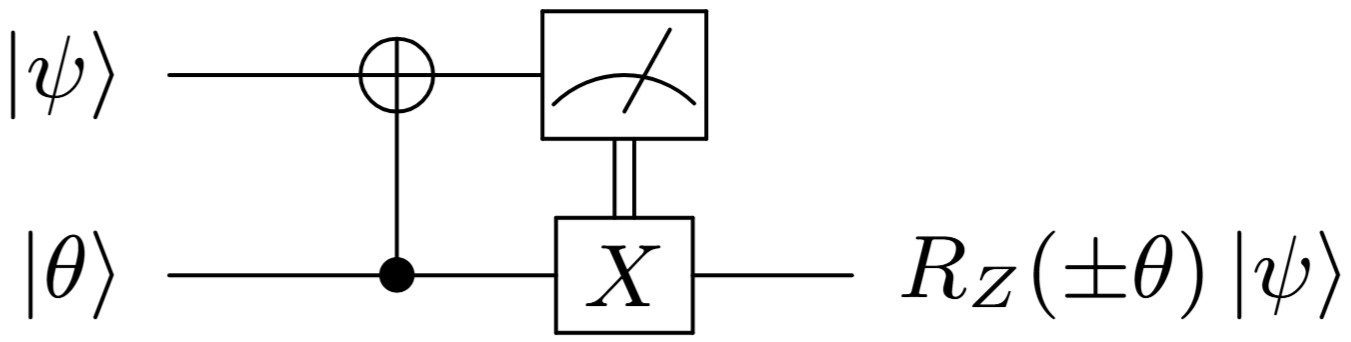}
    \caption{}
\end{subfigure}

\begin{subfigure}[b]{0.23\textwidth}
\vspace{1em}
    \includegraphics[width=\textwidth]{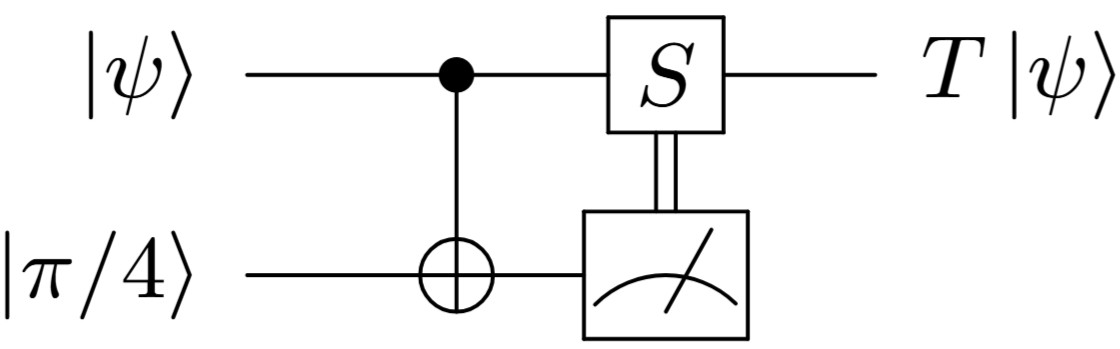}
    \caption{}
\end{subfigure}
\hspace{3.9cm}
% \hfill
\begin{subfigure}[b]{0.27\textwidth}
\vspace{1em}
    \includegraphics[width=\textwidth]{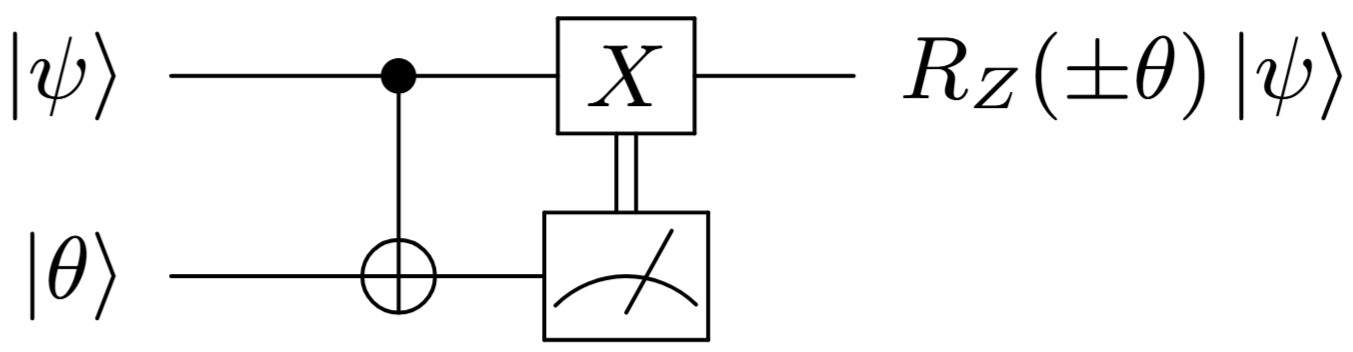}
    \caption{}
\end{subfigure}
\vspace{-0.2em}\caption{Quantum circuit diagrams showing different ways of performing $T$ gates, in (a) and (c), and RUS $R_Z(\theta)$ gates, in (b) and (d), by gate teleportation.
For the RUS gates shown in (b) and (d), if the measurement yields the eigenvalue -1 (i.e., measurement outcome bit “1”), the RUS gate is repeated.
}
\label{fig:tele_twirl}
\end{figure*}

\subsubsection{Twirling logical computational state preparation noise}

Computational logical qubits are initialised in the state $\ket{0^n}:=\ket{0}^{\otimes n}$.
If the prepared logical state is denoted by $\rho$, the logical state preparation error rate is defined as 
\begin{equation}
{\epsilon} :=1-\bra{0^n} \rho \ket{0^n}.
\end{equation}
The twirling gate set is chosen to be $\{I,Z\}^{\otimes n}$, and the twirled logical state may then be written
\begin{equation}
\begin{split}
\rho_{\mathcal{T}} &= 2^{-n}\cdot\sum_{P_i \in \{I,Z\}^{\otimes n}} P_i \rho P_i \\
&\hspace{-0.7em}= 2^{-n}\cdot\sum_{P_i \in \{I,Z\}^{\otimes n}} P_i \bigg( \sum_{x \in \{0,1\}^n }\sum_{y \in \{0,1\}^n } \alpha_{x,y} \ket{x} \bra{y}  \bigg) P_i. \\
\end{split}
\end{equation}
For each $\alpha_{x,y}\ket{x}\bra{y} $ term in the sum within the brackets, if $x \neq y$ then the term acquires a multiplicative prefactor of $-1$ from half the Pauli operators, and a prefactor of $+1$ from the other half.
If $x = y$ then the term will acquire a multiplicative prefactor of $+1$ from all the Pauli operators.
In consequence, the off-diagonal terms cancel, leaving only the diagonal terms in the sum
\begin{equation}
\begin{split}
\rho_{\mathcal{T}}
&= \sum_{x \in \{0,1\}^n}
\alpha_{x,x}\ket{x}\bra{x} \\
&= (1-\epsilon)\ket{0^n}\bra{0^n}
+ \epsilon
\cdot\sum_{s \in \{0,1\}^n\setminus\{0^n\}}
\widetilde{\alpha}_{s,s}\ket{s}\bra{s} \\
&= (1-\epsilon)\ket{0^n}\bra{0^n} \\
&\hspace{2em}
+ \epsilon
\cdot\sum_{P_i \in
\{I,X\}^{\otimes n}\setminus\{I^{\otimes n}\}}
\widetilde{\alpha}_{P_i}
P_i\ket{0^n}\bra{0^n}P_i ,
\end{split}
\end{equation}
where, for $\epsilon>0$,
$\widetilde{\alpha}_{s,s}:=\alpha_{s,s}/\epsilon$, and
$\widetilde{\alpha}_{P_i}$ denotes the corresponding coefficient
under the identification $P_i\ket{0^n}=\ket{s}$. Hence,
\[
\sum_{s \in \{0,1\}^n\setminus\{0^n\}}
\widetilde{\alpha}_{s,s}
=
\sum_{P_i \in
\{I,X\}^{\otimes n}\setminus\{I^{\otimes n}\}}
\widetilde{\alpha}_{P_i}
=1.
\]
When $\epsilon=0$, the state is simply
$\rho_{\mathcal{T}}=\ket{0^n}\bra{0^n}$. Thus, the state-preparation
noise is twirled into a stochastic logical Pauli-$X$ channel.

\subsubsection{Twirling logical Clifford gates}

Single and multi-qubit Clifford gates in the target and trap circuits may be efficiently Pauli twirled.
The Pauli gates applying each twirl are chosen uniformly at random, and, since the gates being twirled are Clifford, the gates resulting from propagating the initial Pauli gate operations through them may be efficiently classically computed.
The twirling operations leave the logical circuit unchanged, only transforming general decoding errors that might occur during the application of the logical Clifford gates into logical stochastic Pauli errors.
The twirling may be performed in a similar way to the randomised compiling technique proposed in \cite{wallman_noise_2016}, where physical Pauli operations are used to twirl noisy physical CNOT gates.

\subsubsection{Twirling logical non-Clifford Pauli rotation gates} \label{arbitrary_twirling}

The non-Clifford gates considered in the framework are performed by the consumption of magic states in magic state injection.
The only non-Clifford component of magic state injection is the magic state preparation, with all other operations involved being Clifford operations.
Each component of magic state injection is Pauli twirled at the logical level.
To simplify analysis of the effect of twirling on non-Clifford gate noise in this subsection, we treat the non-Clifford gate at the logical level without explicitly decomposing it into its injection gadget components.

An $n$-qubit logical Pauli rotation $R_P(\theta)$, where $P \in \{I,X,Y,Z\}^{\otimes n}$ and $\theta \in [-\pi,\pi]$, may be performed through the consumption of a $\ket{\theta}$ magic state.
We will now show that twirling such a gate using the approach described transforms any logical Pauli rotation gate noise from imperfect magic state preparation, teleportation or projective measurement into a logical stochastic Pauli channel.
We note that the method presented here provides one possible resolution to the open question posed by Piveteau et al., about how one might twirl the logical noise of non-transversal non-Clifford gates beyond the $T$ gate \cite{piveteau_error_2021}.

The binary symplectic representation of Pauli operators defined by $\phi(.)$ is the isomorphic map from the group of phaseless $n$-qubit Pauli operators to $2n$-bit binary vectors, that is $\mathcal{P}_n \rightarrow \mathbb{Z}^{2n}_2$.  
Where the $j$-th entry of the vector is equal to the power of the $X$ operator at the corresponding position in the Pauli operator, and the $(j+n)^{th}$ entry the power of the $Z$ operator in the corresponding position. The action of the map $\phi(.)$ on an arbitrary Pauli operator $P \in \mathcal{P}_n$ is
\begin{equation}
\begin{split}
\phi(P) &= \phi\big(\bigotimes_{j=1}^n X^{x_j} Z^{z_j} \big)\\
&= \big(\bigoplus_{j=1}^n x_j\big) \oplus \big(\bigoplus_{j=1}^n z_j\big)\\
&=(\mathbf{x}|\mathbf{z}).
\end{split}
\end{equation}
Now the commutation relation of two Pauli operators, $P_1 = X^{\mathbf{x_1}} Z^{\mathbf{z_1}}$ and $P_2 = X^{\mathbf{x_2}} Z^{\mathbf{z_2}}$, may be written as
\begin{equation}
(X^{\mathbf{x_1}} Z^{\mathbf{z_1}}) (X^{\mathbf{x_2}} Z^{\mathbf{z_2}}) = (-1)^{\mathbf{x_1}\cdot\mathbf{z_2} + \mathbf{x_2}\cdot\mathbf{z_1}}(X^{\mathbf{x_2}} Z^{\mathbf{z_2}})(X^{\mathbf{x_1}} Z^{\mathbf{z_1}}).
\end{equation}
The term dictating the power of $(-1)$ defines the symplectic product, this operation preserves the commutation properties of the Pauli group in this binary representation. The symplectic product is defined as
\begin{equation}
(\mathbf{x_1}|\mathbf{z_1})\odot(\mathbf{x_2}|\mathbf{z_2}) := \mathbf{x_1}\cdot\mathbf{z_2} + \mathbf{x_2}\cdot\mathbf{z_1}\text{ mod }2,
\end{equation}
and two Pauli operators commute iff
\begin{equation}
(\mathbf{x_1}|\mathbf{z_1})\odot(\mathbf{x_2}|\mathbf{z_2})=0,
\end{equation}
where the space $(\mathbb{F}^n_2\oplus\mathbb{F}^n_2, \odot)$ is sometimes referred to as a symplectic product space. 

Let the initial $m$-qubit quantum state, where $m \geq n$, be denoted by $\rho$.
And let the application of the noisy rotation gate be modelled as $\mathcal{E} \circ R_P(\theta) \circ (.)$, where $\mathcal{E}(.)$ is the logical noise channel associated with the gate due to the magic state injection gadget acting on the data qubits, excluding magic state preparation noise.
Applying the noisy Pauli rotation to the initial state $\rho$ then results in the state
\begin{equation}
\rho' = \mathcal{E} \circ R_P(\theta) \circ (\rho).
\end{equation}
When logical Pauli twirling is applied to the noisy rotation gate, the resulting state is
\begin{equation}
\begin{split}
\rho_{\mathcal{T}} &= 4^{-n} \bigg(\sum_{P_i \in \{I,X,Y,Z\}^{\otimes n}}P_i \circ \mathcal{E} \circ R_P(\theta) \circ P_i \bigg)\circ (\rho) \\
&= 4^{-n} \bigg(\sum_{P_i \in \{I,X,Y,Z\}^{\otimes n}} P_i \circ \mathcal{E} \circ P_i \\
& \hspace{3em} \circ R_P\big((-1)^{b_{P_i} \odot b_P}\cdot\theta\big)\bigg) \circ (\rho), \\
\end{split}
\end{equation}
where $b_{P_i}$ and $b_P$ are the binary symplectic representations of the Pauli operators $P_i$ and $P$ respectively.
To get the second equality, the Pauli operators have been propagated through the $R_P(\theta)$ gate to isolate the logical noise channel for the Pauli twirling.
However, the Pauli operators that do not commute with $P$ pick up a phase of $-1$.
The action of these phase factors resulting from the Pauli conjugation can be understood in terms of considering twirling at the level of the decomposition of the gate into its magic state injection components. 
In particular, when expressed at the level of the injection gadget, phase changes arising from anticommuting instances of the Pauli twirling of the logical gate are effectively absorbed into the twirling of the magic state preparation noise.
This is done either by updating the angle of the phase-rotation for non-fault-tolerant magic state preparation, or by including this in the twirling gate set of fault-tolerantly prepared magic states.
For example, the twirling gate set for $\ket{\pi/4}$ is $\{I, SX\}$, and so the $X$-twirling term is applied along with a phase-correction in the form of the $S$ gate.
The methods for twirling magic state preparation noise are described in Appendix \ref{app_section_magic_state_twirling}.

For gates performed using non-fault-tolerantly prepared $\ket{\theta}$ magic states, we require that the phase rotation depends on physical single-qubit rotation gates, rather than physical multi-qubit rotation gates.
Using Assumption (A1), the phase-rotation of the non-fault-tolerantly prepared magic states can then be updated during magic state twirling without changing the state-preparation noise.
This aspect of the compilation of the logical rotation gate used to prepare the magic states is illustrated in Fig. \ref{fig:analog-z-rotation-different-weights} (b) and (c), where multi-qubit Pauli rotations are compiled using physical single-qubit rotation gates and CNOT gates.
These are contrasted with the direct application of a multi-qubit Pauli rotation gate in Fig. \ref{fig:analog-z-rotation-different-weights} (a).
Fault-tolerantly prepared magic states do not require updates to the state-preparation procedure during twirling, as described in Appendix \ref{app_section_magic_state_twirling}. 
Consequently, the logical randomised compiling method is compatible with standard fault-tolerant magic-state preparation methods, including transversal, distillation-based, and cultivation protocols.

This means that we can remove phase changes induced by anticommuting Pauli terms, and the twirled state is
\begin{equation}
\begin{split}
\rho_{\mathcal{T}} &= 4^{-n} \bigg(\sum_{P_i \in \{I,X,Y,Z\}^{\otimes n}}P_i \circ \mathcal{E} \circ \\
&\hspace{3em} R_P\big((-1)^{b_{P_i} \odot b_P }\cdot\theta\big) \circ P_i \bigg)\circ (\rho) \\
&= 4^{-n} \bigg(\sum_{P_i \in \{I,X,Y,Z\}^{\otimes n}} P_i \circ \mathcal{E} \circ P_i \bigg) \\
& \hspace{3em} \circ R_P(\theta) \circ (\rho), \\
\end{split}
\end{equation}
The action of the twirled logical gate on the quantum state $\rho$ is then
\begin{equation}
\begin{split}
\rho_{\mathcal{T}}' &=  \bigg(4^{-n} \cdot \sum_{P_i \in \{I,X,Y,Z\}^{\otimes n}} P_i \circ \mathcal{E} \circ P_i \bigg) \circ R_P(\theta) \circ (\rho). \\
\end{split}
\end{equation}
The term inside the brackets is in the standard form required for Pauli twirling a quantum channel.
Suppose the channel $\mathcal{E}$ has Kraus operators $\{B_k\}_k$, that obey the trace-preserving property $\sum_k B_k^{\dagger}B_k=I$, known as the completeness condition. 
Channels that may be represented by a single Kraus operator are known as coherent, and when more than one Kraus operator is required, as incoherent. 
The evolution of state $\rho$ under the action of $\mathcal{E}$ may be modelled using the operator sum representation
\begin{equation}
\mathcal{E}(\rho) = \sum_k B_k \rho B_k^{\dagger}.
\end{equation}
The $n$-qubit Pauli basis is a complete basis for $n$-qubit quantum channels. Therefore, each Kraus operator in the set $\{B_k\}_k$ may be decomposed in the Pauli basis such that 
\begin{equation}
B_k = \sum_{i}\gamma_i^k P_i ,
\end{equation}
where $\{P_i\}_i$ is the set of all $n$-qubit Pauli operators and $|\{P_i\}_i|=4^n$. 
Now the operator-sum representation of the state evolution may be decomposed in the Pauli basis as
\begin{equation}
\begin{split}
\mathcal{E}(\rho) &= \sum_k \bigg(\sum_{i}\gamma_i^k P_i\bigg) \rho \bigg(\sum_{j}\gamma_j^k P_j\bigg)^{\dagger}\\
&= \sum_{k}\sum_{i,j}\big(\gamma_i^k P_i\big) \rho \big(\gamma_j^k P_j\big)^{\dagger}\\
&= \sum_{i,j}\sum_{k}\gamma_i^k\gamma_j^k{}^*   P_i \rho P_j\\
&= \sum_{i,j}\chi_{i,j}   P_i \rho P_j,\\
\end{split}
\end{equation}
where in the final line we have defined $\chi_{i,j}:=\sum_{k}\gamma_i^k\gamma_j^k{}^*$. 
Indeed, this is the well-known process or $\chi$-representation, and the matrix defined by $\chi_{i,j}$ is known as the  $\chi$-matrix. 
The relation $\sum_{i,j}\chi_{i,j}P_jP_i=I$ follows from the trace-preserving property of the set of Kraus operators representing the quantum channel. 
From which it follows that $\text{tr}(\sum_{i,j}\chi_{i,j}P_jP_i)=2^n$, or, equivalently, 
\begin{equation}
\begin{split}
\sum_{i}\chi_{i,i} &= \sum_{i,k}\gamma_i^k\gamma_i^k{}^*\\
&= \sum_{i,k}|\gamma_i^k|^2 \\
&= 1.
\end{split}
\end{equation}
Therefore, the diagonal entries of the $\chi$-matrix are real and sum to 1.
Now
\begin{equation}
\begin{split}
4^{-n}&\sum_{l} P_l P P_l \rho P_l P' P_l \\
&= 4^{-n}\sum_{l}  (-1)^{b_{P_l} \odot b_P}P \rho  (-1)^{b_{P_l} \odot b_{P'}}P'  \\
&= 4^{-n}\sum_{l}  (-1)^{b_{P_l} \odot b_P + b_{P_l} \odot b_{P'}}P \rho  P'  \\
&= 4^{-n}\sum_{l}  (-1)^{b_{P_l} \odot (b_P + b_{P'}) }P \rho  P'.  \\
\end{split}
\end{equation}
The distributivity of the symplectic inner product is used in the final equality.
Now if $P = P'$, then the symplectic product becomes $b_{P_l} \odot 0 = 0$ and the result is 
\begin{equation}
4^{-n}\sum_{l} P_l P P_l \rho P_l P P_l = P \rho P.
\end{equation}
Whereas if $P \neq P'$, then  for half of the Pauli group $b_{P_l} \odot (b_P + b_{P'}) = 0$ and half $b_{P_l} \odot (b_P + b_{P'}) = 1$ resulting in overall cancellation and the elimination of the term. So that
\begin{equation}
4^{-n}\sum_{l} P_l P P_l \rho P_l P' P_l = 0 \hspace{1.5em} \forall P \neq P'.
\end{equation}
An intuitive interpretation of this is that, if $P \neq P'$, the sum mod 2 of two binary representations $b_P + b_{P'}$ defines a new $n$-qubit Pauli operator, and this operator commutes with half of the Pauli group and anticommutes with the other half, resulting in cancellation.

The Pauli twirling eliminates all of the off-diagonal components of the Pauli decomposition of the channel, leaving only the diagonal elements. As these elements are real, and sum to 1, the transformed channel is a stochastic Pauli channel. Taking the Pauli twirl of the channel $\mathcal{E}$ transforms it to
\begin{equation}
\begin{split}
 \mathcal{E}' (\rho) &= |\{P_l\}_l|^{-1}\sum_{l} \sum_{i,j}\chi_{i,j}   P_l P_i P_l\rho P_l P_j P_l\\
&= \sum_{i}\chi_{i,i}    P_i \rho P_i \\
&= \sum_{i}c_{P_i}    P_i \rho P_i. \\
\end{split}
\end{equation}
Where we have defined the set of Pauli operator probability coefficients $\{c_{P_i}\}_i$, such that the coefficient for Pauli operator $P_i$ is $c_{P_i}$.
Therefore, the logical non-Clifford Pauli rotation gate noise is transformed by the twirling into stochastic logical Pauli noise.

\subsubsection{Logical measurement twirling}

Logical measurement noise may be Pauli twirled by applying random Pauli operations immediately prior to logical measurement, and then undoing the action of the Paulis after measurement through the use of classical postprocessing.
Noise affecting the measurement outcome $\ket{z}$, where $z \in \{0,1\}^n$, is modelled as $ \mathcal{E}_M  \circ (\ket{z}\bra{z}) $ where the logical measurement noise is denoted by $\mathcal{E}_M$. 
The measurement noise may then be Pauli twirled by applying a layer of logical Pauli operations immediately prior to measurement, and a virtual Pauli gate layer after measurement, positioned between the noise channel $\mathcal{E}_M$ and the ideal measurement outcome $\ket{z}\bra{z}$.
The virtual gate layer is performed by classical postprocessing of the measured output bit string, conditionally flipping the output bits to undo the effect of the logical Pauli layer on the measurement outcome.
That is, if the $i$-th term in the tensor product of the Pauli operator applied before measurement is an $X$ or a $Y$ operator, the $i$-th bit in the output bit string is flipped.
Whereas if it is an $I$ or a $Z$ operator, the bit is left unchanged.
This may be summarised by the relation $P_i\ket{z}\bra{z}P_i = \ket{z \oplus x_{P_i}}\bra{z\oplus x_{P_i}}$, where the sum is mod 2 and the bits of the ordered bit string $x_{P_i}$ are $1$s where the logical $n$-qubit Pauli operator $P_i$ acts on the corresponding logical qubit with an $X$ or $Y$ operator, and $0$s otherwise.
The measurement noise is then transformed into a logical stochastic Pauli channel in the following way.

The twirling operations are applied immediately before measurement,
and the twirled measurement outcome may then be written
\begin{equation}
\begin{split}
M_{\mathcal{T}}&= 4^{-n}
\sum_{P_i \in \{I,X,Y,Z\}^{\otimes n}}P_i \circ \mathcal{E}_M\circ\left(\ket{z\oplus x_{P_i}}\bra{z\oplus x_{P_i}}\right)\\
&= 4^{-n}\sum_{P_i \in \{I,X,Y,Z\}^{\otimes n}}P_i \circ \mathcal{E}_M\circ\left(P_i\ket{z}\bra{z}P_i\right)\\
&= 4^{-n}\sum_{P_i \in \{I,X,Y,Z\}^{\otimes n}}P_i \circ \mathcal{E}_M \circ P_i\circ\left(\ket{z}\bra{z}\right)\\
&= (1-\epsilon)\ket{z}\bra{z}\\
&\hspace{3em}+\epsilon\sum_{\substack{P_j\in\{I,X,Y,Z\}^{\otimes n}\\
x_{P_j}\neq 0^n}}\alpha_{P_j}P_j\circ\left(\ket{z}\bra{z}\right).
\end{split}
\end{equation}
The post-twirl measurement error rate is then defined as $\epsilon:=1-\bra{z}M_{\mathcal{T}}\ket{z}$.
The coefficients satisfy $\alpha_{P_j}\geq 0$ and
\[
\qquad
\sum_{\substack{
P_j\in\{I,X,Y,Z\}^{\otimes n}\\
x_{P_j}\neq 0^n
}}
\alpha_{P_j}=1.
\]
The relation $P_i\ket{z}\bra{z}P_i=\ket{z\oplus x_{P_i}}\bra{z\oplus x_{P_i}}$
is used in the second line, where the notation $\oplus$ denotes
addition modulo two. 
This operation can be achieved by reassigning bits in the measured outcomes.

\subsubsection{Compiling the twirling gate operations into the logical circuits}

The Pauli twirling operations are included in the target and trap circuits by effectively interleaving gate layers with Pauli twirling gate layers, with the first Pauli layer of each circuit used to twirl state-preparation noise and the last to twirl measurement noise.
The $i$-th logical circuit, $\mathcal{C}^{(i)}$, may be written as a sequence of $D$ logical gate layers and logical Pauli twirling gate layers as
\begin{equation}
\begin{split}
\tilde{\mathcal{C}}^{(i)} {}'
&= \mathcal{P}_{D+1}^{(1)}  \bigcirc_{j=1}^D (\mathcal{P}_{j}^{(2)} \circ \mathcal{L}_j^{(i)} \circ \mathcal{P}_{j}^{(1)}) \circ \mathcal{P}_{0}^{(2)},\\
\end{split}
\end{equation}
where $\mathcal{L}_j^{(i)}$ denotes the $j$-th logical gate layer of the $i$-th circuit.
For the $j$-th Pauli twirl, the notation $\mathcal{P}_{j}^{(1)}$ and $\mathcal{P}_{j}^{(2)}$ indicates the Pauli twirling operation being applied and then undone, respectively.
The $\mathcal{P}_{0}^{(1)}$ and $\mathcal{P}_{D+1}^{(1)}$ Pauli gate layers twirl the logical state preparation noise and the logical computational qubit measurement noise.
Contiguous twirling operation layers may be compiled together into a single Pauli layer.
The result of the Pauli twirling operations is that all logical state-preparation, gate and measurement noise is effectively transformed into logical stochastic Pauli noise.

\subsection{Twirling magic state preparation noise} \label{app_section_magic_state_twirling}

We distinguish between two frameworks for logical computation, describing means of twirling magic state preparation noise into logical stochastic Pauli noise in each instance.
Firstly, in the case where target magic states are of the form $\ket{\pi/4}$, these states can be purified or unpurified.
And secondly, in the case where target magic states are arbitrarily phase-rotated and are of the form $\ket{\theta}$, these states are not purified.

\vspace{0.1em}
\subsubsection*{\texorpdfstring{\textbf{Twirling scheme for computations involving $\ket{\pi/4}$ target magic states}}{\textbf{Twirling scheme for computations involving |pi/4> target magic states}}}

If $\ket{\pi/4}$ magic states are used in the target circuit, the following magic state twirling methods are used.
This approach is compatible with the inclusion of magic state purification in the computation.

We now describe the procedures for twirling the state-preparation noise of the $\ket{\pi/4}$ target magic states, and the $\ket{\pi/2}$ and $\ket{\pi}$ trap magic states into stochastic logical Pauli noise.
Although the proofs demonstrating that twirling the state-preparation noise of the $\ket{\pi/2}$ and $\ket{\pi}$ magic states results in stochastic $Z$-channels are similar to the one for $\ket{\pi/4}$ magic states, we include them for the sake of completeness.

The initialised state in each instance is a logical $\ket{+}$ state; and the magic state state-preparation noise may be twirled using the gate set $\{I,X\}$ into a logical stochastic Pauli-$Z$ channel.

\vspace{1em}
\subsubsubsection{\texorpdfstring{Twirling $\ket{\pi/4}$ magic states}{Twirling |pi/4> magic states}}

We now show that twirling $\ket{\pi/4}$ magic states transforms any noise affecting these states into stochastic logical $Z$ noise.
These states are twirled with respect to the gate set $\{I, A\}$, where $A=\ket{\pi/4}\bra{\pi/4}\ - Z\ket{\pi/4}\bra{\pi/4}Z$. 

The logical error rate is defined 
\begin{equation}
{\epsilon} :=1-\bra{\pi/4} \rho \ket{\pi/4}.
\end{equation}
Noting that $\ket{\pi/4}\bra{\pi/4}\ - Z\ket{\pi/4}\bra{\pi/4}Z = e^{-i\pi/4}SX$ and that 
\begin{equation}
e^{-i\pi/4} SX = e^{-i\pi/4} \begin{pmatrix}
0 & 1 \\
i & 0
\end{pmatrix} = \begin{pmatrix}
0 & e^{-i\pi/4} \\
 e^{i\pi/4} & 0
\end{pmatrix}.
\end{equation}
Let $\ket{w}:=Z\ket{\pi/4}$, the twirled state is then
\begin{widetext}
\begin{align}
\frac{1}{2}(\rho + A\rho A^\dagger) &= \frac{1}{2}\big((\alpha_1\ket{\pi/4}\bra{\pi/4}
+ \alpha_2\ket{\pi/4}\bra{w} + \alpha_3\ket{w}\bra{\pi/4} + \alpha_4\ket{w}\bra{w}) +A(\alpha_1\ket{\pi/4}\bra{\pi/4} + \alpha_2\ket{\pi/4}\bra{w}\\
&\hspace{3em}+ \alpha_3\ket{w}\bra{\pi/4} + \alpha_4\ket{w}\bra{w} ) A^\dagger)\\
&= \frac{1}{2}\big(\alpha_1(\ket{\pi/4}\bra{\pi/4} + A\ket{\pi/4}\bra{\pi/4}A^{\dagger})+\alpha_4(\ket{w}\bra{w} + A\ket{w}\bra{w}A^{\dagger}) \big)\\
&= \frac{1}{2}\big(\alpha_1(\ket{\pi/4}\bra{\pi/4} + \ket{\pi/4}\bra{\pi/4}) +\alpha_4(\ket{w}\bra{w} + \ket{w}\bra{w}) \big)\\
&= \alpha_1\ket{\pi/4}\bra{\pi/4} + \alpha_4\ket{w}\bra{w}\\
&= (1-{\epsilon})\ket{\pi/4}\bra{\pi/4} + {\epsilon}\ket{w}\bra{w}.
\end{align}
\end{widetext}
We have decomposed $\rho$ in the basis of the orthogonal states $\ket{\pi/4}$ and $\ket{w}$ with $\alpha_i\in\mathbb{C}$, $\alpha_1+\alpha_4=1$, and $\alpha_3=\alpha_2^*$.
%with $\alpha_i \in \mathbb{C}$ and $\sum_i|\alpha_i|^2=1$, and 
We then used that
\begin{equation}
\begin{split}
A\ket{\pi/4}&\bra{w}A^{\dagger} + \ket{\pi/4}\bra{w}\\
&= (\ket{\pi/4}\bra{\pi/4}\ - \ket{w}\bra{w}) \ket{\pi/4}\bra{w}(\ket{\pi/4}\bra{\pi/4}\ \\
&\hspace{3em}- \ket{w}\bra{w})^{\dagger} + \ket{\pi/4}\bra{w}\\
&= -\ket{\pi/4}\bra{\pi/4}\ket{\pi/4}\bra{w}\ket{w}\bra{w}+ \ket{\pi/4}\bra{w} \\
&=0
\end{split}
\end{equation}
to get the second equality.
The twirled noisy state may then be written as
\begin{equation}
\begin{split}
\mathcal{T}(\rho) &= (1-{\epsilon})\ket{\pi/4}\bra{\pi/4} + {\epsilon}\ket{w}\bra{w}\\
&=(1-{\epsilon})\ket{\pi/4}\bra{\pi/4} + {\epsilon}Z\ket{\pi/4}\bra{\pi/4}Z.\\
\end{split}
\end{equation}
And so the twirled noise is a stochastic $Z$ channel.

\vspace{1em}
\subsubsubsection{\texorpdfstring{Twirling $\ket{\pi/2}$ magic states}{Twirling |pi/2> magic states}}

We now show that twirling $\ket{\pi/2}$ magic states transforms any noise affecting these states into stochastic logical $Z$ noise.
These states are twirled with respect to the gate set $\{I, B\}$, where $B = ZX$. 
The logical error rate is defined 
\begin{equation}
{\epsilon} :=1-\bra{\pi/2} \rho \ket{\pi/2}.
\end{equation}
Let $\ket{w}:=Z\ket{\pi/2}$, the twirled state is then
\begin{widetext}
\begin{align}
\frac{1}{2}(\rho + B\rho B^\dagger) 
&= \frac{1}{2}\big( (\alpha_1\ket{\pi/2}\bra{\pi/2}
+ \alpha_2\ket{\pi/2}\bra{w} + \alpha_3\ket{w}\bra{\pi/2} + \alpha_4\ket{w}\bra{w})+B(\alpha_1\ket{\pi/2}\bra{\pi/2} + \alpha_2\ket{\pi/2}\bra{w}\\
&\hspace{3em}+ \alpha_3\ket{w}\bra{\pi/2}) + \alpha_4\ket{w}\bra{w})  B^\dagger \big)\\
&= \frac{1}{2}\big(\alpha_1(\ket{\pi/2}\bra{\pi/2} + B\ket{\pi/2}\bra{\pi/2}B^{\dagger})+\alpha_4(\ket{w}\bra{w} + B\ket{w}\bra{w}B^{\dagger}) \big)\\
&= \frac{1}{2}\big( \alpha_1(\ket{\pi/2}\bra{\pi/2} + \ket{\pi/2}\bra{\pi/2}) +\alpha_4(\ket{w}\bra{w} + \ket{w}\bra{w}) \big)\\
&= \alpha_1\ket{\pi/2}\bra{\pi/2} + \alpha_4\ket{w}\bra{w}\\
&= (1-{\epsilon})\ket{\pi/2}\bra{\pi/2} + {\epsilon}\ket{w}\bra{w}.
\end{align}
\end{widetext}
Similarly to before, $\rho$ is decomposed in the basis of the orthogonal states $\ket{\pi/2}$ and $\ket{w}$ with $\alpha_i\in\mathbb{C}$, $\alpha_1+\alpha_4=1$, and $\alpha_3=\alpha_2^*$.
We then used the relation
\begin{equation}
\begin{split}
B\ket{\pi/2}&\bra{w}B^{\dagger} + \ket{\pi/2}\bra{w}\\
&= (\ket{\pi/2}\bra{\pi/2}\ - \ket{w}\bra{w}) \ket{\pi/2}\bra{w}(\ket{\pi/2}\bra{\pi/2}\ \\
&\hspace{3em}- \ket{w}\bra{w})^{\dagger} + \ket{\pi/2}\bra{w}\\
&= -\ket{\pi/2}\bra{\pi/2}\ket{\pi/2}\bra{w}\ket{w}\bra{w}+ \ket{\pi/2}\bra{w} \\
&=0
\end{split}
\end{equation}
to get the third equality.
And so the twirled noise is a stochastic $Z$ channel with the noisy state expressed as 
\begin{equation}
\begin{split}
\mathcal{T}(\rho) &= (1-{\epsilon})\ket{\pi/2}\bra{\pi/2} + {\epsilon}\ket{w}\bra{w}\\
&=(1-{\epsilon})\ket{\pi/2}\bra{\pi/2} + {\epsilon}Z\ket{\pi/2}\bra{\pi/2}Z.\\
\end{split}
\end{equation}

\vspace{1em}
\subsubsubsection{\texorpdfstring{Twirling $\ket{\pi}$ magic states}{Twirling |pi> magic states}}

We now show that twirling $\ket{\pi}$ magic states transforms any noise affecting the state into stochastic logical $Z$ noise.
The logical error rate is defined as
\begin{equation}
{\epsilon} :=1-\bra{\pi} \rho \ket{\pi}.
\end{equation}
Let $\ket{w}:=Z\ket{\pi}$, the twirled state is then
\begin{widetext}
\begin{align}
\frac{1}{2}(\rho + X\rho X^\dagger) 
&= \frac{1}{2}\big( (\alpha_1\ket{\pi}\bra{\pi}
+ \alpha_2\ket{\pi}\bra{w} + \alpha_3\ket{w}\bra{\pi} + \alpha_4\ket{w}\bra{w})+X(\alpha_1\ket{\pi}\bra{\pi} + \alpha_2\ket{\pi}\bra{w}\\
&\hspace{3em}+ \alpha_3\ket{w}\bra{\pi}) + \alpha_4\ket{w}\bra{w})  X^\dagger \big)\\
&= \frac{1}{2}\big(\alpha_1(\ket{\pi}\bra{\pi} + X\ket{\pi}\bra{\pi}X^{\dagger})+\alpha_4(\ket{w}\bra{w} + X\ket{w}\bra{w}X^{\dagger}) \big)\\
&= \frac{1}{2}\big( \alpha_1(\ket{\pi}\bra{\pi} + \ket{\pi}\bra{\pi}) +\alpha_4(\ket{w}\bra{w} + \ket{w}\bra{w}) \big)\\
&= \alpha_1\ket{\pi}\bra{\pi} + \alpha_4\ket{w}\bra{w}\\
&= (1-{\epsilon})\ket{\pi}\bra{\pi} + {\epsilon}\ket{w}\bra{w}.
\end{align}
\end{widetext}
Similarly to before, $\rho$ is decomposed in the basis of the orthogonal states $\ket{\pi}$ and $\ket{w}$ with $\alpha_i\in\mathbb{C}$, $\alpha_1+\alpha_4=1$, and $\alpha_3=\alpha_2^*$.
We then used the relation
\begin{equation}
\begin{split}
X\ket{\pi}&\bra{w}X^{\dagger} + \ket{\pi}\bra{w}\\
&= (\ket{\pi}\bra{\pi}\ - \ket{w}\bra{w}) \ket{\pi}\bra{w}(\ket{\pi}\bra{\pi}\ \\
&\hspace{3em}- \ket{w}\bra{w})^{\dagger} + \ket{\pi}\bra{w}\\
&= -\ket{\pi}\bra{\pi}\ket{\pi}\bra{w}\ket{w}\bra{w}+ \ket{\pi}\bra{w} \\
&=0
\end{split}
\end{equation}
 to get the third equality.
And so the twirled noise is a stochastic $Z$ channel, and the twirled state may be expressed in the form 
\begin{equation}
\begin{split}
\mathcal{T}(\rho) &= (1-{\epsilon})\ket{\pi}\bra{\pi} + {\epsilon}\ket{w}\bra{w}\\
&=(1-{\epsilon})\ket{\pi}\bra{\pi} + {\epsilon}Z\ket{\pi}\bra{\pi}Z.\\
\end{split}
\end{equation}

\subsubsection*{\texorpdfstring{\textbf{Twirling scheme for computations involving $\ket{\theta}$ magic states}}{\textbf{Twirling scheme for computations involving |theta> magic states}}}

In the partially fault-tolerant regime one may use arbitrarily phase-rotated magic states.
If $\ket{\theta}$ magic states are used in the target circuit, the state preparation noise is transformed into logical stochastic Pauli noise using a Pauli twirling approach.
In this case, the gate applying the Pauli-$Z$ rotation is Pauli twirled, with the phase rotation angle reversed if the $X$ or $Y$ twirling operations are applied.
For this implementation we require that the phase rotation of the non-fault-tolerantly prepared magic states depends on physical single-qubit rotation gates, rather than physical multi-qubit rotation gates.
For example, the Pauli rotations shown in Fig. \ref{fig:analog-z-rotation-different-weights} (b) and (c) are valid, while Fig. \ref{fig:analog-z-rotation-different-weights} (a) is not.

Let the initial ancillary qubit logical state be denoted $\rho$.
In the absence of noise, this initial state is $\rho = \ket{+}\bra{+}$.
The state-preparation noise for this state is twirled using the gate set $\{I,X\}$, which transforms the logical state preparation noise into a Pauli-$Z$ channel.

To prepare the $\ket{\theta}$ state, a logical $Z(\theta)$ rotation must be applied to the ancillary qubit. 
Let the application of the noisy logical $Z$ rotation gate be modelled as $\mathcal{E} \circ Z(\theta) \circ (.)$, where $\mathcal{E}(.)$ is the logical noise channel associated with application of the logical $Z$ rotation gate.
Applying the noisy Pauli rotation gate to the state $\rho$ results in the state
\begin{equation}
\rho' = \mathcal{E} \circ Z(\theta) \circ (\rho).
\end{equation}
If Pauli twirling operations are included, sandwiching the noisy $Z(\theta)$ gate, instead the state is
\begin{equation}
\begin{split}
\rho_{\mathcal{T}} &= 4^{-1} \bigg(\sum_{P_i \in \{I,X,Y,Z\}}P_i \circ \mathcal{E} \circ Z(\theta) \circ P_i \bigg)\circ (\rho) \\
&= 4^{-1} \bigg(\sum_{P_i \in \{I,X,Y,Z\}}  P_i \circ \mathcal{E} \circ P_i \\
& \hspace{3em} \circ Z\big((-1)^{b_{P_i} \odot b_Z}\cdot\theta\big)\bigg) \circ (\rho), \\
\end{split}
\end{equation}
where $b_{P_i}$ and $b_Z$ are the binary symplectic representations of the Pauli operators $P_i$ and $Z$ respectively.
To get the second equality, the Pauli operators have been propagated through the $Z(\theta)$ gate to isolate the logical error channel for the Pauli twirling.
However, the Pauli operators that do not commute with $Z$ reverse the rotation angle of the $Z$ rotation gate.
This change in rotation angle can be counteracted by updating the rotation angle in each twirling instance.
Since the rotation angle depends on the action of physical rotation gates, and we are assuming that physical single-qubit gate noise is gate-independent, the updates can be made without changing the gate noise.
The twirling operation then becomes
\begin{equation}
\begin{split}
\rho_{\mathcal{T}} &= 4^{-1} \bigg(\sum_{P_i \in \{I,X,Y,Z\}}P_i \circ \mathcal{E} \circ Z\big((-1)^{b_{P_i} \odot b_Z }\cdot\theta\big) \circ P_i \bigg)\\
& \hspace{3em }\circ (\rho) \\
&= 4^{-1} \bigg(\sum_{P_i \in \{I,X,Y,Z\}}  P_i \circ \mathcal{E} \circ P_i \bigg)  \circ Z(\theta) \circ (\rho). \\
\end{split}
\end{equation}
The logical rotation gate channel has now been isolated from the gate, and the transformation to logical stochastic Pauli noise follows in the same way as in Appendix \ref{arbitrary_twirling}.
The Pauli twirling eliminates all of the off-diagonal components of the Pauli decomposition of the channel, leaving only the diagonal elements. 
Taking the Pauli twirl of the channel $\mathcal{E}$ transforms it to
\begin{equation}
\begin{split}
 \mathcal{E}' (\rho) 
&= \sum_{P_i \in \{I,X,Y,Z\}}c_{P_i}    P_i \circ Z(\theta) (\rho). \\
\end{split}
\end{equation}
Where we have defined the set of Pauli operator probability coefficients $\{c_{P_i}\}_i$, such that the coefficient for Pauli operator $P_i$ is $c_{P_i}$.
Therefore, the magic state preparation noise is transformed by the twirling into logical stochastic Pauli noise.

The corresponding trap magic states are twirled using the same approach.

\section{Computing the upper bound on the TVD of the experimental target circuit output from  the ideal output} \label{computing_bound_app}

We now derive Theorem \ref{main_theorem} from the main text.

To implement the certification protocol, the target and trap circuits are run using encoded logical qubits on the specified quantum device, and the measured outputs of the trap circuits are used to certify the target computation in the following way.
The logical twirling scheme described in Section \ref{log_circ_twirling} transforms general logical noise into stochastic Pauli noise. %, and so the state generated by running one of the trap circuits or the target circuit may be decomposed as
For each possible circuit position $i$, the output state that would be generated if the target circuit were executed in that
position may be decomposed as
\begin{equation}
    \label{eqaccreditation-1}
    \rho_{out}^{(i)}=(1-p_{err}^{(i)})\rho_{out,id}^{(i)} + p_{err}^{(i)} \rho_{noisy}^{(i)},
    \end{equation}
where $p_{\mathrm{err}}^{(i)}$ is the probability that errors affecting execution position $i$ would have a non-trivial effect on the target circuit if it were run in that position, and
$\rho_{\mathrm{noisy}}^{(i)}$ is the corresponding target output
state when such errors occur.
The gate operations of the ideal target circuit, $C_T$, may be
written as a sequence of $D$ logical gate layers,
\begin{equation}
    C_T=\prod_{j=1}^{D}L_j,
\end{equation}
where $L_j$ denotes the $j$-th logical gate layer of the target
circuit. Now, including the logical circuit noise associated with
execution position $i$, the previous term becomes
\begin{equation}
    \widetilde{C}_T^{(i)} = \mathcal{E}_{D+1}^{(i)} \cdot \left(  \prod_{j=1}^{D}  \mathcal{E}_{j}^{(i)}L_j  \right)  \cdot  \mathcal{E}_{0}^{(i)},
\end{equation}
where $\mathcal{E}_{j}^{(i)}$ denotes the logical effective Pauli noise channel associated with logical gate layer $\mathcal{L}_j$.
The noise channels $\mathcal{E}_{0}^{(i)}$ and $\mathcal{E}_{D+1}^{(i)}$ denote logical state preparation and measurement noise, respectively.
Noise arising from the use of noisy magic states and from incorrect decoding during cycles of error correction are included in the set of gate layer noise channels $\{\mathcal{E}_{j}^{(i)}\}_j$.
Under the stated noise assumptions, the same distribution of
position-dependent noise channels affects a trap circuit executed
in position $i$, although the resulting errors may propagate
differently through the trap circuit than the target circuit.
The value $p_{\mathrm{err}}^{(i)}$ therefore depends on the noise
affecting execution position $i$ and on its effect on the target
circuit.

After the randomly ordered target and trap circuits are run on the quantum device and the measured outputs are recorded, the outputs of the trap circuits are classically postprocessed to compute an upper bound on the error of the target circuit output.
Specifically, they are used to bound the TVD between the experimentally sampled target circuit output distribution, $\mathcal{D}_{exp}$, and the ideal distribution, $\mathcal{D}_{ideal}$.
If the states $\rho_{out}$ and $\rho_{out,id}$ are measured in the computational basis, for which the measurement projection operators are $\{\Pi_s\}_s$ where $\Pi_s = \ket{s}\bra{s}$ for $s \in \{0,1\}^n$, the TVD of the output distributions can be expressed as 
\begin{equation}
\begin{split}
\delta(\mathcal{D}_{exp},\mathcal{D}_{ideal}) &=  \frac{1}{2}\sum_{s \in \{0,1\}^n} \big|\text{Tr}\big[(\rho_{out} - \rho_{out,id})\Pi_s\big]\big|\\
&=\frac{1}{2}\sum_{s \in \{0,1\}^n} |p_{exp}(s) - p_{ideal}(s)|.
\end{split}
\end{equation} 
Let the
probability that the $i$-th trap circuit returns a bit string other
than $0^n$ be denoted by $p_{\mathrm{inc}}^{(i)}$. Let
$p_{\mathrm{canc}}^{(i)}$ denote the conditional probability that
errors which would have a non-trivial effect on the target circuit
cancel in the $i$-th trap circuit, and let
$p_{\mathrm{det}}^{(i)}$ denote the probability that the remaining
non-trivial error is detected.
The probability that a trap circuit does not output the bit string $0^n$, i.e. that errors occur and are detected, is lower-bounded by the probability that errors occur in a trap circuit, do not cancel, and are detected, so that 
\begin{equation} \label{p_inc_bound}
p_{inc}^{(i)} \geq p_{err}^{(i)}\cdot(1 - p_{canc}^{(i)})\cdot p_{det}^{(i)},
\end{equation}
which may be rearranged to 
\begin{equation} 
p_{err}^{(i)} \leq \frac{p_{inc}^{(i)}}{(1 - p_{canc}^{(i)})\cdot p_{det}^{(i)}}.
\end{equation}
The TVD can be upper-bounded using the inequality: $\delta(\mathcal{D}_{exp}^{(i)},\mathcal{D}_{ideal}^{(i)}) \leq D(\rho_{out}^{(i)},\rho_{out,id}^{(i)})$, where $D(.,.)$ is the trace distance.
The trace distance between the experimental and ideal target
output states, if the target circuit were executed in position
$i$, is
\begin{equation}
\begin{split}
D(\rho_{out}^{(i)},\rho_{out,id}^{(i)}) &= \frac{1}{2}\text{Tr}(|\rho_{out}^{(i)} - \rho_{out,id}^{(i)}|) \\
&= \frac{1}{2}\text{Tr}(|(1 - p_{err}^{(i)})\rho_{out,id}^{(i)}\\
&\hspace{3.5em}+ p_{err}^{(i)} \rho_{noisy}^{(i)} - \rho_{out,id}^{(i)}|) \\
&= \frac{1}{2}p_{err}^{(i)} \text{Tr}(|\rho_{noisy}^{(i)} - \rho_{out,id}^{(i)}|) \\
\end{split}
\end{equation}
And because $\frac{1}{2}\text{Tr}(|\rho_{noisy}^{(i)} - \rho_{out,id}^{(i)}|) \leq 1$, it follows that \begin{equation}
\label{eq:positionwise_trace_distance_bound}D(\rho_{out}^{(i)},\rho_{out,id}^{(i)})  \leq p_{err}^{(i)}.
\end{equation}
% Therefore, the TVD bounded of the $i$-th circuit is bounded by
Therefore, the TVD that the target circuit would have if executed
in position $i$ is bounded by 
 \begin{equation}
     \label{eqtvd}
\delta(\mathcal{D}_{exp}^{(i)},\mathcal{D}_{ideal}^{(i)}) \leq \frac{p_{inc}^{(i)}}{(1 - p_{canc}^{(i)})\cdot p_{det}^{(i)}}.
 \end{equation}

The noise channels affecting different circuit runs can be different, and are assumed to be independent, as stated in assumption A3.
The output of the trap circuits may be used to upper-bound the average total circuit error rate over the ensemble of possible noise behaviours affecting the target and trap circuits.
The target circuit can be viewed, when averaged over the uniformly random circuit ordering, as subject to the distribution of all noise behaviours, with an overall error rate equal to the expected error rate of the distribution. 
Therefore, the upper bound on the total circuit error rate computed from the trap circuits may be applied to the target circuit.

Let $\beta$ be a uniform upper bound on the cancellation probabilities and $\omega$ be a uniform lower bound on detection probabilities, such that $p_{canc}^{(i)}\leq\beta$ and $p_{det}^{(i)}\geq \omega$ for every trap circuit index $i$.
Let $R\in\{1,\ldots,M+1\}$ denote the uniformly random
physical execution position assigned to the target circuit. For
each physical execution position $r$, let
$p_{\mathrm{inc}}^{(r)}$ denote the probability that a trap
circuit would return a bit string other than $0^n$ if it were
executed in position $r$, under the noise associated with that
position. Define the target output distribution averaged over the
uniformly random circuit ordering as
$ \overline{\mathcal{D}}_{\mathrm{exp}}
    :=
    \mathbb{E}_{R}\!\left[
        \mathcal{D}_{\mathrm{exp}}^{(R)}
    \right]$.
By convexity of the TVD and the bound in
Eq.~\eqref{eqtvd}, we have
\begin{align}
\delta\!\left(
    \overline{\mathcal{D}}_{\mathrm{exp}},
    \mathcal{D}_{\mathrm{ideal}}
\right)
&\leq
\mathbb{E}_{R}\!\left[
    \delta\!\left(
        \mathcal{D}_{\mathrm{exp}}^{(R)},
        \mathcal{D}_{\mathrm{ideal}}
    \right)
\right]
\nonumber\\
&\leq
\frac{
    \mathbb{E}_{R}\!\left[
        p_{\mathrm{inc}}^{(R)}
    \right]
}{
    (1-\beta)\omega
}
\nonumber\\
&=
\frac{1}{(1-\beta)\omega}
\mathbb{E}_{R}\!\left[
    \frac{1}{M}
    \sum_{r\neq R}
    p_{\mathrm{inc}}^{(r)}
\right].
\label{eq:target_tvd_random_order}
\end{align}
The final equality follows because
\begin{equation}
\mathbb{E}_{R}\!\left[
    p_{\mathrm{inc}}^{(R)}
\right]
=
\frac{1}{M+1}
\sum_{r=1}^{M+1}p_{\mathrm{inc}}^{(r)}
=
\mathbb{E}_{R}\!\left[
    \frac{1}{M}
    \sum_{r\neq R}p_{\mathrm{inc}}^{(r)}
\right].
\end{equation}
Thus, the expected target error probability and the expected mean
trap failure probability are averages over the same ensemble of
position-dependent noise behaviours.
For notational simplicity, we will henceforth write
$\mathcal{D}_{\mathrm{exp}}
:=\overline{\mathcal{D}}_{\mathrm{exp}}$, so that all subsequent
target output distributions and probabilities include the uniform
randomisation of the target execution position.

For a realised target circuit position
$R\in\{1,\ldots,M+1\}$, let
$\{X_r\}_{r\neq R}$ denote the outcomes of the $M$ trap circuit
runs, where $X_r=0$ indicates that the all-zero string $0^n$ is
measured and $X_r=1$ indicates that any other bit string is
measured. Conditioned on $R$, the variables
$\{X_r\}_{r\neq R}$ are independent Bernoulli random variables
satisfying $X_r\sim\operatorname{Bern}
    (p_{\mathrm{inc}}^{(r)}),
    $ and $
    \mathbb{E}[X_r\mid R]
    =
    p_{\mathrm{inc}}^{(r)}.$
Here, $p_{\mathrm{inc}}^{(r)}$ is the probability that a trap
circuit executed in physical position $r$ returns a bit string
other than $0^n$.
Define the empirical trap failure probability by
\begin{equation}
    \widehat{p}_{\mathrm{inc}}
    :=
    \frac{1}{M}
    \sum_{r\neq R}X_r
    =
    \frac{N_{\mathrm{inc}}}{M},
\end{equation}
and define the corresponding mean trap failure probability by
$ \mu_R := \frac{1}{M} \sum_{r\neq R}p_{\mathrm{inc}}^{(r)}$.
It follows that $\mathbb{E}\!\left[
        \widehat{p}_{\mathrm{inc}}
        \,\middle|\,R
    \right]
    =
    \mu_R$.
By Hoeffding's inequality,
\begin{equation}
    \Pr\!\left(
        \left|
            \widehat{p}_{\mathrm{inc}}-\mu_R
        \right|
        \leq\epsilon
        \,\middle|\,R
    \right)
    \geq
    1-2e^{-2\epsilon^2M}.
\end{equation}
Because this bound holds for every value of $R$, it also holds
after averaging over the uniformly random target position.
Define the position-averaged trap circuit failure probability appearing in
Eq.~\eqref{eq:target_tvd_random_order} by
\begin{equation}
    \overline{p}_{\mathrm{inc}}
    :=
    \mathbb{E}_R\!\left[
        p_{\mathrm{inc}}^{(R)}
    \right]
    =
    \frac{1}{M+1}
    \sum_{r=1}^{M+1}p_{\mathrm{inc}}^{(r)}.
\end{equation}
For every realised value of $R$,
\begin{align}
    \overline{p}_{\mathrm{inc}}-\mu_R
    &=
    \frac{
        p_{\mathrm{inc}}^{(R)}-\mu_R
    }{
        M+1
    },
\end{align}
and therefore
\begin{equation}
    \left|
        \overline{p}_{\mathrm{inc}}-\mu_R
    \right|
    \leq
    \frac{1}{M+1}.
\end{equation}
Consequently, with confidence at least $1-\alpha$, it holds that
   $\label{eq:finite_sampling_pinc_bound}\overline{p}_{\mathrm{inc}}
    \leq
    \widehat{p}_{\mathrm{inc}}
    +\epsilon
    +\frac{1}{M+1}$,
provided that
\begin{equation}
    M
    \geq
    \frac{1}{2\epsilon^2}
    \log\!\left(\frac{2}{\alpha}\right).
\end{equation}
Thus, the above condition on $M$ is sufficient for
$\widehat{p}_{\mathrm{inc}}$ to estimate $\mu_R$ to within
$\epsilon$ with confidence at least $1-\alpha$; the additional
term $1/(M+1)$ accounts for the execution position occupied by
the target circuit and therefore omitted from the observed trap
mean.

The probability that a single error affects a trap circuit and is detected, $p_{det}$, is lower-bounded according to the following lemma:
\begin{lemma}\label{single_detection_lemma}
    Any single logical Pauli error of arbitrary weight occurring at any single time-step during a randomly generated trap circuit is detected with probability $p_{det} \geq 1/2$.
\end{lemma}
This result is derived in Appendix \ref{trap_detect_errors}.

Cancellation effects arising from multiple errors may be treated using one of the following three lemmas derived in Appendix \ref{bounds_on_prob_of_error_canc}.
In Lemma \ref{neglect_canc}, the probability of error cancellation is neglected.
In Lemma \ref{only_J_canc_bound}, the probability of error cancellation from errors affecting multi-qubit Clifford gates and gates performed using magic states is bounded, and any other possible error cancellation is neglected.
In Lemma \ref{bounding_all_canc}, the probability of error cancellation for any collection of errors is bounded.
Lemmas \ref{neglect_canc} and \ref{only_J_canc_bound} use the original trap construction.
Lemma \ref{bounding_all_canc}, however, uses a modified trap circuit construction, this is described in Appendix \ref{altered_trap_construction}.

\begin{lemma} \label{neglect_canc}
    In the Markovian noise regime where $N_{\mathcal{E}} \cdot q_{max} \ll 1$, the probability of error cancellation is bounded
    \begin{equation} 
    p_{err}p_{canc} \leq O((N_{\mathcal{E}} \cdot q_{max})^2),
    \end{equation}    
    where $N_{\mathcal{E}}$ is the number of logical noise channels affecting the circuit, and $q_{max}$ is the highest total channel error rate of the noise channels. 
    In this regime, we use $\beta=0$ as a leading-order approximation, with the resulting additive approximation error bounded as stated below.
    If this $\beta$ value is used to compute the TVD upper bound, $\gamma$, the magnitude of the resulting approximation error is bounded by 
    \begin{equation} \label{approx_error_bound1}
    |\epsilon_{\gamma}| \leq O((N_{\mathcal{E}} \cdot q_{max})^2).
    \end{equation}
\end{lemma}

This result is derived in Appendix \ref{neglect_error_canc_}.

\begin{lemma} \label{only_J_canc_bound}
    In the regime where the dominant source of noise is due to the $\mathcal{J}$-type gate layers, the probability of error cancellation is bounded by
    \begin{equation} 
    p_{err}p_{canc} \leq 1/2p_{err} + O\big(N_{\mathcal{J}} \cdot q_{1,\mathcal{J},max} \cdot q_{\mathcal{C},max} \big),
    \end{equation}  
    where the noise affecting the $\mathcal{J}$-type gate layers can be non-Markovian, and the rest of the noise is assumed to be Markovian.
    Here $N_{\mathcal{J}}$ denotes the number of logical noise channels acting on $\mathcal{J}$-type gate layers,
    $q_{1,\mathcal{J},max}$ the highest total channel error rate among these, 
    $N_{\mathcal{C}}$ the number of logical noise channels acting on single-qubit Clifford gate layers, state preparation, and measurement,
    and $q_{\mathcal{C},max}$ the highest total channel error rate among these.
    In this noise regime, we use $\beta=1/2$ as a leading-order approximation, with the resulting additive approximation error bounded as stated below.
    If this $\beta$ value is used to compute the TVD upper bound, $\gamma$, the magnitude of the resulting approximation error is bounded by
    \begin{equation} \label{approx_error_bound2}
    |\epsilon_{\gamma}| \leq O\big(N_{\mathcal{J}} \cdot q_{1,\mathcal{J},max} \cdot q_{\mathcal{C},max} \big).
    \end{equation}
\end{lemma}

This result is derived in Appendix \ref{J_type_cancellation_section}.

The following lemma uses a slightly different trap circuit construction, referred to as \textit{the modified trap circuit construction}, described in Appendix \ref{altered_trap_construction}.
\begin{lemma}\label{bounding_all_canc}
For the modified trap circuit construction, the probability of error cancellation in the trap circuits is upper-bounded by $p_{canc}\leq 7/8$.
\end{lemma}
This result is derived in Appendix \ref{altered_trap_construction}.
It can be shown that the same error detection bound stated in Lemma \ref{single_detection_lemma} also applies for the modified trap circuits; i.e. $p_{det} \geq 1/2$.

% So that, using Lemma \ref{single_detection_lemma}, $p_{det}\geq 1/2$, and $p_{canc} \leq \beta$, where

Combining Lemma~\ref{single_detection_lemma}, which gives $p_{\mathrm{det}}\geq 1/2$,
with the cancellation results, we use $\beta=0$ for Lemma~\ref{neglect_canc}
and $\beta=1/2$ for Lemma~\ref{only_J_canc_bound}, in each case neglecting the
corresponding additive approximation term stated above.
For Lemma~\ref{bounding_all_canc}, $\beta=7/8$ is an exact uniform conditional
upper bound. Thus, let
\begin{equation}
\beta = 
\begin{cases}
0 & \text{if Lemma \ref{neglect_canc} is used} \\
1/2 & \text{if Lemma \ref{only_J_canc_bound} is used} \\
7/8 & \text{if Lemma \ref{bounding_all_canc} is used},
\end{cases}
\end{equation}
the TVD is bounded by
\begin{equation} \label{TVD_bound_with_beta}
\delta(\mathcal{D}_{exp},\mathcal{D}_{ideal}) \leq \frac{2  p_{inc}}{1 - \beta}.
\end{equation}
The soundness of the protocol comes from upper bounding the false positive rate (i.e. $1 - (1 - p_{canc})p_{det}$) with a soundness parameter $\epsilon_s$.
% The different values of $\beta$ and Lemma \ref{single_detection_lemma} provide soundness bounds of
The different values of $\beta$ and Lemma~\ref{single_detection_lemma} give the
corresponding leading-order soundness parameters
\begin{equation} \label{false_positive_bounds}
\epsilon_s = 
\begin{cases}
1/2 & \text{if Lemma \ref{neglect_canc} is used} \\
3/4 & \text{if Lemma \ref{only_J_canc_bound} is used} \\
15/16 & \text{if Lemma \ref{bounding_all_canc} is used}.
\end{cases}
\end{equation}
Let $\gamma:=
\min\left\{
1,\frac{2(\widehat{p}_{\mathrm{inc}}+\epsilon +\frac{1}{M+1})}{1-\beta}
\right\}$, where $\hat{p}_{inc}=N_{inc}/M$, the bound
\begin{equation}\label{to_ref_eqfin} \frac{1}{2}\sum_{s \in \{0,1\}^n} |p_{exp}(s) - p_{ideal}(s)| \leq \gamma,
\end{equation}
then applies with  
%$\epsilon$ and 
confidence $(1-\alpha)$, where the number of trap circuits satisfies the inequality $M \geq \frac{1}{2\epsilon^2}\log(\frac{2}{\alpha})$.
That is, at least $\lceil\frac{1}{2\epsilon^2}\log(\frac{2}{\alpha})\rceil$ traps are needed to ensure a confidence of at least $(1-\alpha)$ that the $\gamma$ bound is valid.
When Lemma~\ref{neglect_canc} or Lemma~\ref{only_J_canc_bound} is used, the right-hand side
of Eq.~\ref{to_ref_eqfin}, and of the subsequent state and observable
bounds, is understood up to the additive approximation
error $|\epsilon_{\gamma}|$ stated in Eq.~\ref{approx_error_bound1} or
Eq.~\ref{approx_error_bound2}, respectively. No such approximation is required
when Lemma~6 is used.
%is within $\pm \epsilon$ of the ideal value.

The accreditation bound also bounds the error in the expectation
value of an observable. Define the experimental target output state,
averaged over the uniformly random target execution position, as
\begin{equation}
\label{eq:position_averaged_output_state}
    \rho_{\mathrm{exp}}
    :=
    \mathbb{E}_{R}\!\left[
        \rho_{\mathrm{out}}^{(R)}
    \right].
\end{equation}
The ideal target output state, denoted by
$\rho_{\mathrm{out,id}}$, is independent of $R$ because the
random circuit ordering changes only the physical execution
position of the target circuit and not its ideal logical
operation.
By convexity of the trace distance,
Eq.~\eqref{eq:positionwise_trace_distance_bound}, and the same
averaging, detection, and cancellation bounds used to obtain
Eq.~\eqref{TVD_bound_with_beta}, we have
\begin{align}
D\!\left(
    \rho_{\mathrm{exp}},
    \rho_{\mathrm{out,id}}
\right)
&\leq
\mathbb{E}_{R}\!\left[
    D\!\left(
        \rho_{\mathrm{out}}^{(R)},
        \rho_{\mathrm{out,id}}
    \right)
\right]
\nonumber\\
&\leq
\mathbb{E}_{R}\!\left[
    p_{\mathrm{err}}^{(R)}
\right]
\nonumber\\
&\leq
\frac{
    2\overline{p}_{\mathrm{inc}}
}{
    1-\beta
}.
\label{eq:trace_distance_bound_with_beta}
\end{align}
Combining the $\bar{p}_{inc}$ bound and
the definition of $\gamma$ with
Eq.~\eqref{eq:trace_distance_bound_with_beta} gives
\begin{equation}
\label{eq:accredited_trace_distance_bound}
D\!\left(
    \rho_{\mathrm{exp}},
    \rho_{\mathrm{out,id}}
\right)
\leq
\gamma
\end{equation}
with confidence at least $1-\alpha$.
Therefore, for any Hermitian observable $O$, H\"older's
inequality gives
\begin{align}
\left|
    \langle O\rangle_{\rho_{\mathrm{exp}}}
    -
    \langle O\rangle_{\rho_{\mathrm{out,id}}}
\right|
&=
\left|
    \operatorname{Tr}\!\left[
        O\left(
            \rho_{\mathrm{exp}}
            -
            \rho_{\mathrm{out,id}}
        \right)
    \right]
\right|
\nonumber\\
&\leq
\|O\|_{\infty}
\left\|
    \rho_{\mathrm{exp}}
    -
    \rho_{\mathrm{out,id}}
\right\|_{1}
\nonumber\\
&=
2\|O\|_{\infty}
D\!\left(
    \rho_{\mathrm{exp}},
    \rho_{\mathrm{out,id}}
\right)
\nonumber\\
&\leq
2\gamma\|O\|_{\infty},
\label{eq:observable_expectation_bound}
\end{align}
where
$\langle O\rangle_{\rho}:=\operatorname{Tr}(O\rho)$ and
$\|O\|_{\infty}$ denotes the operator norm. For a Pauli
observable, $\|O\|_{\infty}=1$, and the expectation-value error is
therefore upper-bounded by $2\gamma$.

\section{Lower-bound for the probability that trap circuits detect any single error} \label{trap_detect_errors}

We now prove Lemma \ref{single_detection_lemma}, that states:

\textit{
Any single logical Pauli error of arbitrary weight occurring at any single time-step during a randomly generated trap circuit is detected with probability $p_{det} \geq 1/2$.
}
\vspace{0.2em}

\begin{proof}
Let $\{{\mathcal{C}}_T^{(i)}\}_i$ denote the set of all possible trap circuits for a given target circuit.
The $i$-th trap circuit affected by logical-qubit-level circuit noise, $\tilde{\mathcal{C}}_T^{(i)}$, can be decomposed into logical gate layers and noise channels, and expressed in the form
\begin{equation}
\begin{split}
\tilde{\mathcal{C}}_T^{(i)} 
&= \mathcal{E}_{D+1}\cdot \prod_{j=1}^D ( \mathcal{E}_{j} \cdot\mathcal{L}_j^{(i)}) \cdot \mathcal{E}_{0},\\
\end{split}
\end{equation}
where $\mathcal{E}_{j}$ denotes the logical Pauli noise channel associated with logical gate layer $\mathcal{L}_j^{(i)}$, $\mathcal{E}_{0}$ denotes logical state preparation noise and $\mathcal{E}_{D+1}$ logical measurement noise.
Explicitly labelling gate layers where logical Pauli twirling operations are performed, the sequence of logical gate layers of an ideal trap circuit may be written
\begin{equation}
\begin{split}
{\mathcal{C}}^{(i)}_T &= \mathcal{P}_{D+1}^{(1)} \cdot \prod_{j=1}^{D}  \big(\mathcal{P}_{j}^{(2)}  \cdot \mathcal{L}_{j}^{(i)} \cdot  \mathcal{P}_{j}^{(1)}\big) \cdot \mathcal{P}_{0}^{(2)}. \\
\end{split}
\end{equation}
The initial application of a Pauli twirling gate layer is indicated by the superscript shown in $\mathcal{P}^{(1)}$, while the subsequent Pauli gate layer that cancels it is denoted by $\mathcal{P}^{(2)}$.
As shown in Fig. \ref{fig:target_and_trap_circuit_structure}, each trap circuit is generated by randomly choosing whether to apply a layer of Hadamard gates at the beginning and end of the circuit, %i.e. $t\in\{0,1\}$, 
and randomly choosing the composition of the sandwiching gate layers.
For the $i$-th trap circuit, let the $j$-th sandwiching gate layer be denoted $\mathcal{W}_j^{(i)}\in \{S,S^{\dagger},H \}^{\otimes n}$.
Let the random logical gate layers at the beginning and end of the trap circuits be denoted by $\mathcal{H}^t = (H{}^t)^{\otimes n}$ where $t\in\{0,1\}$. 
Let $R_{\mathcal{W},l}$ denote a specific instance of sandwiching gate layers for a particular trap circuit instance, and let $\{R_{\mathcal{W},l}\}_l$ denote the set of all possible sandwiching gate layer instances.
Let $\mathcal{J}_j$ denote the $j$-th logical gate layer that can contain gates implemented using trap magic states, as well as multi-qubit Clifford gates.
As $t\in\{0,1\}$ and $R_{\mathcal{W},l} \in \{R_{\mathcal{W},l}\}_l$, there are then $|\{{\mathcal{C}}_T^{(i)}\}_i| = 2 \cdot |\{R_{\mathcal{W},l}\}_l|$ unique possible trap circuits.
The $i$-th trap circuit from the set $\{{\mathcal{C}}_T^{(i)}\}_i$, including randomly chosen Hadamard gate layers, sandwiching layers, and Pauli twirling layers, may be expanded as

\begin{widetext}
\begin{align}
{\mathcal{C}}^{(i)}_T &=  \mathcal{P}_{3K+3}^{(1)} \cdot \mathcal{P}_{3K+2}^{(2)} \cdot \mathcal{H}^t \cdot \mathcal{P}_{3K+2}^{(1)}  \cdot \prod_{k=1}^{K}   \big( \mathcal{P}_{3k+1}^{(2)} \cdot \mathcal{W}_{3k+1} \cdot \mathcal{P}_{3k+1}^{(1)} \cdot \mathcal{P}_{3k}^{(2)} \cdot \mathcal{J}_{3k} \cdot \mathcal{P}_{3k}^{(1)} 
\cdot \mathcal{P}_{3k-1}^{(2)} \cdot \mathcal{W}_{3k-1} \cdot \mathcal{P}_{3k-1}^{(1)}\big)\\
& \hspace{3em}\cdot  \mathcal{P}_{1}^{(2)} \cdot  \mathcal{H}^t \cdot \mathcal{P}_{1}^{(1)} \cdot \mathcal{P}_{0}^{(2)}. 
\end{align}
\end{widetext}
In practice, during circuit compilation any adjacent logical Pauli gate layers will be merged together. 

We now ignore Pauli twirling gate layers, as these do not change the circuit logic but rather serve to transform general logical noise into stochastic logical Pauli noise (see Appendix \ref{log_rand_comp}).
However, the noise channels associated with the Pauli gate layers can be combined with the noise channels associated with the previous gate layers, and so are included in the following analysis.
The previous expression then may be written as
\begin{equation}
\begin{split}
{\mathcal{C}}_T^{(i)}{}' &=  \mathcal{H}^t \cdot \prod_{k=1}^{K}  (  \mathcal{W}_{3k+1}
\cdot  \mathcal{J}_{3k} \cdot \mathcal{W}_{3k-1} ) \cdot \mathcal{H}^t. \\
\end{split}
\end{equation}
Each distinct set of randomising gates defines a unique trap circuit.
As $t\in\{0,1\}$ and $R_{\mathcal{W},l} \in \{R_{\mathcal{W},l}\}_l$ are both chosen with uniform probability over their respective sets, the probability of measuring $s=0^n$ in the absence of logical noise, summing over all possible trap circuits, is

\begin{widetext}
\begin{align}
\Pr(s=0^n) &= \frac{1}{2\cdot |\{R_{\mathcal{W},l}\}_l|}\sum_{t\in\{0,1\}} \sum_{R_{\mathcal{W},l} \in \{R_{\mathcal{W},l}\}_l} \bra{0^{n}}  \mathcal{H}^t \bigcirc_{k=1}^{K}  (  \mathcal{W}_{3k+1}  \circ  \mathcal{J}_{3k} \circ \mathcal{W}_{3k-1} ) \circ \mathcal{H}^t \circ (\ket{0^n}\bra{0^n})\ket{0^{n}} \\
& = 1.
\end{align}
\end{widetext}
The second line follows from the fact that all of the logical gate layers either cancel or stabilise the initial state.
The $\mathcal{J}$-type gate layers can include identity operations (from the use of trap magic states instead of target magic states when performing state injection), as well as C$Z$.
This means that $\mathcal{W}_{3k+1} \mathcal{J}_{3k}  \mathcal{W}_{3k-1} \mathcal{H}^{t}\ket{0^n} = \mathcal{H}^{t}\ket{0^n}$, for all $k \in\{1,\ldots, K\}$ and $t \in\{0,1\}$. %, in the absence of errors $p(s=0^n)=1$.
The result is then that no error is registered as having occurred during any ideal trap computation with probability 1.

We now consider noise channels acting at different time steps of the computation and establish a lower bound on the detection probability of all possible single time-step logical Pauli errors across the set of possible trap circuits.

With the inclusion of a logical state-preparation Pauli noise channel, $\mathcal{E}_0$, the probability that a randomly generated trap circuit fails to detect a Pauli error due to this channel is
\begin{widetext}
\begin{align}
\Pr(s=0^n| \mathcal{E}_0) &= \frac{1}{2\cdot|\{R_{\mathcal{W},l}\}_l|}\sum_{t\in\{0,1\}} \sum_{R_{\mathcal{W},l} \in \{R_{\mathcal{W},l}\}_l}\bra{0^{n}} \mathcal{H}^t  \bigcirc_{k=1}^{K}   (  \mathcal{W}_{3k+1} \circ \mathcal{J}_{3k}  \circ \mathcal{W}_{3k-1} )  \circ  \mathcal{H}^t \circ \mathcal{E}_0\circ (\ket{0^n}\bra{0^n}) \ket{0^{n}}. 
\end{align}
\end{widetext}
From the set of logical Pauli errors induced by the channel $ \mathcal{E}_0$, only those that have a non-trivial $X$ or $Y$ component will not stabilise the initial quantum state. 
When a logical Pauli error of this type occurs it follows that $p(s=0^n| P_0)= \bra{0^{n}}P_0 \ket{0^{n}}=0$, where $P_0$ is a Pauli error with a non-trivial $X$ or $Y$ component.
This means that any single Pauli error that non-trivially changes the prepared computational state is detected with probability 1.
Similar arguments can be made regarding logical measurement Pauli noise, as such errors affecting the logical measurement operation at the end of the circuit may also be detected with probability 1.

We now consider logical noise affecting the first Hadamard gate layer.
The probability a randomly generated trap circuit not detecting a logical Pauli error due to this type of noise is
\begin{widetext}
\begin{align}
\Pr(s=0^n| \mathcal{E}_1) &= \frac{1}{2\cdot|\{R_{\mathcal{W},l}\}_l|}\sum_{t\in\{0,1\}} \sum_{R_{\mathcal{W},l} \in \{R_{\mathcal{W},l}\}_l} \bra{0^{n}}  \mathcal{H}^t \bigcirc_{k=1}^{L}   (  \mathcal{W}_{3k+1} \circ \mathcal{J}_{3k}  \circ \mathcal{W}_{3k-1} ) \circ \mathcal{E}_1  \circ  \mathcal{H}^t \circ (\ket{0^n}\bra{0^n}) \ket{0^{n}} \\
&= \frac{1}{2\cdot|\{R_{\mathcal{W},l}\}_l|}\sum_{t\in\{0,1\}} \sum_{R_{\mathcal{W},l} \in \{R_{\mathcal{W},l}\}_l} \bra{0^{n}}  \mathcal{H}^t \circ \mathcal{E}_1  \circ  \mathcal{H}^t \circ (\ket{0^n}\bra{0^n}) \ket{0^{n}} \\
&\leq\frac{1}{2}.
\end{align}
\end{widetext}
The inequality of the final line is obtained as follows.
Since the parameter $t$ is either $0$ or $1$, each with probability $1/2$, if a Pauli error $P_1$ occurs with a non-trivial $Y$ component, it is detected with probability $1$.
If $P_1$ has a non-trivial $X$ component, but a trivial $Y$ and $Z$ component, it is detected with probability $1/2$.
Likewise, if $P_1$ has a non-trivial $Z$ component, but a trivial $X$ and $Y$ component, it is detected with probability $1/2$.
This then means that any Pauli error is detected with probability at least $1/2$.
Similar arguments can be applied to Pauli noise affecting the final logical Hadamard gate layer.

We now consider noise affecting one of the $\mathcal{W}$-type gate layers.
The probability that noise affects the gate layer $\mathcal{W}_{3M-1}$ and is not detected is
\begin{widetext}
\begin{align}
\Pr(s=0^n | \mathcal{E}_{3M-1}) 
&= \frac{1}{2\cdot|\{R_{\mathcal{W},l}\}_l|}\sum_{t\in\{0,1\}} \sum_{R_{\mathcal{W},l} \in \{R_{\mathcal{W},l}\}_l} \bra{0^{n}}   \mathcal{H}^t \bigcirc_{k=M+1}^{L}   (  \mathcal{W}_{3k+1} \circ \mathcal{J}_{3k}  \circ \mathcal{W}_{3k-1} )\\
& \hspace{2em} \circ (  \mathcal{W}_{3M+1} \circ \mathcal{J}_{3M} \circ \mathcal{E}_{3M-1}  \circ \mathcal{W}_{3M-1} ) \bigcirc_{k=1}^{M-1}   (  \mathcal{W}_{3k+1} \circ \mathcal{J}_{3k}  \circ \mathcal{W}_{3k-1} )\\
& \hspace{2em}\circ  \mathcal{H}^t \circ (\ket{0^n}\bra{0^n}) \ket{0^{n}} \\
&= \frac{1}{2\cdot|\{R_{\mathcal{W},l}\}_l|}\sum_{t\in\{0,1\}} \sum_{R_{\mathcal{W},l} \in \{R_{\mathcal{W},l}\}_l} \bra{0^{n}}   \mathcal{H}^t  \bigcirc_{k=M}^{L}   (  \mathcal{W}_{3k+1} \circ \mathcal{J}_{3k}  \circ \mathcal{W}_{3k-1} )\\
& \hspace{2em} \circ \mathcal{E}_{3M-1}{}'   \bigcirc_{k=1}^{M-1}   (  \mathcal{W}_{3k+1} \circ \mathcal{J}_{3k}  \circ \mathcal{W}_{3k-1} )\circ  \mathcal{H}^t \circ (\ket{0^n}\bra{0^n}) \ket{0^{n}} \\
&= \frac{1}{2\cdot|\{R_{\mathcal{W},l}\}_l|}\sum_{t\in\{0,1\}} \sum_{R_{\mathcal{W},l} \in \{R_{\mathcal{W},l}\}_l}\bra{0^{n}}   \mathcal{H}^t  \circ \mathcal{E}_{3M-1}{}' \circ  \mathcal{H}^t \circ (\ket{0^n}\bra{0^n}) \ket{0^{n}} \\
&\leq\frac{1}{2}.
\end{align}
\end{widetext}
To get the second line, we have that since $\mathcal{W}_{3M-1}\in \{S,S^{\dagger},H \}^{\otimes n}$, the following equality holds: $ \mathcal{E}_{3M-1}\mathcal{W}_{3M-1} = \mathcal{W}_{3M-1}\mathcal{E}_{3M-1}{}'$, where $\mathcal{E}_{3M-1}{}'$ is another Pauli channel.
The inequality of the final line is obtained through similar reasoning as was used for the case of noise affecting the logical Hadamard gate layers.
So that any error occurring as a result of the noise channel $\mathcal{E}_{3M-1}{}'$ is detected with probability at least $1/2$.
Similar arguments can be made regarding noise affecting the $\mathcal{J}$-type gate layers.
Therefore, any single logical Pauli error occurring at any single time-step during the execution of a trap circuit is detected with probability $p_{det} \geq 1/2.$ 
\end{proof}

\section{Upper bounds for the probability of error cancellation in trap circuits} \label{bounds_on_prob_of_error_canc}

We now bound the probability that errors act on a randomly generated trap circuit but cancel, thereby becoming undetectable.
Together with the bound on the trap circuit detection probability for individual errors derived in the previous section, this provides an overall bound on the \textit{soundness} of the framework.
That is, a lower bound on the detection probability for any collection of errors affecting a randomly generated trap circuit; see eqn. \ref{TVD_bound_with_beta}.

We consider three different scenarios, deriving upper bounds on error cancellation probability in the trap circuits, $p_{canc}$, for each. 
The three scenarios considered are:
\begin{enumerate}
    \item The probability of logical error cancellation is sufficiently small that cancellation effects can be reasonably neglected.
    \item The logical error rates for gates implemented via magic states are significantly larger than those for performing Clifford gates, state preparation, or measurement. 
    The dominant cancellation effects are then from errors affecting gates performed using magic states, while other cancellation effects can be reasonably neglected.
    \item No cancellation effects are neglected.
\end{enumerate}
Scenario (1) is appropriate when it is assumed that the noise channels affecting circuits are Markovian, with no time-correlated noise effects.
Scenario (2) is appropriate when it is assumed that the noise from $\mathcal{J}$-type gate layers, i.e. the noise from multi-qubit Clifford gates and gates requiring magic states, is non-Markovian, while noise from single-qubit Clifford gates is Markovian.
Or, alternatively, where the dominant logical noise comes from $\mathcal{J}$-type gate layers, and so the cancellation effects of other lesser sources of noise can reasonably be ignored.
In Scenario (3), no assumptions need to be made about the Markovianity of the noise.
Thus, Scenario (1) considers complete Markovianity, Scenario (2) partial Markovianity and part non-Markovianity, and Scenario (3) potentially complete non-Markovianity.

The Markovianity of the logical noise is relevant to cancellation effects because Pauli twirled non-Markovian noise results in classically correlated stochastic Pauli channels \cite{liu_non-markovian_2024, winick_concepts_2022}.
Time-correlated logical Pauli errors may cause systematic cancellation in trap circuits, and must therefore be accounted for to ensure an accurate TVD bound.
In the following analysis, we show that for Scenario (2) the randomisation of the trap circuits means that the probability of errors affecting $\mathcal{J}$-type layers and cancelling can be upper-bounded.
For Scenario (3), we show that, using a modified version of the trap circuit construction, the probability of any errors affecting a trap circuit and cancelling is upper-bounded by a constant.

All of the error cancellation probability upper bounds derived are constant.
In the Appendix \ref{neglect_error_canc_} analysis, Scenario (1) provides the smallest upper bound on the probability of error cancellation, but this relies on the assumption that errors do not cancel in the trap circuits.
The next smallest upper bound value is derived in Appendix \ref{J_type_cancellation_section} for Scenario (2), where the probability of error cancellation due to errors affecting gates performed using magic states as well as C$Z$ gates is bounded, and the assumption is required that there are no cancellation effects due to errors from other components of the trap circuit.
No assumptions are made regarding error cancellation for the Scenario (3) analysis presented in Appendix \ref{altered_trap_construction}; however, the price for this is a slightly larger upper bound.
The logical randomised compiling described in Appendix \ref{log_rand_comp} transforms general logical noise to effective logical stochastic Pauli noise. 
In the following error cancellation analysis, the effective logical stochastic Pauli channels resulting from twirling are propagated together through the trap circuits to form a single Pauli channel.
A bound on the trap circuit error cancellation probability may be derived by analysing the components of the error channels that cancel when the channels are combined.

\subsection{Neglecting error cancellation} \label{neglect_error_canc_}

In the case of Markovian noise, if the total error rate of each logical noise channel is significantly smaller than one, the probability of error cancellation is negligible.
The expected number of logical errors when running a circuit, $N_e$, is bounded $N_e \leq N_{\mathcal{E}}q_{max}$, where $N_{\mathcal{E}}$ is the number of logical noise channels affecting the circuit, and $q_{max}$ is the maximum total error rate among these channels.
Only in the regime where $ N_{\mathcal{E}}q_{max}  \ll 1$ is it possible to achieve a non-trivial computational output.
In this noise regime, it is reasonable to assume that error cancellation occurs with sufficiently low probability that it can be neglected in the computation of the TVD bound.
% This is formalised in the following lemma.

% We now prove Lemma \ref{neglect_canc}, that states:

% \textit{
% In the Markovian noise regime where $N_{\mathcal{E}} \cdot q_{max} \ll 1$, the probability of error cancellation is bounded
% \begin{equation} 
% p_{err}p_{canc} \leq O((N_{\mathcal{E}} \cdot q_{max})^2),
% \end{equation}    
% where $N_{\mathcal{E}}$ is the number of logical noise channels affecting the circuit, and $q_{max}$ is the highest total channel error rate of the noise channels. 
% In this regime, it may be assumed that $p_{canc} \leq \beta$, where $\beta=0$, with the resulting additive approximation error bounded as stated below.
% If this $\beta$ value is used to compute the TVD upper bound, $\gamma$, the magnitude of the resulting approximation error is bounded by 
% \begin{equation}
% |\epsilon_{\gamma}| \leq O((N_{\mathcal{E}} \cdot q_{max})^2).
% \end{equation}
% }

\begin{proof}[Proof of Lemma~\ref{neglect_canc}]

If noise is Markovian and there are $N_{\mathcal{E}}$ stochastic logical Pauli channels affecting a trap circuit, the set of total error rates for these channels is $\{q_{tot}(j)\}_{j=1}^{N_{\mathcal{E}}}$,
and the highest total error rate is $q_{max}:= \max(\{q_{tot}(j)\}_{j=1}^{N_{\mathcal{E}}})$.
The probability of error cancellation is upper-bounded, such that
\begin{equation} 
p_{err}p_{canc} \leq \sum_{j=2}^{N_{\mathcal{E}}}{\binom{N_{\mathcal{E}}}{j}} \cdot q_{max}^j .
\end{equation}
And as ${\binom{N_{\mathcal{E}}}{j}} \leq N_{\mathcal{E}}^j$, and $1 - q_{max} \leq 1$, it follows that
\begin{equation} 
p_{err}p_{canc} \leq \sum_{j=2}^{N_{\mathcal{E}}}(N_{\mathcal{E}} q_{max})^j.
\end{equation}
Now, as $N_{\mathcal{E}}\in \mathbb{Z}^+$ and $N_{\mathcal{E}}\gg 1$, the condition that $N_{\mathcal{E}} q_{max} \ll 1$ implies $q_{max} \ll 1$.
And so the leading order contribution to the previous sum is from the cancellation of pairs of errors, corresponding to the $j=2$ term, with the probability that two errors cancel upper-bounded by $(N_{\mathcal{E}} q_{max})^2$.
This means that
\begin{equation} \label{fdete}
p_{err}p_{canc} \leq O((N_{\mathcal{E}} q_{max})^2).
\end{equation}
This bound will generally be loose, since not every
multiple-error configuration cancels. Using Eq.~\ref{p_inc_bound},
$p_{\mathrm{det}}\geq 1/2$, and Eq.~\ref{fdete}, we obtain
\begin{align}
p_{\mathrm{inc}}
&\geq
p_{\mathrm{err}}
\bigl(1-p_{\mathrm{canc}}\bigr)
p_{\mathrm{det}}
\nonumber\\
&=
\bigl(
p_{\mathrm{err}}
-
p_{\mathrm{err}}p_{\mathrm{canc}}
\bigr)
p_{\mathrm{det}}
\nonumber\\
&\geq
\frac{1}{2}
\left[
p_{\mathrm{err}}
-
O\!\left((N_E q_{\max})^2\right)
\right].
\end{align}
Therefore,
\begin{equation}
\label{<retain-the-current-E6-label>}
p_{\mathrm{err}}
\leq
2p_{\mathrm{inc}}
+
O\!\left((N_E q_{\max})^2\right).
\end{equation}
Thus, using $\beta=0$ introduces an additive approximation
error satisfying
\begin{equation}
|\epsilon_{\gamma}|
\leq
O\!\left((N_E q_{\max})^2\right).
\end{equation}
\end{proof}

\subsection{\texorpdfstring{Bounding the probability of error cancellation for errors affecting solely  $\mathcal{J}$-type gate layers in a trap circuit}{Bounding the probability of error cancellation for errors affecting J-type gate layers in a trap circuit}}

\label{J_type_cancellation_section}

In the target and trap circuits, the $\mathcal{J}$-type gate layers are composed of gates performed by the consumption of magic states and C$Z$ gates.
In early fault-tolerant computation, gates performed using imperfectly purified magic states are likely to have considerably higher logical error rates than those of the rest of the circuit operations.
In this scenario, the most significant error cancellation effects come from noise affecting the $\mathcal{J}$-type gate layers.
We now show that the probability of error cancellation due to errors affecting $\mathcal{J}$-type gate layers in the trap circuits can be upper-bounded by a constant.
In the following analysis, the noise affecting $\mathcal{J}$-type gate layers can be Markovian or non-Markovian, and the rest of the noise is assumed to be Markovian.

The bound on the error cancellation probability from noise affecting $\mathcal{J}$-type gate layers is formalised in the following lemma:
\begin{lemma}\label{cancellation_lemma2}
The probability that errors affecting $\mathcal{J}$-type gate layers in a randomly generated trap circuit cancel is upper-bounded $p_{canc}\leq 1/2$.
\end{lemma}

\begin{proof}
A trap circuit in which two nearest-neighbour $\mathcal{J}$-type gate layers are affected by errors may be expressed in the form:
\begin{widetext}
\begin{equation}
\begin{split}
\mathcal{H}^{t} \circ (\circ^m_{j=l+2}&\mathcal{W}_{3j} \mathcal{J}_{3j-1} \mathcal{W}_{3j-2}) \circ (\mathcal{W}_{3(l+1)} P_{3(l+1)-1} \mathcal{J}_{3(l+1)-1} \mathcal{W}_{3(l+1)-2}) \circ (\mathcal{W}_{3l} P_{3l-1} \mathcal{J}_{3l-1} \mathcal{W}_{3l-2}) \\
&\circ (\circ^{l-1}_{j=1}\mathcal{W}_{3j} \mathcal{J}_{3j-1} \mathcal{W}_{3j-2}) \circ \mathcal{H}^{t}  ,  \\
\end{split}
\end{equation}
\end{widetext}
where $P_{3(l+1)-1}$ is a logical Pauli error affecting the gate layer  $\mathcal{J}_{3(l+1)-1}$, and $P_{3l-1}$ is a logical Pauli error affecting the gate layer $\mathcal{J}_{3l-1}$.
Let $\mathcal{J}_{3j-1}{}' := \mathcal{W}_{3j} \mathcal{J}_{3j-1} \mathcal{W}_{3j-2}$, where $\mathcal{J}_{3j-1}{}'$ is a logical gate layer of randomly oriented CNOT gates and identity operations.
Using this identity, and the fact that $\mathcal{W}_{3k} =  \mathcal{W}_{3k-2}^{\dagger}$ for $k \in\{1, \ldots, m\}$, we can rewrite the previous expression as:
\begin{widetext}
\begin{equation}
\begin{split}
\mathcal{H}^{t} \circ& (\circ^m_{j=l+2} \mathcal{J}_{3j-1}{}') \circ (\mathcal{W}_{3(l+1)} P_{3(l+1)-1} \mathcal{W}_{3(l+1)-2} \mathcal{J}_{3(l+1)-1}{}') \circ (\mathcal{W}_{3l} P_{3l-1}  \mathcal{W}_{3l-2} \mathcal{J}_{3l-1}{}') \circ (\circ^{l-1}_{j=1} \mathcal{J}_{3j-1}{}') \circ \mathcal{H}^{t}  .  \\
\end{split}
\end{equation}
\end{widetext}
Since conjugation of each Pauli error by $\mathcal{W}$-type gate layers just maps it to another Pauli operator, the previous expression becomes:
\begin{widetext}
\begin{equation}
\begin{split}
\mathcal{H}^{t} \circ& (\circ^m_{j=l+2} \mathcal{J}_{3j-1}{}') \circ P_{3(l+1)-1}' \circ \mathcal{J}_{3(l+1)-1}{}' \circ  P_{3l-1}'  \circ \mathcal{J}_{3l-1}{}' \circ (\circ^{l-1}_{j=1} \mathcal{J}_{3j-1}{}') \circ \mathcal{H}^{t}  .  \\
\end{split}
\end{equation}
\end{widetext}
where $P_{3(l+1)-1}' = \mathcal{W}_{3(l+1)} P_{3(l+1)-1} \mathcal{W}_{3(l+1)-2}$ and $P_{3l-1}' = \mathcal{W}_{3l} P_{3l-1}  \mathcal{W}_{3l-2}$.
The mapping caused by the random $\mathcal{W}$-type gate layer conjugation may be understood in terms of the $H$-conjugation Pauli identities:
\begin{equation}
\begin{split}
HXH &:= Z, \\
HYH &:= -Y, \\
HZH &:= X, \\
\end{split}
\end{equation}
and the $S$-conjugation Pauli identities:
\begin{equation}
\begin{split}
SXS^\dagger &:= Y, \\
SYS^\dagger &:= -X, \\
SZS^\dagger &:= Z. \\
\end{split}
\end{equation}
The randomly chosen $\mathcal{W}$-type gate layers ensure that each Pauli error is mapped to a random new Pauli operator.
Since we only care about the error cancellation probability, we now consider only the Pauli operators and the gate layer separating them, this is the sequence:
\begin{equation}
\begin{split}
 P_{3(l+1)-1}' \circ \mathcal{J}_{3(l+1)-1}{}' \circ  P_{3l-1}' .  \\
\end{split}
\end{equation}
The Pauli error $P_{3(l+1)-1}{}'$ is now propagated through the Clifford gate layer $\mathcal{J}_{3(l+1)-1}{}'$, leading to the expression
\begin{equation}
\begin{split}
\mathcal{J}_{3(l+1)-1}{}' \circ   P_{3(l+1)-1}'' \circ P_{3l-1}',  \\
\end{split}
\end{equation}
where %$P_{3(j+1)-1}{}' \mathcal{J}_{3(l+1)-1}{}' = \mathcal{J}_{3(l+1)-1}{}' P_{3(j+1)-1}{}''$,
\begin{equation}
P_{3(l+1)-1}{}'' := (\mathcal{J}_{3(l+1)-1}{}')^{\dagger} \circ P_{3(l+1)-1}{}' \circ \mathcal{J}_{3(l+1)-1}{}',
\end{equation}
and $P_{3(l+1)-1}{}''$ is again a Pauli operator.
The two Pauli errors cancel if
\begin{equation}
\begin{split}
P_{3(l+1)-1}'' \cdot P_{3l-1}' \propto I^{\otimes n},  \\
\end{split}
\end{equation}
where the notation `$\propto$' indicates equality up to a phase.

In the case of a single logical qubit, the Pauli errors $P_{3(j+1)-1}, P_{3j-1}\in\{X,Y,Z\} $ are mapped to new Pauli operators under random conjugation by $\mathcal{W}$-type gate layers with probability at least $1/2$.
The $\mathcal{J}$-type layer is an identity operation and thus does not affect error propagation.
The random mapping induced by the conjugation means that each Pauli error randomly becomes one of two single-qubit Pauli operators.
In the worst case, cancellation occurs for two out of the four possible error combinations, and so
\begin{equation} \label{one_q_case_cancellation}
\Pr (P_{3(j+1)-1}''P_{3j-1}' \propto I) \leq 1/2.
\end{equation}

Similar analysis can be performed for the case of two logical qubits.
Now $P_{3(j+1)-1}, P_{3j-1}\in\{I,X,Y,Z\}^{\otimes 2} \backslash I^{\otimes 2}$, and the $\mathcal{J}$-type layer can be either an $I^{\otimes 2}$ operation or a randomly oriented CNOT gate.
The conjugation by $\mathcal{W}$-type gate layer randomly maps each Pauli error to a different Pauli error with probability at least $1/2$.
And, when moving the Pauli errors together to combine them, if the $\mathcal{J}$-type layer contains a randomly-oriented CNOT gate this randomises the error propagation, mapping the Pauli error to a new Pauli error when it is moved through the CNOT gate.
This random error propagation means that after the two errors have been moved together to be in the form $P_{3(j+1)-1}''P_{3j-1}'$, both $P_{3(j+1)-1}''$ and $P_{3j-1}'$ can each be at least two different Pauli errors. 
In the worst case, both errors are effectively chosen uniformly at random from two Pauli operators, and error cancellation occurs with probability $1/2$.
It follows that
\begin{equation} \label{two_q_case_cancellation}
\Pr(P_{3(j+1)-1}''P_{3j-1}' \propto I^{\otimes 2}) \leq 1/2.
\end{equation}

For Pauli errors of weight greater than two, the probability of error cancellation may be bounded by considering the probability that all one- or two-qubit components of the errors cancel.
For each logical qubit where the $\mathcal{J}$-type gate layer acts with a local identity operation, eqn. \ref{one_q_case_cancellation} may be applied to bound the probability that the components of the errors acting on these qubits cancel.
For each pair of logical qubits where the $\mathcal{J}$-type gate layer acts with a randomly oriented CNOT operation, eqn. \ref{two_q_case_cancellation} may be applied to bound the probability that the components of the errors acting on each such pair of qubits cancel.
As all such error components must cancel for the errors to cancel overall, the probability of error cancellation is upper-bounded by $(1/2)^{a+b}$, where $a$ is the number of qubits for which the component of the $\mathcal{J}$-type layer is identity and $b$ is the number of CNOT operations in the layer. 
It follows that 
\begin{equation} \label{n_q_case_cancellation}
\Pr(P_{3(j+1)-1}''P_{3j-1}' \propto I^{\otimes n}) \leq 1/2,
\end{equation}
where $n$ is an arbitrary number of logical qubits, and the two Pauli errors can be of any non-trivial weight up to $n$.
If two errors affect $\mathcal{J}$-type gate layers that are not nearest-neighbour, the same bound can be derived using a similar approach, the only difference being that the Pauli errors must be propagated through multiple $\mathcal{J}$-type layers instead of just one.
This means that any two Pauli errors that affect the $\mathcal{J}$-type gate layers of a randomly generated trap circuit will cancel with a probability upper-bounded by $1/2$.

We now consider the case where three Pauli errors affect
$\mathcal{J}$-type gate layers of a randomly generated trap circuit.
We have just shown that the probability of cancellation of any two
errors affecting $\mathcal{J}$-type gate layers within a trap circuit
is upper-bounded by $1/2$.
Let $c=1/2$ denote this worst-case cancellation probability, and let
$q_j$ denote the worst-case probability that a collection of $j$
Pauli errors affecting $\mathcal{J}$-type gate layers cancels.
Thus, $q_2\leq c$.
If two of the three Pauli errors are propagated together, let
$p\leq c$ denote the probability that they cancel. 
If they cancel,
then the remaining non-trivial Pauli error prevents all three errors
from cancelling. With probability $1-p$, they instead combine into
a new non-trivial Pauli error. 
By the two error result, the probability
that this combined error cancels with the remaining Pauli error is
upper-bounded by $c$. 
Therefore, $q_3\leq (1-p)c\leq c$.
So we have that $q_2\leq c$ and $q_3\leq c$.
Now consider the case of four Pauli errors.
If two of the errors cancel then two errors remain, otherwise they
combine into a new non-trivial Pauli error and three errors remain.
Therefore,
\begin{equation}
\begin{split}
q_4
&\leq (1-p)q_3+p q_2\\
&\leq (1-p)c+p c\\
&\leq \frac{1}{2}.
\end{split}
\end{equation}
More generally,
\[
q_j\leq (1-p)q_{j-1}+p q_{j-2},
\]
so if $q_{j-1}\leq c$ and $q_{j-2}\leq c$, then $q_j\leq c$.
This extends iteratively to any number of Pauli errors
affecting $\mathcal{J}$-type gate layers in a trap circuit.
Therefore, for any collection of Pauli errors affecting
$\mathcal{J}$-type gate layers in a trap circuit, the cancellation
probability satisfies the inequality $p_{\mathrm{canc}}\leq 1/2$. 
\end{proof}

% We now prove Lemma \ref{only_J_canc_bound}, that states:

% \textit{
% In the regime where the dominant source of noise is due to the $\mathcal{J}$-type gate layers, the probability of error cancellation is bounded by
% \begin{equation} 
% p_{err}p_{canc} \leq 1/2p_{err} + O\big(N_{\mathcal{J}} \cdot q_{1,\mathcal{J},max} \cdot q_{\mathcal{C},max} \big),
% \end{equation}  
% where the noise affecting the $\mathcal{J}$-type gate layers can be non-Markovian, and the rest of the noise is assumed to be Markovian.
% Here $N_{\mathcal{J}}$ denotes the number of logical noise channels acting on $\mathcal{J}$-type gate layers,
% $q_{1,\mathcal{J},max}$ the highest total channel error rate among these, 
% $N_{\mathcal{C}}$ the number of logical noise channels acting on single-qubit Clifford gate layers, state preparation, and measurement,
% and $q_{\mathcal{C},max}$ the highest total channel error rate among these.
% In this noise regime, it may be assumed that $p_{canc} \leq \beta$, where $\beta=1/2$, with the remaining configurations contributing the additive approximation error stated below.
% If this $\beta$ value is used to compute the TVD upper bound, $\gamma$, the magnitude of the resulting approximation error is bounded by
% \begin{equation}
% |\epsilon_{\gamma}| \leq O\big(N_{\mathcal{J}} \cdot q_{1,\mathcal{J},max} \cdot q_{\mathcal{C},max} \big).
% \end{equation}
% }

\begin{proof}[Proof of Lemma~\ref{only_J_canc_bound}]

The noise affecting the $\mathcal{J}$-type gate layers is permitted to be non-Markovian, and so applying the logical randomised compiling scheme can result in Pauli errors that are classically correlated in time.
If there are $N_J$ logical noise channels affecting the
$J$-type gate layers of each trap circuit, then the probabilities
associated with all non-empty subsets of these layers are given by
the set $\left\{q_{J,\mathrm{tot}(j)}\right\}_{j=1}^{2^{N_J}-1}$,
where each subset specifies a combination of error locations.
Let $\mathcal{I}_i$ denote the set of indices $j$ corresponding
to subsets containing exactly $i$ $J$-type gate layers. 
The maximum probability that any set of $i \in \{1,\ldots,N_J\}$ $J$-type gate layers is affected by errors is $q_{i,J,\max}:=\max_{j \in\mathcal{I}_i}q_{J,\mathrm{tot}(j)}$.
If there are $N_{\mathcal{C}}$ Markovian logical noise channels affecting the logical single-qubit Clifford gate layers, state preparation and measurement, the set of total channel error rates for these is $\{q_{\mathcal{C},tot}(i)\}_{i=1}^{N_{\mathcal{C}}}$,
and the largest value from this set is $q_{\mathcal{C},max}:= \max(\{q_{\mathcal{C},tot}(i)\}_{i=1}^{N_{\mathcal{C}}})$.

Let $\epsilon_{\mathrm{canc}}$ denote the unconditional
probability mass of cancellation configurations not covered by
the conditional $1/2$ bound of Lemma~\ref{cancellation_lemma2}.
Two scenarios are neglected when using the $1/2$ bound, both of which contribute to this approximation error.
These are: (1) the cancellation of errors affecting $\mathcal{J}$-type gate layers with errors affecting the rest of the circuit, 
and (2) the cancellation of errors affecting all circuit components except $\mathcal{J}$-type gate layers.

The approximation error resulting from Scenario (1) may be bounded in the following way.
By arguments similar to those used in the proof of Lemma \ref{cancellation_lemma2}, if a collection of errors includes at least one error affecting a $\mathcal{J}$-type gate layer and at least one error affecting another circuit component, and these errors are separated by at least one randomising gate layer, then the probability of error cancellation is upper-bounded by $1/2$.
Therefore, these errors are accounted for if the value $1/2$ is used for the cancellation probability bound.
What is not accounted for is the cancellation of errors affecting $\mathcal{J}$-type gate layers with those affecting immediately preceding $\mathcal{W}$-type gate layers, as no randomising gate layers separate these errors.
The probability that errors of this type occur and cancel is
upper-bounded by
\begin{equation}
\sum_{i=1}^{N_{\mathcal{J}}}
\binom{N_{\mathcal{J}}}{i}
q_{i,\mathcal{J},\max}
q_{\mathcal{C},\max}^{\,i}.
\end{equation}
Recall that $q_{i,\mathcal{J},\max}$ is the maximum probability
that any set of $i$ $\mathcal{J}$-type gate layers are all affected
by errors, and $q_{\mathcal{C},\max}$ is the maximum value of the
set of total channel error probabilities for logical single-qubit
Clifford gate layers, state preparation, and measurement.

The approximation error resulting from Scenario~(2) may be bounded
using a union bound over subsets of at least two
$\mathcal{C}$-type noise channels.
This means that the total approximation error from using the value
$1/2$ to bound the cancellation probability for any collection of
errors is upper-bounded by
\begin{equation} \label{eqn_to_ref11}
\begin{split}
|\epsilon_{\mathrm{canc}}|
&\leq
\sum_{i=1}^{N_{\mathcal{J}}}
\binom{N_{\mathcal{J}}}{i}
q_{i,\mathcal{J},\max}
q_{\mathcal{C},\max}^{\,i}
\\
&\hspace{3.5em}
+
\sum_{i=2}^{N_{\mathcal{C}}}
\binom{N_{\mathcal{C}}}{i}
q_{\mathcal{C},\max}^{\,i}.
\end{split}
\end{equation}
If noise affecting the $\mathcal{J}$-type gate layers is the dominant source of circuit noise, then $q_{1,\mathcal{J},max} \gg q_{\mathcal{C},max}$.
And so, considering only the leading order term, the previous bound becomes
\begin{equation}
\begin{split}
|\epsilon_{canc}| &\leq O\big(N_{\mathcal{J}} \cdot q_{1,\mathcal{J},max} \cdot q_{\mathcal{C},max} \big). \\
\end{split}
\end{equation}
This means that the probability of error cancellation is bounded by
\begin{equation} \label{ghhh}
p_{err}p_{canc} \leq 1/2p_{err} + O\big(N_{\mathcal{J}} \cdot q_{1,\mathcal{J},max} \cdot q_{\mathcal{C},max} \big).
\end{equation} 
Using Eq.~\ref{p_inc_bound}, $p_{\mathrm{det}}\geq 1/2$, and
Eq.~\ref{ghhh}, we obtain
\begin{align}
p_{\mathrm{inc}}
&\geq
p_{\mathrm{err}}
\bigl(1-p_{\mathrm{canc}}\bigr)
p_{\mathrm{det}}
\nonumber\\
&=
\bigl(
p_{\mathrm{err}}
-
p_{\mathrm{err}}p_{\mathrm{canc}}
\bigr)
p_{\mathrm{det}}
\nonumber\\
&\geq
\frac{1}{2}
\left[
\frac{1}{2}p_{\mathrm{err}}
-
O\!\left(
N_J q_{1,J,\max} q_{C,\max}
\right)
\right].
\end{align}
Therefore,
\begin{equation}
\label{<retain-the-current-E26-label>}
p_{\mathrm{err}}
\leq
4p_{\mathrm{inc}}
+
O\!\left(
N_J q_{1,J,\max} q_{C,\max}
\right).
\end{equation}
Thus, using $\beta=1/2$ introduces an additive approximation
error satisfying
\begin{equation}
|\epsilon_{\gamma}|
\leq
O\!\left(
N_J q_{1,J,\max} q_{C,\max}
\right).
\end{equation}
\end{proof}

\subsection{Bounding the probability of error cancellation for arbitrary errors affecting a trap circuit} \label{altered_trap_construction}

\begin{figure*}%[thb]
\centering
\includegraphics[width=0.34\textwidth]{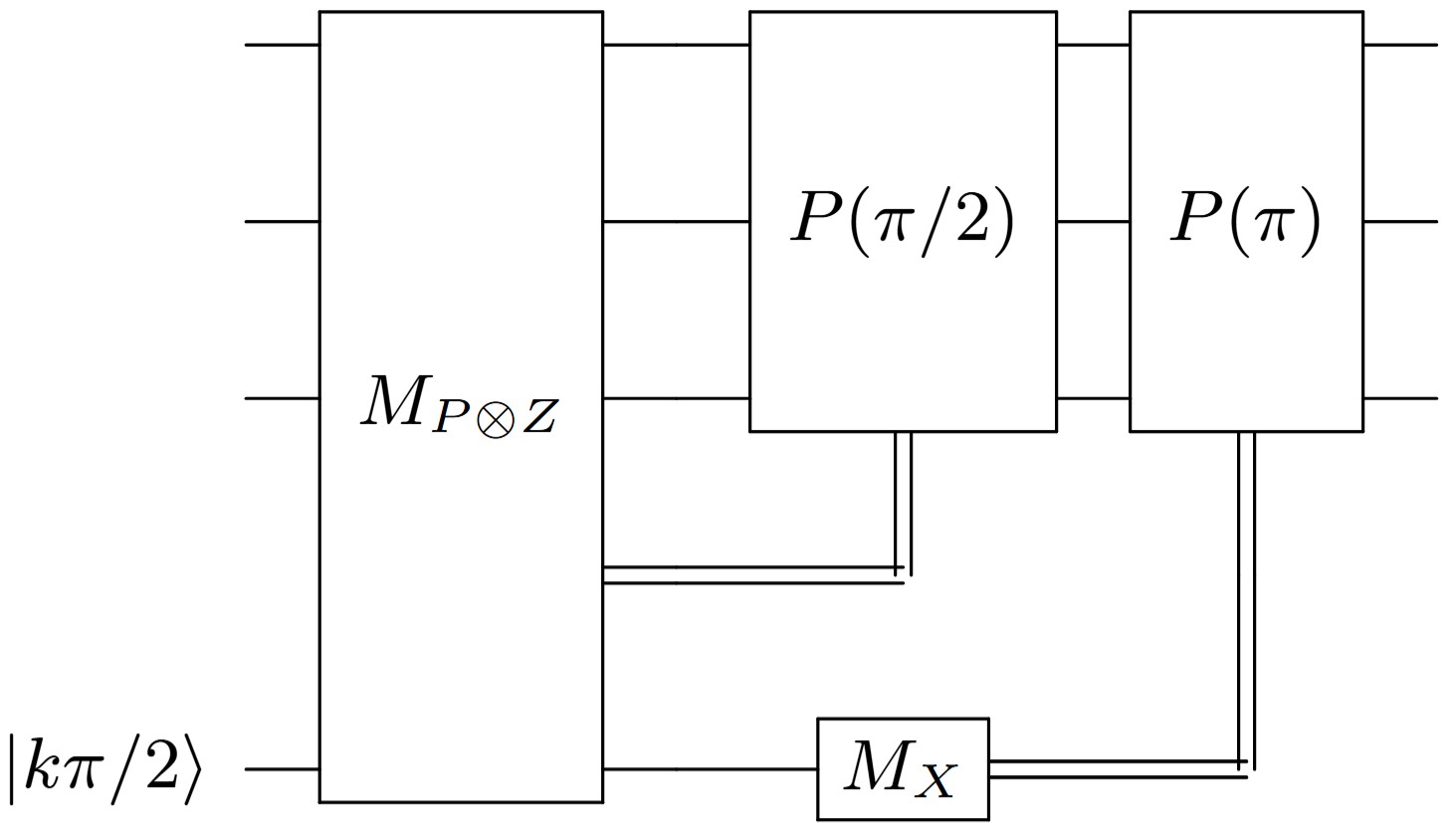}\vspace{0.7em}
\caption{In the modified trap construction described in Section \ref{altered_trap_construction}, the value of $k$ is randomly chosen where $k \in \{1,2\}$ with probability $1/2$. 
This means that either a $P(\pi/2)$ or $P(\pi)$ Pauli rotation gate is applied, depending on the magic state that is prepared.
If a $P(\pi/2)$ rotation is applied, this stabilises the quantum state with probability $1/2$, otherwise a correction operation must be included at the end of the circuit to ensure that the trap computation is deterministic.
The two adaptive correction operations after the adaptive measurements are applied in the same way as in the target circuit.
Namely if the rotation is in the wrong direction the $P(\pi/2)$ operation is applied.
So that if a $\ket{\pi/2}$ magic state is used, the overall operation the gadget performs is an identity operation or a $P(\pi/2)$ rotation operation.
While if a $\ket{\pi}$ magic state is used, the overall operation the gadget performs is an $P(-\pi/2)$ rotation operation or a $P(\pi)$ rotation operation.
}
\label{fig:cancel_all_errors}
\end{figure*}

We now describe a \textit{modified trap circuit} construction that guarantees the cancellation probability of any collection of errors affecting a randomly generated trap circuit is upper-bounded by $7/8$.
Without loss of generality, we will assume magic states are only used to perform $Z$-type rotations in the target and trap circuits by projective measurement, i.e. these logical gate operations are of the form: $Z^{\otimes a}(\theta)$ where $a \in \{0,1\}^n$.
This does not restrict computation, since straightforward conjugation using Clifford gates can be used to arbitrarily change the basis of the Pauli rotation.

Although otherwise identical to the previous construction, the modified trap circuit construction differs in several respects:
\begin{enumerate}
    \item The ordering of neighbouring $\mathcal{W}$-type gate layers is randomised.
    \item For each gate performed by the consumption of a magic state, either a $\ket{\pi}$ or $\ket{\pi/2}$ magic state is prepared, each with probability $1/2$.
    The same correction gates are applied as for the corresponding gates in the target circuit, where $\ket{\pi/4}$ magic states are used to implement $\pi/4$ rotations.
    Including the correction gates, the overall gate operations performed are then of the form $Z^{\otimes a}(k\cdot\pi/2)$ for $a \in \{0,1\}^n$ and $k\in\{1,2\}$.
    The value of $k$ is then 1 or 2 depending on whether the magic state used is $\ket{\pi/2}$ or $\ket{\pi}$, and on the random direction of the initial projective measurement operation.
    Fig. \ref{fig:cancel_all_errors} shows the operations required to implement such a gate in a trap circuit.
    \item To ensure that the outputs of the trap circuits are deterministic in the absence of error, an adjoint operation is included for each $\pi/2$ rotation gate that does not stabilise the quantum state.
    These adjoint operations are propagated to the end of each trap circuit, and the order in which they are applied at the end of the circuit is randomised.
    Of these operations, those that commute may be parallelised, while those that combine into an identity operation may be removed entirely.
    \item For any gate $Z^{\otimes a}(k\cdot\pi/2)$ conjugated by $\mathcal{W}$-type gate layers, for qubits in its support (i.e. where $a_i=1$) the components of the conjugating $\mathcal{W}$-type gate layers are either all $H$-type or all $S$-type. 
    Considering only logical qubits in the support of $a$, the conjugation takes the form $U^{\otimes a} \cdot Z^{\otimes a}(k \pi/2) \cdot (U^\dagger)^{\otimes a}$ for $U \in \{H,S\}$ chosen with uniform probability.
\end{enumerate}
During the protocol, we assume that the additional adjoint
operations are non-noise-decreasing, so that the total error
rate of each modified trap circuit is no smaller than that of
the corresponding target circuit. The resulting TVD bound is
therefore conservative.

Using the modified construction, each trap circuit is of the form:
\begin{equation}
\mathcal{H}^{t} \circ (\circ^p_{k=1}\mathcal{C}_k)  \circ (\circ^m_{j=1}\mathcal{W}_{3j} \mathcal{J}_{3j-1} \mathcal{W}_{3j-2}) \circ \mathcal{H}^{t}.  \\
\end{equation}
Note that neighbouring $\mathcal{W}$-type gate layers are randomly ordered.
If this randomisation reverses the order of an $S$ and an $H$ gate, the $S$ gate is mapped to an $X(\pi/2)$ gate; likewise, an $S^\dagger$ gate is mapped to an $X(\pi/2)^\dagger$ gate.
The $(\circ^p_{k=1}\mathcal{C}_k)$ gate layers are randomly ordered adjoint operations of the $\pi/2$ $Z$-basis rotations that do not stabilise the state $\mathcal{H}^{t}\ket{0}^{\otimes n}$, and have been propagated to the end of the circuit.
Where, the same as for the previous construction, the value of $t$ is chosen uniformly at random from $\{0,1\}$.

We now show that any two Pauli errors that occur while performing such a trap circuit cancel with bounded probability.

\vspace{1em}
We now prove Lemma \ref{bounding_all_canc}, that states:
\vspace{1em}

\textit{
For the modified trap circuit construction, the probability of error cancellation in the trap circuits is upper-bounded by $p_{canc}\leq 7/8$.
}

\begin{proof}

This proof is organised as follows.
We first show that errors that occur at any two positions in a randomly generated trap circuit cancel with bounded probability. 
We then extend this to bound the error cancellation probability for any number of errors affecting a randomly generated trap circuit.

By Lemma \ref{cancellation_lemma2}, the probability of cancellation for any two errors affecting $J$-type layers is upper-bounded by $1/2$.
Likewise, using arguments similar to those used in the proof of Lemma \ref{cancellation_lemma2}, any two errors separated by randomly chosen $\mathcal{W}$-type gate layers, where these intervening gate layers do not compile together into an identity operation,
are randomly mapped to different Pauli errors when the errors are propagated together and so
cancel with probability upper-bounded by $1/2$.
This includes the case of a gate preparation or measurement error with an error affecting a gate layer, and the case of errors affecting any two $\mathcal{W}$-type gate layers.

For readability, the rest of the proof is structured in sections.
In Sections $(a)$-$(e)$, upper bounds are derived for the probability of cancellation for all other scenarios where two Pauli errors affect the trap circuit.
In Section $(f)$, this analysis is extended to bound the probability of cancellation for any number of errors.

\vspace{0.5em}
\subsubsection{Cancellation of state preparation and measurement errors}

The gate layers of a trap circuit affected by state preparation and measurement Pauli errors may be expressed in the form
\begin{equation}
P_{M}\circ\mathcal{H}^{t}  \circ (\circ^p_{k=1}\mathcal{C}_k) \circ (\circ^m_{j=1}\mathcal{W}_{3j} \mathcal{J}_{3j-1} \mathcal{W}_{3j-2}) \circ \mathcal{H}^{t} \circ P_{P},  \\
\end{equation}
where $P_{P}$ and $P_{M}$ denote Pauli errors affecting state-preparation and measurement, respectively.
Firstly, without changing the logic of the trap circuit, the ordering of neighbouring $\mathcal{W}$-type layers can be un-randomised to
\begin{equation}
P_{M}\circ\mathcal{H}^{t} \circ (\circ^p_{k=1}\mathcal{C}_k) \circ (\circ^m_{j=1}\mathcal{W}_{3j}' \mathcal{J}_{3j-1} \mathcal{W}_{3j-2}')  \circ \mathcal{H}^{t} \circ P_{P},  \\
\end{equation}
where the prime notation indicates the un-randomisation of contiguous $\mathcal{W}$-type gate layer ordering.
The gate layers separating the two Pauli errors compile to a mixture of randomly oriented CNOT gates, 
$X$- or $Z$-type $\pi/2$ rotation gates, and identity operations.
The $\ket{0}^{\otimes n}$ state is always initially prepared, and measurement is always in the $Z$ basis.
So the only state preparation and measurement Pauli errors that do not stabilise the initial state or measurement, and non-trivially affect the target and trap computations, are $X$-type errors.

The randomised construction of the trap circuits ensures that state-preparation and measurement errors propagate in a random manner, so that their cancellation occurs with bounded probability.
This may be seen by propagating the state-preparation and measurement Pauli errors through the trap circuit until they are either side of the $j$-th gate layer sequence from $(\circ^m_{j=1}\mathcal{W}_{3j}' \mathcal{J}_{3j-1} \mathcal{W}_{3j-2}')$, 
i.e. $\mathcal{W}_{3j}' \mathcal{J}_{3j-1} \mathcal{W}_{3j-2}'$, where the $\mathcal{J}$-type layer has non-trivial support on at least one of the same qubits as one of the propagated Pauli errors.
Such a $J$-type gate layer must exist for the qubits that the errors are affecting to be non-trivially included in the target computation.
The errors and the gate layers may then be considered independently of the rest of the circuits as 
\begin{equation}
P_M' (\mathcal{W}_{3j}' \mathcal{J}_{3j-1} \mathcal{W}_{3j-2}') P_P',\\
\end{equation}
where $P_P'$ and $P_M'$ denote the propagated state-preparation and measurement errors, respectively.
And because the gate layers compile together into randomly oriented CNOT gates, and either $X$-or $Z$-type $\pi/2$ rotations or identity operations, the propagation of the Pauli errors through these intervening gate layers is randomised.

Any randomly oriented CNOT in $\mathcal{J}_{3j-1}$ that has shared support on at least one qubit affected by at least one of the Pauli errors randomises the propagation of the Pauli errors through it, resulting in cancellation of those components of the Pauli errors with probability upper-bounded by $1/2$. 
Any Pauli rotation operation in $\mathcal{J}_{3j-1}$ that has shared support on at least one of the same qubits as at least one of the Pauli errors is a rotation by $\pi/2$ with probability $1/2$.
With probability $1/2$, any Pauli rotation operation stabilises the ideal quantum state, and so does not require an adjoint operation to keep the circuit deterministic in the absence of errors.
This means that the random propagation of the errors through the rotation operation is not undone by an adjoint operation.
Furthermore, the sandwiching $\mathcal{W}$-type gate layers randomise the rotation basis to be either $Z$-type or $X$-type with probability $1/2$.
Therefore, a Pauli operation in $\mathcal{J}_{3j-1}$ is: (1) a $\pi/2$ rotation, (2) stabilises the ideal quantum state, and (3) does not commute with the Pauli errors, with probability at least $1/2^3$.
This means that the propagation through the circuit of any state preparation error and measurement error is randomised with a probability at least $1/8$.
Consequently, the probability of state-preparation and measurement errors cancelling in a randomly generated trap circuit is upper-bounded by $7/8$.

\subsubsection{Cancellation of gate errors and state-preparation or measurement errors}

Only Pauli errors with an $X$ component non-trivially affect state preparation and measurement.
Errors affecting the final gate layer of the circuit may be combined with measurement errors and are non-trivial only if they have an $X$ component.
The modified trap construction ensures that any Pauli error affecting a gate layer is separated from errors arising during state preparation or measurement by at least one randomly chosen single-qubit Clifford gate layer.
This randomises the propagation required to combine either a state preparation or measurement Pauli error with a gate Pauli error.
And, through similar arguments as were used to prove Lemma \ref{cancellation_lemma2}, any state preparation or measurement Pauli error cancels with a gate Pauli error with probability upper-bounded by $1/2$.

\subsubsection{\texorpdfstring{Cancellation of gate errors occurring before and after $J$-type layers}{Cancellation of gate errors occurring before and after J-type layers}}

A trap circuit affected by errors occurring before and after a $J$-type layer may be expressed in the form
\begin{widetext}
\begin{equation}
\begin{split}
\mathcal{H}^{t} \circ& (\circ^m_{j=l+1}\mathcal{W}_{3j} \mathcal{J}_{3j-1} \mathcal{W}_{3j-2}) \circ (\mathcal{W}_{3l} P_{3l-1} \mathcal{J}_{3l-1} P_{3l-1-2}\mathcal{W}_{3l-2}) \circ (\circ^{l-1}_{j=1}\mathcal{W}_{3j} \mathcal{J}_{3j-1} \mathcal{W}_{3j-2}) \circ \mathcal{H}^{t}  .  \\
\end{split}
\end{equation}
\end{widetext}
Restricting our attention to only the errors and the intervening gate layer, we obtain
\begin{equation}
    P_{3l-1} \mathcal{J}_{3l-1} P_{3l-1-2}.
\end{equation}
As the $\mathcal{J}$-type layer is a mixture of $CZ$, identity and $Z$-type rotations in the target circuit, and gate positioning and type is unchanged in the trap circuits, any Pauli errors that cancel in the target circuit also cancel in the trap circuits.
Cancellation of these errors in the trap circuits is permitted since these errors do not contribute to the total circuit error rate of the target circuit.
The rotation angle of the $Z$-type rotations in the traps is chosen to be $\pi/2$ or 0, each with probability $1/2$.
This randomises the propagation of errors in the trap circuits that do not cancel in the target circuit with probability $1/2$, and so the probability of errors of this type cancelling in the trap circuits is upper-bounded by $1/2$.

\subsubsection{Cancellation of  gate errors occurring between sequences of $\mathcal{W}$-type and $\mathcal{J}$-type gate layers }

A trap circuit affected by errors occurring between sequences of $\mathcal{W}$-type and $\mathcal{J}$-type gate layers may be written
\begin{widetext}
\begin{equation}
\begin{split}
\mathcal{H}^{t} \circ& (\circ^m_{j=l+1}\mathcal{W}_{3j} \mathcal{J}_{3j-1} \mathcal{W}_{3j-2})  \circ P_{3l} \circ (\mathcal{W}_{3l} \mathcal{J}_{3l-1} \mathcal{W}_{3l-2}) \circ P_{3(l-1)} \circ (\circ^{l-1}_{j=1}\mathcal{W}_{3j} \mathcal{J}_{3j-1} \mathcal{W}_{3j-2}) \circ \mathcal{H}^{t}  .  \\
\end{split}
\end{equation}
\end{widetext}
Again isolating the errors and gate layers dividing them from the rest of the circuit, we obtain the sequence
\begin{equation}
    P_{3l} (\mathcal{W}_{3l} \mathcal{J}_{3l-1} \mathcal{W}_{3l-2}) P_{3(l-1)}.
\end{equation}
Since the ordering of neighbouring $\mathcal{W}$-type gate layers is randomised in the modified trap circuit construction, the Pauli errors are randomly mapped to new Pauli operators when they are propagated together through the intervening gate layers. 
There is at least one such randomising layer between the Pauli errors with probability $3/4$.
This can be combined with the previously derived result that two Pauli errors separated by a randomising gate layer cancel with probability upper-bounded by $1/2$, to get that the propagation of the errors is randomised with probability $3/8$.
This means that the Pauli errors cancel with probability upper-bounded by $5/8$.
So that for
\begin{equation}
    (\mathcal{W}_{3l} \mathcal{J}_{3l-1} \mathcal{W}_{3l-2}) P_{3l}' P_{3(l-1)},
\end{equation}
where $ P_{3l}' = (\mathcal{W}_{3l} \mathcal{J}_{3l-1} \mathcal{W}_{3l-2})^\dagger P_{3l} (\mathcal{W}_{3l} \mathcal{J}_{3l-1} \mathcal{W}_{3l-2})$,
it is true that 
\begin{equation}
\Pr(P_{3l}' P_{3(l-1)} \propto I) \leq 5/8.
\end{equation}

\subsubsection{Cancellation of any errors with errors from adjoint Clifford gate layers}

In this trap circuit construction, Pauli rotation gates applied by the consumption of magic states randomly perform rotations by either $\pi/2$ or 0.
If a rotation by $\pi/2$ is performed, this operation stabilises the quantum state with probability $1/2$.
This depends on whether the layer of Hadamard gates is applied at the beginning and end of the circuit, making the quantum state either $\ket{0}^{\otimes n}$ or $\ket{+}^{\otimes n}$, and on the choice of random sandwiching gates, rendering the rotation gate basis either $Z$-type or $X$-type.
If the gate does not stabilise the state, then a Clifford correction must be applied that performs the adjoint operation of the Pauli rotation gate.
These operations are propagated to the end of the circuit and applied in a random order.
All of the trap circuit gate operations are Clifford, and so this propagation may be performed efficiently.
As stated, the additional Clifford operation is only necessary when both the Pauli operation is a $\pi/2$ rotation, and when this operation does not stabilise the quantum state.
As the probability of each of these is $1/2$, the probability of a correction operation being necessary is $1/4$.
So that the probability that an error affects another part of the circuit and cancels due to an error affecting one of the correction operations is upper-bounded by $1/4$.

\subsubsection{Cancellation of any number of errors}

We have shown that, in this construction, any two non-trivial
Pauli errors affecting a randomly generated trap circuit cancel
with probability upper-bounded by $c={7}/{8}$.
The preceding two-error bounds apply uniformly to any non-trivial
residual Pauli errors obtained by propagating and combining other
errors, since the remaining trap randomisation continues to
randomise their relative Pauli orientation.

For $j\geq 1$, let $q_j$ denote the worst-case probability that a
residual configuration of $j$ non-trivial Pauli errors affecting a
trap circuit cancels. Clearly, $q_1=0$, while the preceding
two-error analysis gives $q_2\leq c$.
Now consider a residual configuration of $j\geq 3$ errors. Propagate
any two of these errors together, and let $p_j$ denote the probability
that they cancel. If they cancel, then $j-2$ non-trivial errors remain.
Otherwise, they combine into a new non-trivial Pauli error, leaving
$j-1$ non-trivial errors. Therefore,
\begin{equation}
q_j \leq p_j q_{j-2} + (1-p_j)q_{j-1}.
\end{equation}
For $j=3$, this gives
\[
q_3 \leq p_3 q_1+(1-p_3)q_2 \leq c.
\]
More generally, if $q_{j-2}\leq c$ and $q_{j-1}\leq c$, then
\[
q_j \leq p_jc+(1-p_j)c = c.
\]
It follows by strong induction that $q_j\leq c$ for every $j\geq 1$.
Therefore, the cancellation probability for any collection of
errors affecting a randomly generated trap circuit satisfies $p_{\mathrm{canc}}\leq\frac{7}{8}$.
\end{proof}

\section{Numerical decomposition of false positive rates}
\label{appendix:false_positive_numerics}

\begin{figure*}[!ht]
\vspace{-1.4em}
    \centering
    \begin{subfigure}[t]{0.44\textwidth}
        \includegraphics[width=\linewidth]{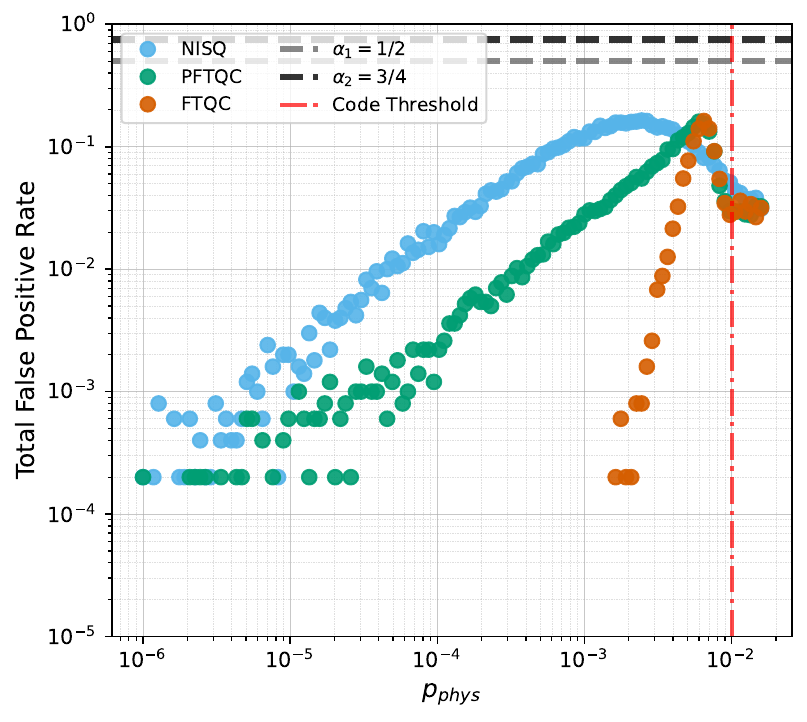}
        \caption{Total false positive rate}
        \label{fig:fp-panel-a}
    \end{subfigure}
    % \hfill
    \hspace{4em}
    \begin{subfigure}[t]{0.44\textwidth}
        \includegraphics[width=\linewidth]{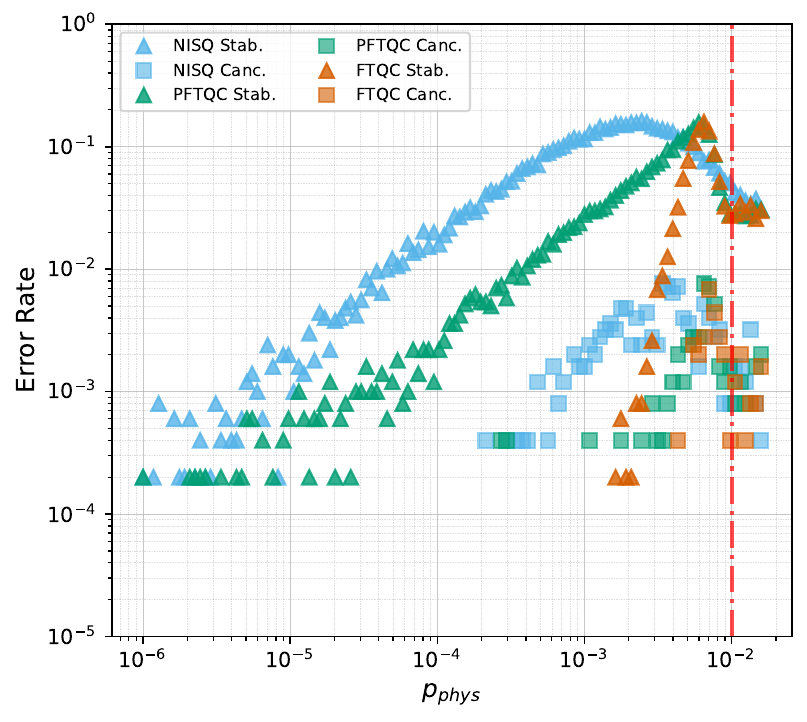}
        \caption{False positive decomposition}
        \label{fig:fp-panel-b}
    \end{subfigure}
    % \hfill
    \begin{subfigure}[t]{0.44\textwidth}
        \includegraphics[width=\linewidth]{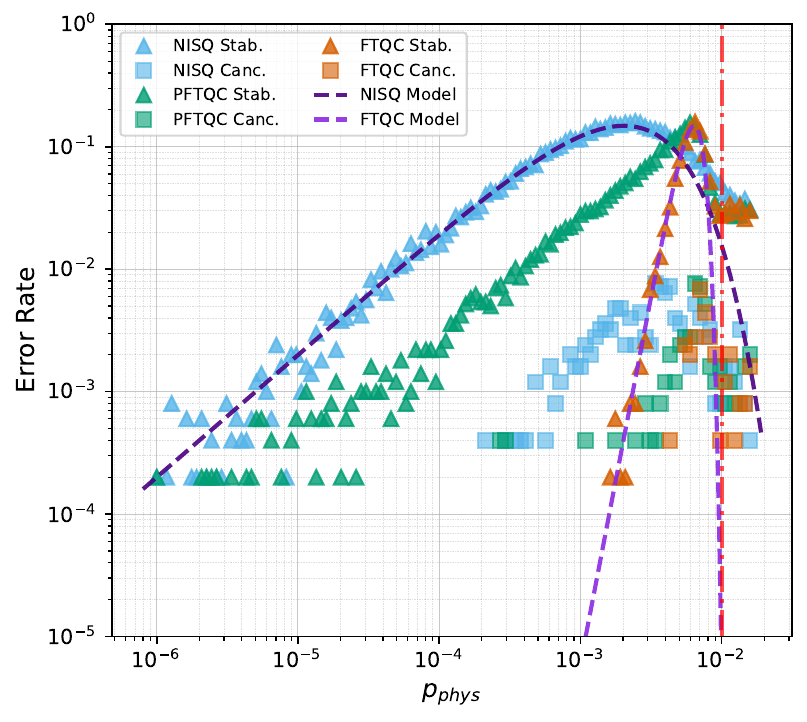}
        \caption{Stabilisation-induced false positives vs. analytical prediction}
        \label{fig:fp-panel-c}
    \end{subfigure}
    \caption{\emph{Numerical analysis of false positive rates for IQP circuits.} 
    The trap circuit false positive rates are plotted for 5-qubit, 40-layer IQP circuits for a range of physical error rates. 
    In (a), false positive rates are plotted against the physical error rate for NISQ, PFTQC and FTQC circuits.
    Two analytical upper bounds on the false positive rates from eqn. \ref{false_positive_bounds} are also plotted as horizontal dashed lines: $\epsilon_{s,1} = 1/2$ and $\epsilon_{s,2} = 3/4$.
    The numerical false positive rates are observed to lie well below the analytical bounds.
    In (b), false positive rates are decomposed into stabilisation-induced false positive rates and cancellation-induced false positive rates.
    These results indicate that the stabilisation-induced false positive rates are the dominant contribution to the total false positive rates. 
    Cancellation-induced false positive rates are consistently one to two orders of magnitude smaller than those from stabilisation. 
    In (c), the analytical formula from eqn. \ref{stabi_formula} is overlaid on the plot from (b), showing close agreement with the observed behaviour.
    The formula used to generate the model curves was: $s(k, p_{L}) = \left(1 - \frac{2p_{L}}{3} \right)^k - (1 - p_{L})^k$, where $k = 3 m n$ denotes the total number of error locations in the circuit, with $m$ the circuit depth and $n$ the number of qubits.
    }
    \label{fig:iqp-false-positives-split-errors}
\end{figure*}

We carried out a numerical study of the false positive rates arising when logical accreditation is applied to IQP circuit sampling under a local logical-level depolarising noise model. 
False positives, where errors affect trap circuits and are not registered in the measurement outcomes, arise from two primary sources.
The first is stabilisation-induced errors, which correspond to logical Pauli errors that stabilise the logical quantum state or commute with the trap circuit measurement basis, leading to false positives where the measurement output is $0^n$ despite errors having occurred. 
The second source is cancellation-induced errors, which occur when multiple errors combine to cancel, resulting in trap circuit false positive measurement outcomes.

The total trap circuit false positive rate may be decomposed as
\begin{equation}
\Pr[\text{false positive}] = p_{\mathrm{stab}} + p_{\mathrm{canc}},
\end{equation}
where \(p_{\mathrm{stab}}\) and \(p_{\mathrm{canc}}\) denote the probability contributions from stabilisation and cancellation effects, respectively.

Under the local logical depolarising noise model with physical error rate $p_{phys}$, each error location applies identity with probability $1-p_L$, or one of the three logical Pauli errors $\{X, Y, Z\}$ each with probability $p_L/3$. 
Given that the trap circuit measurement is chosen uniformly at random to be in the $Z$- or $X$-basis, the per-location probability that stabilises the logical state is
\begin{equation}
p_{\mathrm{stab, loc}} = 1 - \frac{2p_L}{3}.
\end{equation}

Assuming $k$ independent error locations, the probability that the operators probabilistically applied at all error locations stabilise the logical state (including the trivial no-error case) is
\begin{equation}
\left(1 - \frac{2p_L}{3}\right)^k.
\end{equation}
Subtracting from this the probability that no errors occur, $(1-p_L)^k$, isolates the trap circuit false positives arising from non-trivial stabilising errors:
\begin{equation} \label{stabi_formula}
p_{\mathrm{stab}} = \left(1 - \frac{2p_L}{3}\right)^k - (1 - p_L)^k.
\end{equation}

The numerical results shown in Fig. \ref{fig:iqp-false-positives-split-errors} indicate that the false positive contribution from cancellation effects is minimal, remaining consistently $1$–$2$ orders of magnitude lower than the contribution from stabilisation effects across the full range of physical error rates examined.
These findings support the conclusion that, for stochastic noise models relevant to fault-tolerant quantum computation that are well-approximated by depolarising noise, false positive rates where logical accreditation is applied to IQP circuits are dominated by stabilisation-induced errors. 
The analytical expression in eqn. \ref{stabi_formula} provides a means of approximating the trap circuit false positive rate caused by this effect and closely corresponds to the experimental results plotted in Fig. \ref{fig:iqp-false-positives-split-errors} (c).

\noindent\textbf{Remark:} \emph{This analysis is primarily valid under the local depolarising noise model used for these numerical experiments. 
However, the implication that stabilisation effects dominate the false positive rate likely extends to more general stochastic error models, particularly those where weight-1 errors are the most probable and higher-weight errors occur with decreasing rates. 
This would generally be true in regimes of large-scale fault-tolerance, and where the error channels are independent.}

\vspace{-2em}
\section{Robustness bound} \label{robustness_derivation}
\vspace{-1.3em}

Assumptions A1 and A2 can be violated if logical single-qubit Clifford gate- or physical single-qubit gate noise exhibits gate-dependence.
This may cause noise channels to act differently on a circuit instance during logical accreditation depending on the choice of gates used.
We will now show that, when these assumptions are violated, the resulting difference in the computed TVD upper bound depends only linearly on the diamond norm distance of each of the noise channels.
In this analysis we will ignore finite-sampling error.
\begin{definition}[Diamond norm distance]
For any two CPTP maps $A$ and $B$ acting on density matrices of a $2^n$-dimensional Hilbert space $\mathcal{H}$, their diamond norm distance is defined as
\begin{equation}
||A - B||_{\diamond}:= \max_\rho ( || (A \otimes \mathbb{I})(\rho) - (B \otimes \mathbb{I})(\rho) ||_1 ),
\end{equation}
where $\mathbb{I}$ is identity channel acting on auxiliary system $\mathcal{H}'$ of size $\text{dim}(\mathcal{H}') = \text{dim}(\mathcal{H})$, $||.||_1$ is the $l_1$-norm or trace norm, where $||X||_1:=\text{Tr}(\sqrt{X^\dagger X})$, the density matrix $\rho$ ranges over the Hilbert space $\mathcal{H} \otimes \mathcal{H}'$, and the expression is maximised over all possible density matrices.
\end{definition}

In order to derive this result, we will make use of the following lemma:
\begin{lemma} \label{diamond_norm_lemma}
For CPTP maps $A$, $B$, $C$ and $D$, the diamond norm, indicated $|| \hspace{0.2em}.\hspace{0.2em} ||_{\diamond}$, exhibits the property
    \begin{equation*}
    ||AB - CD||_{\diamond} \leq ||A - C||_{\diamond} + ||B - D||_{\diamond}.
\end{equation*}
\end{lemma}

\vspace{0em}
\begin{proof}
\begin{equation}
    \begin{split}
        ||AB - CD||_{\diamond} &= ||AB +AD -AD - CD||_{\diamond} \\
        &\leq ||AB - AD||_{\diamond} + ||AD - CD||_{\diamond}\\
         &\leq ||A||_{\diamond}\cdot||B - D||_{\diamond} + ||A - C||_{\diamond} \cdot ||D||_{\diamond}\\
         &\leq ||B - D||_{\diamond} + ||A - C||_{\diamond}. \\
\vspace{2em}
    \end{split}
\end{equation}
\end{proof}
where the second line follows from a triangle inequality, the third line from the submultiplicity of the diamond norm, and the final line from the property that $||G||_{\diamond}\leq 1$ for any quantum channel $G$.

\subsection{Proof of Theorem \ref{robustness_theorem}}

\begin{proof}
Let the $k$-th logical trap circuit affected by noise, where the previously stated assumptions hold, be denoted
\begin{equation}
\begin{split}
\tilde{\mathcal{C}}^{(k)} &= \mathcal{E}_{M}^{(k)} \cdot \prod_{j=1}^D (\mathcal{E}_{j}^{(k)} \mathcal{L}_j^{(k)}) \cdot \mathcal{E}_{P}^{(k)},\\
\end{split}
\end{equation}
and the $k$-th logical trap circuit affected by noise where the assumptions are violated be denoted
\begin{equation}
\begin{split}
\tilde{\mathcal{C}}^{(k)}{}' &= \mathcal{E}_{M}^{(k)}{}' \cdot \prod_{j=1}^D (\mathcal{E}_{j}^{(k)}{}' \mathcal{L}_j^{(k)}) \cdot \mathcal{E}_{P}^{(k)}{}'.\\
\end{split}
\end{equation}
Applying Lemma \ref{diamond_norm_lemma} once to the quantity $||\tilde{\mathcal{C}}^{(k)} - \tilde{\mathcal{C}}^{(k)}{}'||_{\diamond}$, we get that
\begin{widetext}
\begin{equation}
\begin{split}
||\tilde{\mathcal{C}}^{(k)} - \tilde{\mathcal{C}}^{(k)}{}'||_{\diamond}  &= || \mathcal{E}_{M}^{(k)} \cdot \prod_{j=1}^D (\mathcal{E}_{j}^{(k)} \mathcal{L}_j^{(k)}) \cdot \mathcal{E}_{P}^{(k)} - \mathcal{E}_{M}^{(k)}{}' \cdot \prod_{j=1}^D (\mathcal{E}_{j}^{(k)}{}' \mathcal{L}_j^{(k)}) \cdot \mathcal{E}_{P}^{(k)}{}' ||_{\diamond} \\
&  \leq || \prod_{j=1}^D (\mathcal{E}_{j}^{(k)} \mathcal{L}_j^{(k)}) \cdot \mathcal{E}_{P}^{(k)} - \prod_{j=1}^D (\mathcal{E}_{j}^{(k)}{}' \mathcal{L}_j^{(k)}) \cdot \mathcal{E}_{P}^{(k)}{}' ||_{\diamond} + || \mathcal{E}_{M}^{(k)} - \mathcal{E}_{M}^{(k)}{}'  ||_{\diamond}. \\
\end{split}
\end{equation}
\end{widetext}
This operation can be iteratively applied a further $D$ times, resulting in the bound
\begin{widetext}
\begin{equation}
\begin{split}
 &|| \prod_{j=1}^D (\mathcal{E}_{j}^{(k)} \mathcal{L}_j^{(k)}) \cdot \mathcal{E}_{P}^{(k)} - \prod_{j=1}^D (\mathcal{E}_{j}^{(k)}{}' \mathcal{L}_j^{(k)}) \cdot \mathcal{E}_{P}^{(k)}{}' ||_{\diamond} + || \mathcal{E}_{M}^{(k)} - \mathcal{E}_{M}^{(k)}{}'  ||_{\diamond} \\
&\hspace{4em}\leq  ||\mathcal{E}_{P}^{(k)} - \mathcal{E}_{P}^{(k)}{}'||_{\diamond} + \sum_{j=1}^D||\mathcal{E}_{j}^{(k)}-  \mathcal{E}_{j}^{(k)}{}'||_{\diamond} + || \mathcal{E}_{M}^{(k)} - \mathcal{E}_{M}^{(k)}{}'  ||_{\diamond} \\
&\hspace{4em}\leq \sum_{j=0}^{D+1}||\mathcal{E}_{j}^{(k)}-  \mathcal{E}_{j}^{(k)}{}'||_{\diamond} , \\
\end{split}
\end{equation}
\end{widetext}
where to get the final line, the logical state preparation and measurement noise channels are relabelled using indices `$0$' and `$D+1$' respectively.
And so the diamond norm distance of the $k$-th trap circuit affected by noise where the assumptions hold, $\tilde{\mathcal{C}}^{(k)}$, from the $k$-th trap circuit affected by noise where the assumptions are violated, $\tilde{\mathcal{C}}^{(k)}{'}$ is bounded by
\begin{equation} \label{diamond_norm_distance_bound}
    ||\tilde{\mathcal{C}}^{(k)} - \tilde{\mathcal{C}}^{(k)}{'}||_{\diamond} \leq \sum_j || \mathcal{E}_{j}^{(k)} - \mathcal{E}_{j}^{(k)}{'}||_{\diamond}. 
\end{equation}
Now, from the definition of the diamond norm
\begin{equation*}
  | \rho_{\tilde{\mathcal{C}}^{(k)}} - \rho_{\tilde{\mathcal{C}}^{(k)}{'}} |_{1} \leq ||\tilde{\mathcal{C}}^{(k)} - \tilde{\mathcal{C}}^{(k)}{'}||_{\diamond}, 
\end{equation*}
where $\rho_{\tilde{\mathcal{C}}^{(k)}} $ is the output state generated by applying the circuit ${\tilde{\mathcal{C}}}^{(k)} $ to any input state, and $\rho_{\tilde{\mathcal{C}}^{(k)}{'}}$ is the output state generated by applying the circuit $\tilde{\mathcal{C}}^{(k)}{'}$ to the same input state.
Additionally, we have that
\begin{equation*}
  \delta(\mathcal{D}( \rho_{\tilde{\mathcal{C}}^{(k)}} ) , \mathcal{D}( \rho_{\tilde{\mathcal{C}}^{(k)}{'}} )) \leq \frac{1}{2}| \rho_{\tilde{\mathcal{C}}^{(k)}} - \rho_{\tilde{\mathcal{C}}^{(k)}{'}} |_{1} , 
\end{equation*}
where $\mathcal{D}(.)$ indicates the output probability distribution generated from measuring a quantum state in a given basis.
Therefore, it follows that 
\begin{equation*}
  \delta(\mathcal{D}( \rho_{\tilde{\mathcal{C}}^{(k)}} ) , \mathcal{D}( \rho_{\tilde{\mathcal{C}}^{(k)}{'}} )) \leq \frac{1}{2}\sum_j || \mathcal{E}_{j}^{(k)} - \mathcal{E}_{j}^{(k)}{'}||_{\diamond} .
\end{equation*}
The protocol bound computed if assumptions are violated is of the form $\gamma' = \gamma + \epsilon_d$, where $\gamma$ is the bound where the assumptions hold,  $\gamma'$ is the bound computed where the assumptions are violated, and $\epsilon_d$ is the difference in the protocol bound induced by the violation of assumptions.
The absolute value of $\epsilon_d$ may be bounded by taking the sum over the eqn. \ref{diamond_norm_distance_bound} bounds for each of the trap circuits.
The mean deviation over $M$ trap circuits provides the robustness bound for the protocol:
 \begin{equation*}
 | \gamma - \gamma'| \leq M^{-1}(1-\beta)^{-1}\sum_k\sum_j || \mathcal{E}_{j}^{(k)} - \mathcal{E}_{j}^{(k)}{'}||_{\diamond} .
 \end{equation*}
\end{proof}

\begin{figure*}[]
    \centering
    \begin{subfigure}[t]{0.4\textwidth}
        \centering
        \includegraphics[width=\textwidth]{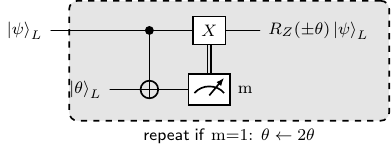}
        \caption{}
    \end{subfigure}
    \hspace{1cm}
    \begin{subfigure}[t]{0.4\textwidth}
        \centering
        \includegraphics[width=\textwidth]{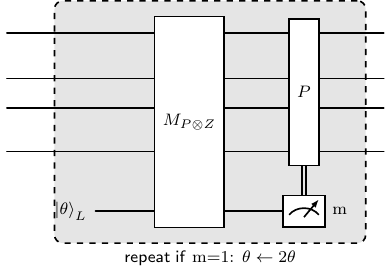}
        \caption{}
    \end{subfigure}
\caption{\textit{Repeat-until-success (RUS) gadgets for implementing single- and multi-qubit Pauli rotations using magic states.} 
(a) A quantum circuit diagram representing an analog single-qubit $Z$-rotation gate, $R_Z(\theta)$, where the ancilla qubit is measured destructively in the $Z$ basis.
If the ancilla qubit measurement outcome yields eigenvalue $-1$ (i.e., measurement outcome bit “1”), the RUS protocol is repeated. 
(b) A quantum circuit diagram representing an analog multi-qubit Pauli rotation gate, $R_P(\theta)$, where $P$ denotes an arbitrary Pauli string operator.
If the measurement outcome of $P \otimes Z$ yields eigenvalue $-1$ (i.e., measurement outcome bit “1”), the RUS protocol is repeated. 
A destructive $X$-basis measurement is used to determine whether a byproduct Pauli operator $P$ has been applied to the output state and requires correction.
}
\label{fig:STAR_teleport}
\end{figure*}

\section{Upper bound for the second-order Rényi entropy density using logical accreditation} \label{entropy_bound_derivation}

In the logical accreditation protocol, the $n$-logical-qubit
experimental target output state, averaged over the uniformly random
target execution position, may be expressed as
\begin{equation}
    \label{eqaccreditation-2}
    \rho_{\mathrm{out}}
    =
    (1-p_{\mathrm{err}})\rho_{\mathrm{out,id}}
    +
    p_{\mathrm{err}}\rho_{\mathrm{noisy}},
\end{equation}
where $p_{\mathrm{err}}$ is the corresponding position-averaged total
logical error rate of the target circuit.
 The purity of the logical output state can be lower-bounded as
\begin{widetext}
\begin{equation}
\begin{aligned}
\label{renyi_ent_bound}
\text{Tr}\big[\rho_{out} {}^2\big]&=\text{Tr}\big[\big((1-p_{err})\rho_{out,id} + p_{err} \rho_{noisy}\big)^2\big]\\
&=(1-p_{err})^2  + 2(1-p_{err})p_{err}\text{Tr}(\rho_{out,id}\rho_{noisy})+ p_{err}^2 \text{Tr}(\rho_{noisy} {}^2)\\
&\geq(1-p_{err})^2  + 2(1-p_{err})p_{err}\text{Tr}(\rho_{out,id}\rho_{noisy})+ p_{err}^2 \text{Tr}((\mathds{1}/2^{n}) {}^2)\\
&\geq(1-p_{err})^2  + p_{err}^2 \text{Tr}(\mathds{1}/2^{2n} )\\
&\geq(1-p_{err})^2  + p_{err}^2 /2^{n} \\
&\geq 1 -2 p_{err} + p_{err}^2(1  + 2^{-n})  \\
&\geq 1 -2 \gamma + \gamma^2(1  + 2^{-n}).  \\
\end{aligned}
\end{equation}
\end{widetext}
The distributivity of the trace is used to obtain the second line. 
The inequality of the third line is reached by replacing $\rho_{noisy}$ in the final term with the maximally mixed state, using that the maximally mixed state is the maximally entropic state.
To obtain the fourth line, we use the observation that
\begin{equation}
\begin{split}
\text{Tr}(\rho_{out,id}\rho_{noisy}) &= \text{Tr}\big( \ket{\psi}\bra{\psi} \sum_j \mu_j \ket{\phi_j}\bra{\phi_j} \big) \\
& = \sum_{j}  \mu_j |\bra{\psi}\ket{\phi_j}|^2 \\
& \geq 0,
\end{split}
\end{equation}
where the density matrices have been initially decomposed in their eigenbases, and the final inequality follows from the positivity of the eigenvalues and the absolute inner product.
Since $0\leq p_{\mathrm{err}}\leq\gamma\leq 1$, the final
inequality in Eq.~(H2) follows from the monotonic decrease of
$f(p)=1-2p+(1+2^{-n})p^2$ for
$p\leq 2^n/(2^n+1)$. For larger $\gamma$,
$f(\gamma)\leq 2^{-n}$, while every $n$-qubit state has purity
at least $2^{-n}$.
The second-order Rényi entropy density of the state $\rho_{out}$ is defined as
\begin{equation}
n^{-1}S^{(2)}(\rho_{out}) = - n^{-1}\log_2 (\text{Tr}[\rho_{out} {}^2]).
\end{equation}
Now, from the final inequality of eqn. \ref{renyi_ent_bound}, and using that for positive reals $x$ and $y$, if $x>y$ then $\text{log}_2(x)>\text{log}_2(y)$, it follows that
\begin{equation}
\begin{split}
\log_2 \big(\text{Tr}\big[\rho_{out} {}^2\big] \big) \geq \text{log}_2 (1 -2 \gamma + \gamma^2(1  + 2^{-n})). \\
\end{split}
\end{equation}
Therefore, the second-order Rényi entropy density of the logical output state can be upper-bounded using the experimentally estimated output value of $\gamma$, specifically
\begin{equation}
n^{-1}S^{(2)}(\rho_{out})  \leq - n^{-1}\log_2 (1 -2 \gamma + \gamma^2(1  + 2^{-n})).
\end{equation}
This bound holds with the same confidence as the logical accreditation bound.

\section{Frameworks for logical computation} \label{frameworks_logical_computation}
\vspace{-0.5em}

We now describe two frameworks for computation with different logical gate sets, and discuss the implications for logical accreditation within each framework.

These are: (1) the Clifford and $T$ gate framework, and (2) the Clifford and analog rotation gate framework. 

\vspace{-0.5em}

\subsection{\texorpdfstring{Clifford and $T$ gate logical circuits}{Clifford and $T$ gate logical circuits}}
\vspace{-0.5em}

Gate sets consisting of a subgroup of the Clifford group along with the $T$ gate are often considered in logical computation.
In the Clifford and $T$ framework, an input computation is compiled using gates from the Clifford group and $T$ gates.
The inclusion of the $T$ gate is necessary to extend the computational power of the Clifford group for universal computation.
Many CSS codes, such as the family of 2D colour codes, can perform Clifford gates transversally, with the non-Clifford $T$ gates needing to be applied through the use of magic state purification and gate teleportation.
This means the Clifford and $T$ framework can be made fully fault-tolerant, where Clifford gates may be performed using transversal operations or lattice surgery, and $T$ gates are performed using purified magic states.
In the Clifford and noisy $T$ gate framework, magic states are not purified to avoid the associated space-time overhead.
Instead, noisy $\ket{\pi/4}$ magic states are prepared,
and are directly used to perform $T$ gates.
This framework is often considered in the context of \textit{partial fault-tolerance} \cite{piveteau_error_2021}, as the Clifford gates are performed fault-tolerantly and the $T$ gates are not.

For computation using either of these frameworks to be compatible with logical accreditation, it is required that any input target computation is compiled using gates from the allowed Clifford gate set $\{\text{C}Z, H, S, I, X, Y, Z\}$ and the $T$ gate.
The Clifford gates may be implemented fault-tolerantly without the use of magic states, and the $T$ gates are implemented by the preparation of $\ket{\pi/4}$ magic states which are then used to perform the $T$ gates by gate teleportation. 
The computational logical qubits are required to be initialised in the $\ket{0}^{\otimes n}$ state, and the final logical measurement is in the computational basis. 
Such a set-up is universal for quantum computation.

\begin{figure*}[]%[hbt!]
    \centering
    \begin{minipage}[c]{0.61\textwidth}
        \centering
        \includegraphics[width=\linewidth]{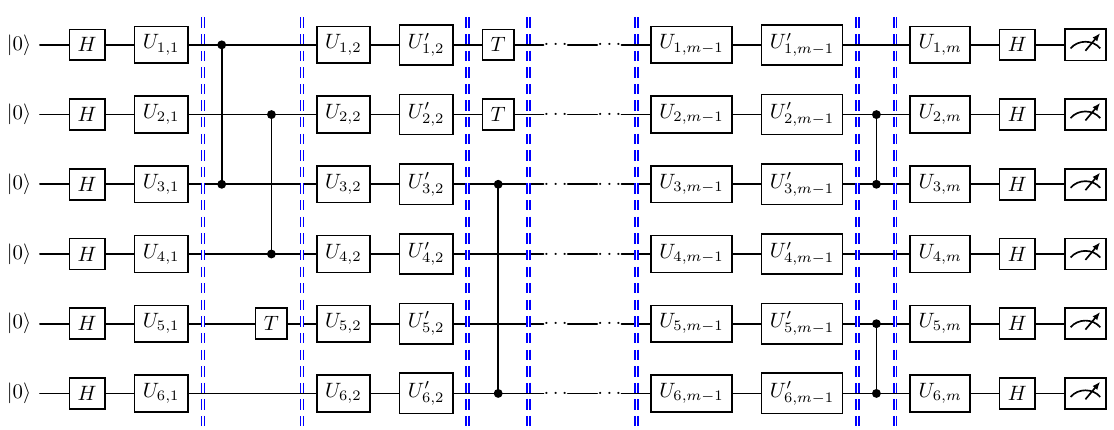}
        \subcaption{Target IQP circuit}
    \end{minipage}
    \hfill
    \begin{minipage}[c]{0.58\textwidth}
        \centering
        \includegraphics[width=\linewidth]{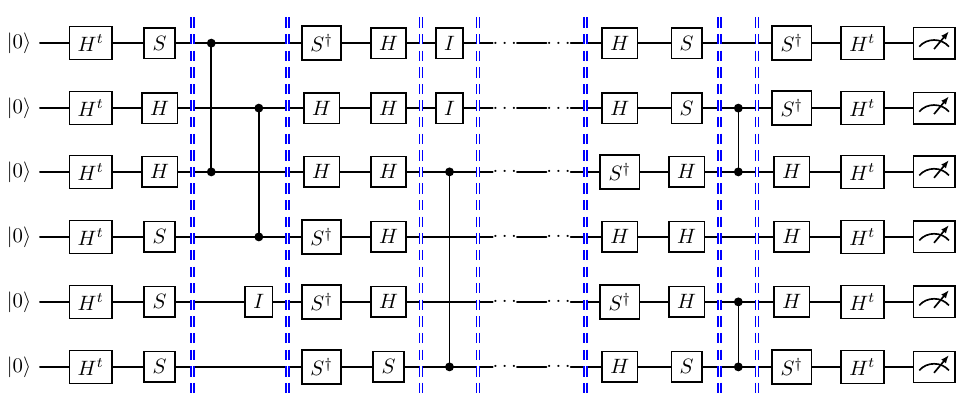}
        \subcaption{Trap IQP circuit}
    \end{minipage}
    \caption{\textit{Structure of target and trap circuits used for the IQP circuit numerical simulations.} 
    The target circuit and the trap circuit structures used for the IQP circuit numerical simulations are shown in (a) and (b), respectively.
    In the trap circuit structure, the target circuit gate layers sandwiching the $\mathcal{J}$-type gate layers are replaced by conjugating $S-S^\dagger$ gates and $H-H$ gates.
    The logic of the C$Z$ gates combined with these conjugating gates corresponds to randomly oriented $CNOT$ gates. 
    }
    \label{fig:iqp-target-trap-subfigs}
\end{figure*}

\begin{figure*}[]%[hbt!]
    \centering
    \begin{subfigure}{0.64\textwidth}
        \centering
        \includegraphics[width=\linewidth]{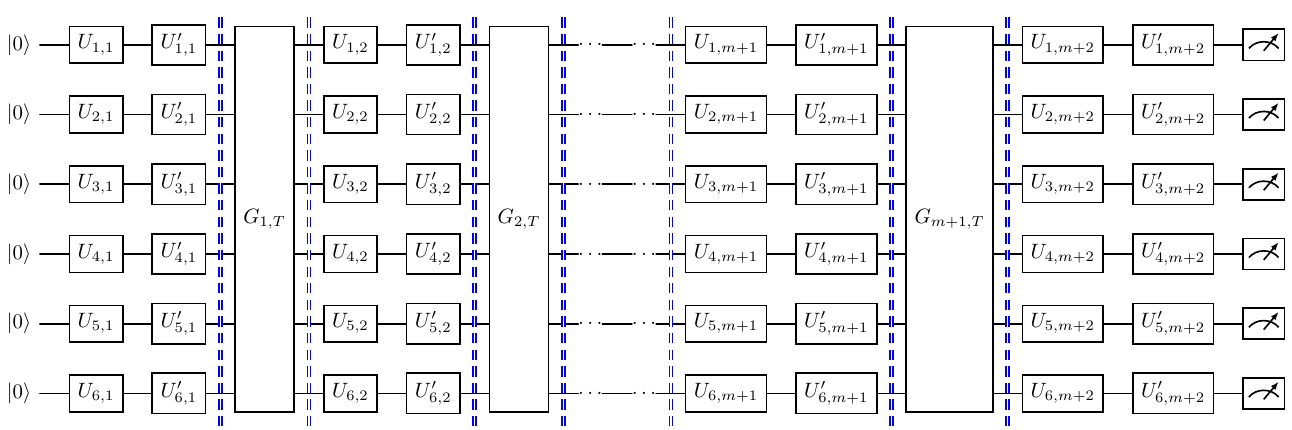}
        \caption{Target Trotter circuit}
        \label{fig:target-circuit}
    \end{subfigure}
    \hfill
    \begin{subfigure}{0.63\textwidth}
        \centering
        \includegraphics[width=\linewidth]{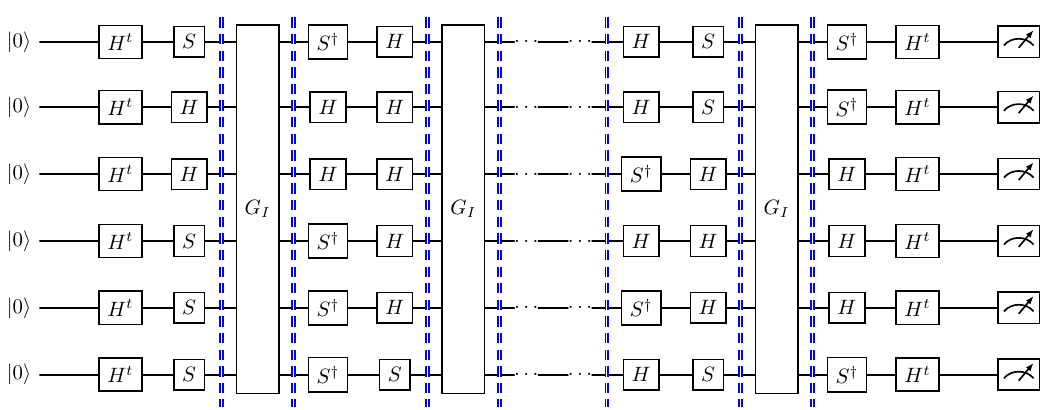} 
        \caption{Trap Trotter circuit}
        \label{fig:trap-circuit}
    \end{subfigure}
    \caption{\textit{Structure of target and trap circuits used for the Trotterised circuit numerical simulations.} 
    The target circuit and the trap circuit structures used for the Trotterised circuit numerical simulations are shown in (a) and (b), respectively.
    The gate layers denoted $G_{i,T}$ and $G_I$ are the Trotterised logical circuit operations used for the target and trap circuits, respectively.
    In the trap circuit structure, the target circuit gate layers sandwiching the $\mathcal{J}$-type gate layers are replaced by conjugating $S-S^\dagger$ gates and $H-H$ gates.}
    
    \label{fig:simulations-trotterized-circuit-structure}
\end{figure*}

\vspace{-1em}
\subsection{\texorpdfstring{Clifford and analog $Z$-rotation gate logical circuits}{Clifford and analog Z-rotation gate logical circuits}}
\vspace{-0.6em}

Another framework suited to partially fault-tolerant computation uses a gate set of Clifford operations and analog rotation operations. 
For computation using this framework to be compatible with logical accreditation, it is required that any input target computation is compiled using gates from the allowed Clifford gate set $\{\text{C}Z, H, S, I, X, Y, Z\}$ and gates of the form $R_P(\theta)$, where $P \in \{I,X,Y,Z\}^{\otimes n}$, and $n$ denotes the number of logical qubits involved in the rotation.
The Clifford gates may be implemented fault-tolerantly without the use of magic states, and the $R_Z(\theta)$ gates by the preparation of noisy $\ket{\theta}$ magic states, which are then used to perform the analog rotations by projective measurement. 
Logical qubits are initialised in the $\ket{0}^{\otimes n}$ state and the final logical measurement is in the computational basis.
Again this set-up is universal for quantum computation.

One example of a Clifford and analog rotation framework for logical computation is a scheme for partially fault-tolerant logical computation called the space-time efficient analog rotation (STAR) framework \cite{akahoshi_partially_2024, toshio_practical_2025}.
This allows for fault-tolerant Clifford operations and noisy arbitrary Pauli rotations using the surface code. 
These analog rotations can be local Pauli-$Z$ rotations, shown in Fig. \ref{fig:STAR_teleport} (a), or global logical Pauli rotations as shown in Fig. \ref{fig:STAR_teleport} (b).

\section{Certifying Trotterised quantum Hamiltonian simulation circuits}

Quantum states evolve in time according to the Schr\"{o}dinger equation: $i\frac{\partial}{\partial t}\ket{\psi(t)}=H(t) \ket{\psi(t)}$, where $\hbar$ is assumed to be 1, $H(t)$ is a (possibly time-dependent) Hamiltonian, and $\ket{\psi(t)}$ is a quantum state. 
If the Hamiltonian is time-independent, 
 the Schr\"{o}dinger equation can be solved exactly: $\ket{\psi(t)} = e^{-iHt}\ket{\psi(0)}$.
Trotterisation uses Trotter–Suzuki product formulas to approximate the time-evolution operator, $e^{-iHt}$, as a product of exponentials of simpler terms.
A Hamiltonian, $H$, that has been mapped to qubits, by, for example, the Jordan-Wigner transformation \cite{jordan_uber_1928, somma_simulating_2002}, may be expressed as a linear combination of Pauli operators
\begin{equation}
H = \sum_{i=1}^{L} a_i P_i,
\end{equation}
where $P_i \in \{I,X,Y,Z\}^{\otimes n}$ for $n$ computational logical qubits.
The Trotter-Suzuki decomposition allows the time evolution operator $e^{-iHt}$ to be approximately decomposed into a sequence of $N$ discrete time-steps.
The time evolution operator may be written as a product of $N$ terms 
\begin{equation}
e^{-iHt} = \Big(e^{-i(\sum_{i=1}^{L} a_i P_i)t/N}\Big)^N,
\end{equation}
where the evolution operator within the brackets is applied $N$ times and each evolution step is by time $t/N$.
The second-order Trotter-Suzuki decomposition approximates these time-step terms as
\begin{equation}
\Big(e^{-i(\sum_{i=1}^{L} a_i P_i)t/N}\Big)^N \approx \bigg(\prod_{i=1}^{L}e^{-i (\frac{a_it}{2N})P_i} \prod_{i=L}^{1} e^{-i (\frac{a_it}{2N})P_i} \bigg)^N.
\end{equation}
The approximation error of the Trotterised evolution is
\begin{equation}
\norm{e^{-iHt} - \bigg(\prod_{i=1}^{L}e^{-i (\frac{a_it}{2N})P_i} \prod_{i=L}^{1} e^{-i (\frac{a_it}{2N})P_i} \bigg)^N} \leq \mathcal{O}\Big(\frac{t^3}{N^2}\Big),
\end{equation}
where $\norm{.}$ is the operator norm, also known as the spectral norm. Specifically, the upper bound is $Wt^3/N^2$ where $W$ is the Trotter error norm - a constant depending on the Hamiltonian and the type of Trotterisation used \cite{childs_theory_2021, campbell_early_2022}.
As the Hamiltonian is expressed as a linear combination of Pauli operators, the previous decomposition may be written as a sequence of multi-qubit Pauli rotations:
\begin{equation} 
\begin{split}
\bigg(\prod_{i=1}^{L} e^{-i (\frac{a_it}{2N})P_i} &\prod_{i=L}^{1} e^{-i (\frac{a_it}{2N})P_i} \bigg)^N \\
&= \bigg(\prod_{i=1}^{L} R_{P_i}\Big( \frac{a_it}{N}\Big) \prod_{i=L}^{1} R_{P_i}\Big( \frac{a_it}{N}\Big) \bigg)^N.
\end{split}
\end{equation}
To perform this sequence of operations with a logical circuit, each Pauli rotation gate is performed by the consumption of a single magic state.
As the logical accreditation protocol permits target circuits including multi-qubit Pauli rotation operations (the construction required is described in Appendix \ref{trap_circuits_section}) it can be used to provide an upper bound on the total circuit error of the circuit consisting of a sequence of multi-qubit Pauli operations like eqn. \ref{trotter_circuit}, and also an upper bound on the TVD of the output when measuring any observable of the time-evolved state.
Logical accreditation may therefore be used to certify the Trotterised evolution of a Hamiltonian using a fault-tolerant logical circuit.

\section{IQP circuit numerical simulations}

The structure of IQP circuits used for the simulations is shown in Fig.~\ref{fig:iqp-target-trap-subfigs}, and is similar to the circuits considered in \cite{ferracin_experimental_2021}. 
For the target circuits, layers of single-qubit Clifford phase gates $U_{i,j}$ were alternated with layers composed of C$Z$ gates and $T$ gates. 
In the experiments, the IQP trap circuits were constructed as follows:
\begin{itemize}
    \item Every $T$ gate was replaced with an identity operation ($I$).
    \item Each $CZ$ gate was sandwiched with $S, S^\dagger$, and $H$ gates, as shown in Fig.~\ref{fig:CZ_rand}.
    \item Random $S-S^\dagger$ or $H-H$ gate pairs were applied to qubits not involved in $CZ$ gates.
    \item For each trap circuit, it was randomly chosen whether to prepend and append Hadamard gate layers (denoted by $H^t$ in Fig.~\ref{fig:iqp-target-trap-subfigs}) to the circuit.
\end{itemize}

In Fig.~\ref{fig:tvd-tcr-infidelity}, we compare the total circuit error rate, the TVD upper bound and the logical state infidelity for 5-qubit, 40-layer IQP circuits, using a similar setting as throughout (e.g., distance-11 surface codes for the fault-tolerant configurations). 
Interestingly, the results suggest that the TVD upper bound not only upper bounds the total circuit error rate, but also upper bounds the infidelity of the logical output state. 
We show that this upper bound holds analytically in Appendix~\ref{sec:upper-bound-infidelity}.

\section{Upper bound for the infidelity between the experimental and ideal output states}
\label{sec:upper-bound-infidelity} 
\begin{figure*}
    \centering
    \begin{subfigure}[t]{0.41\textwidth}
        \includegraphics[width=\linewidth]{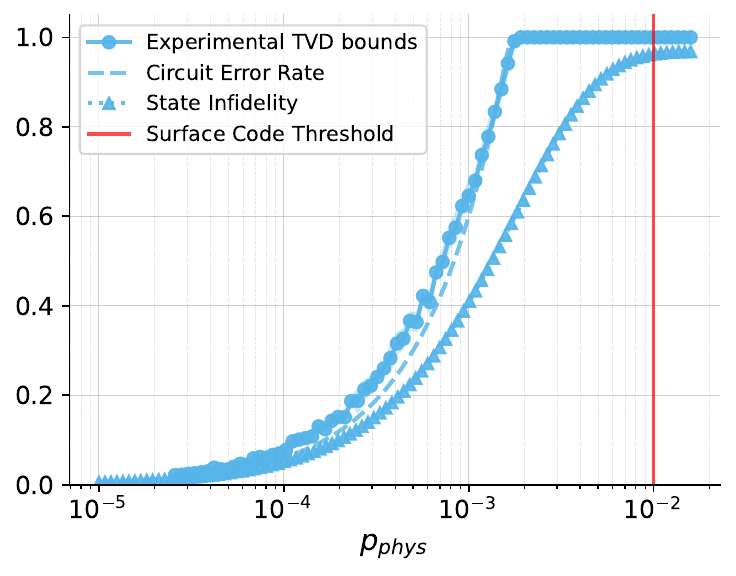}
        \caption{NISQ}
        \label{fig:infidelity-tvd-tcr-nisq}
    \end{subfigure}
    \hfill
    \begin{subfigure}[t]{0.41\textwidth}
        \includegraphics[width=\linewidth]{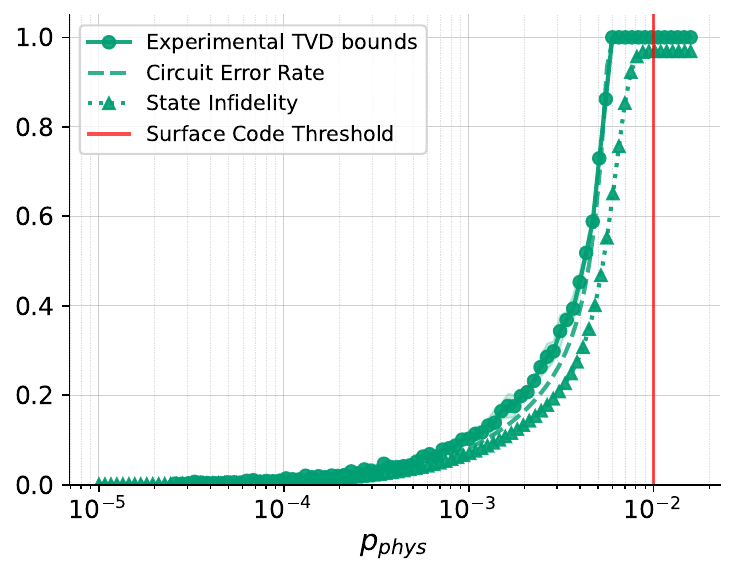}
        \caption{PFTQC}
        \label{fig:infidelity-tvd-tcr-pftqc}
    \end{subfigure}
    \hfill
    \begin{subfigure}[t]{0.41\textwidth}
        \includegraphics[width=\linewidth]{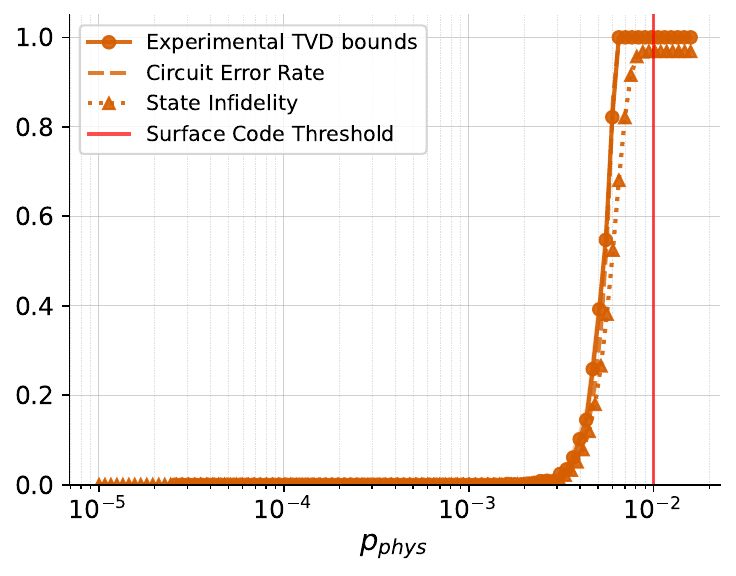}
        \caption{FTQC}
        \label{fig:infidelity-tvd-tcr-ftqc}
    \end{subfigure}
    \caption{\emph{Comparison of the logical accreditation TVD bound, total circuit error rate, and logical state infidelity for the NISQ, PFTQC and FTQC regimes.} 
    This data was obtained from experiments involving 5-qubit IQP circuits with 40 gate layers.
    The experimentally derived TVD bounds, the circuit error rate, the logical output state infidelity, and the surface code threshold are plotted as functions of the physical error rate $p_{phys}$ in the range $[10^{-5},10^{-1}]$.
    The results for NISQ circuits are plotted in (a), the results for PFTQC circuits are plotted in (b), and the results for FTQC circuits are plotted in (c).
    These results agree with the analytical bound derived in Appendix \ref{sec:upper-bound-infidelity}, showing that the logical output state infidelity is upper-bounded by the $\gamma$ value provided by logical accreditation.
    }
    \label{fig:tvd-tcr-infidelity}
\end{figure*}

The fidelity between the ideal logical target circuit output state,
$\rho_{\mathrm{out,id}}$, and the experimental logical target circuit
output state, $\rho_{\mathrm{out}}$, is defined by
\begin{equation}
F(\rho_{out,id}, \rho_{out}) =\bigg(\text{Tr}\sqrt{\sqrt{\rho_{out,id}}\rho_{out}\sqrt{\rho_{out,id}}} \bigg)^2.
\end{equation}
Since $\rho_{out,id}$ is a pure state, we have that $\rho_{out,id} = \ket{\psi}\bra{\psi}$.
We know that projection operators are idempotent, so $(\ket{\psi}\bra{\psi})^2=\ket{\psi}\bra{\psi}$.
Therefore, $\ket{\psi}\bra{\psi}=\sqrt{\ket{\psi}\bra{\psi}}$.
So the expression for the fidelity then becomes
\begin{equation}
\begin{split} \bigg(\text{Tr}\sqrt{\ket{\psi}\bra{\psi}\rho_{out}\ket{\psi}\bra{\psi}}  \bigg)^2 &= \bra{\psi}\rho_{out}\ket{\psi}\Big(\text{Tr}\sqrt{\ket{\psi}\bra{\psi}}  \Big)^2 \\
    &= \bra{\psi}\rho_{out}\ket{\psi}\\
    &= \text{Tr}\big(\ket{\psi}\bra{\psi}\rho_{out}\big).
\end{split}
\end{equation}
The experimental output state $\rho_{out}$ is of the form
\begin{equation}
\rho_{out}=(1-p_{err})\rho_{out,id} + p_{err} \rho_{noisy}.
\end{equation}
The fidelity of the experimental output state and the ideal output state is, therefore,
\begin{equation}
\begin{split}
\text{Tr}\Big(&\rho_{out,id}\big( (1-p_{err})\rho_{out,id} + p_{err} \rho_{noisy}\big)\Big)\\
&=\text{Tr}\big((1-p_{err})\rho_{out,id}{}^2 + p_{err} \rho_{noisy}\rho_{out,id}\big)\\
&=(1-p_{err})+ p_{err}\text{Tr}\big( \rho_{noisy}\rho_{out,id}\big).\\
\end{split}
\end{equation}
The infidelity is the complement of the fidelity, so the infidelity between $\rho_{out,id}$ and $\rho_{out}$ is 
\begin{equation}
\begin{split}
1 - &\text{Tr}\Big(\rho_{out,id}\big( (1-p_{err})\rho_{out,id} + p_{err} \rho_{noisy}\big)\Big).\\
\end{split}
\end{equation}
This may be upper-bounded
\begin{equation}
\begin{split}
1 - &\text{Tr}\Big(\rho_{out,id}\big( (1-p_{err})\rho_{out,id} + p_{err} \rho_{noisy}\big)\Big)\\
&=1 - (1-p_{err}) - p_{err}\text{Tr}\big( \rho_{noisy}\rho_{out,id}\big)\\
&= p_{err} \Big(1 -\text{Tr}\big( \rho_{noisy}\rho_{out,id}\big) \Big)\\
& \leq p_{err} \\
% & \leq 2p_{inc} \\
&\leq \gamma .\\
\end{split}
\end{equation}
Here we have used the inequality 
\begin{equation}
\begin{split}
\text{Tr}(\rho_{out,id}\rho_{noisy}) &= \text{Tr}\big( \ket{\psi}\bra{\psi} \sum_j \mu_j \ket{\phi_j}\bra{\phi_j} \big) \\
& = \sum_{j}  \mu_j |\bra{\psi}\ket{\phi_j}|^2 \\
& \geq 0,
\end{split}
\end{equation}
 along with the upper bound previously derived for the total logical circuit error rate.

\begin{figure*}[t]
    \centering
    \includegraphics[width=0.54\textwidth]{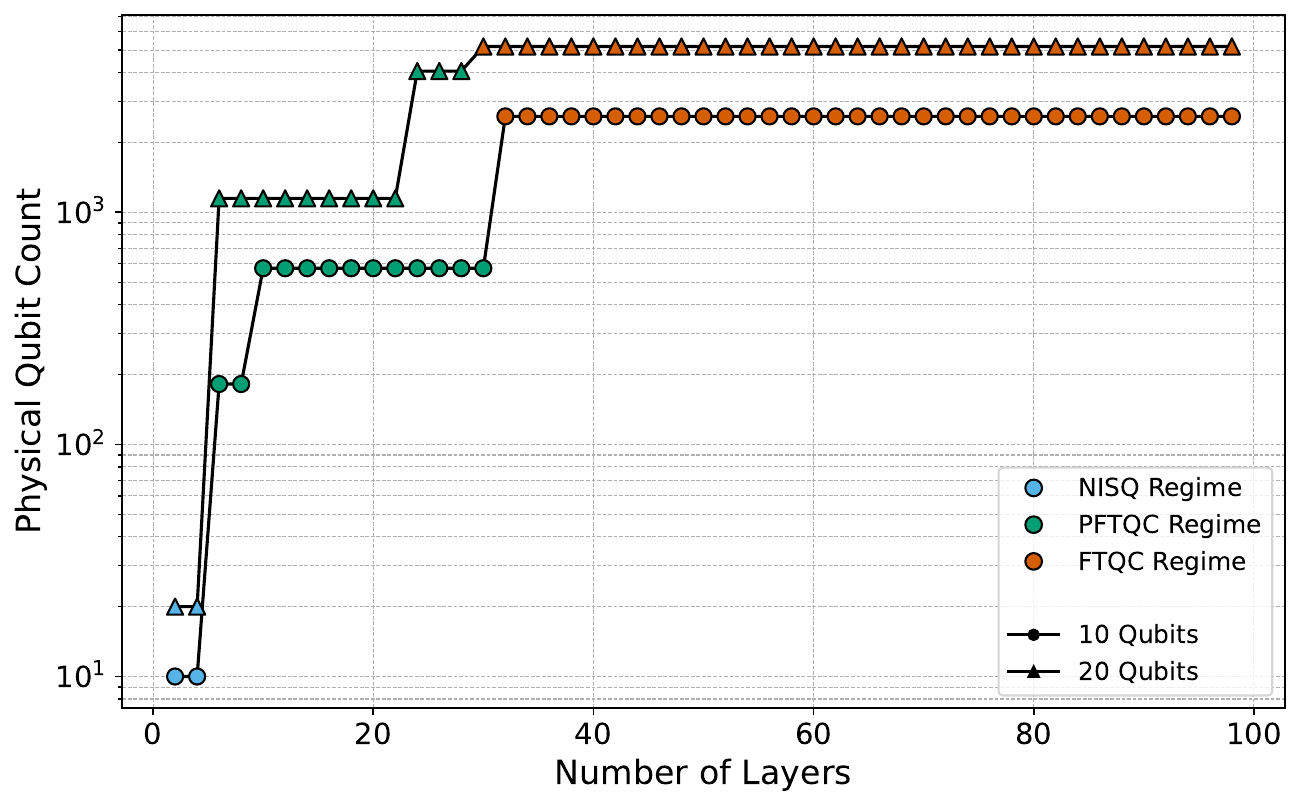}
    \caption{\emph{Numerical analysis of the minimum number of physical qubits required to achieve the IQP threshold for quantum advantage.} 
    The plot shows the most resource-efficient option among NISQ, PFTQC, and FTQC for achieving the quantum advantage threshold of \cite{bremner_average-case_2016}, as circuit depth is increased up to 100 gate layers.
    The physical error rate is assumed to be $10^{-5}$. 
    The colour of each data point indicates the most resource-efficient regime to achieve the threshold: NISQ (blue), PFTQC (green), or FTQC (orange). 
    The physical qubit counts for PFTQC and FTQC are determined by the minimum surface code distance (and hence the minimum number of physical qubits) required such that the advantage bound is satisfied.}
    \label{fig:optimal_resource_comparison}
\end{figure*}

\section{Physical Resource Analysis for IQP Advantage}
\label{apx:resource-analysis-for-IQP-advantage}

We also used logical accreditation to carry out a resource analysis for an IQP circuit sampling quantum advantage with logical qubits encoded using the surface code. 
We estimate the physical qubit performance requirements needed to ensure that the TVD of an IQP circuit remains below the quantum advantage threshold provided in \cite{bremner_average-case_2016}.

In the numerical simulations, the physical error rate was set to $p_{\text{phys}} = 10^{-5}$ for 10 and 20 logical qubit IQP circuits. 
The results of these experiments are plotted in Fig. \ref{fig:optimal_resource_comparison}. 
For increasing circuit depths, we find the minimum code distance required in the NISQ, PFTQC, and FTQC regimes to keep the TVD below the quantum advantage threshold. 
If a regime fails to satisfy the threshold at a given depth, no deeper circuits are attempted beyond that point. 
The optimal strategy for an IQP circuit with a given number of layers is then simply the viable regime with the lowest physical qubit cost.

The results show resource-efficiency crossovers between the regimes. 
NISQ is only optimal for very shallow circuits, after which the PFTQC scheme has a lower resource cost. 
With increasing circuit depth, the accumulation of errors from PFTQC's unprotected $T$ gates forces its required code distance to become prohibitively large. 
At this point, the FTQC scheme, despite the overhead of magic state distillation, becomes the most cost-efficient option.

\end{document}